\documentclass[superscriptaddress, 10pt, amsmath,amssymb, aps, pra, twocolumn,preprintnumbers]{revtex4-2}

\usepackage[utf8]{inputenc}
\usepackage[T1]{fontenc}
\usepackage{amsmath}
\usepackage{amsfonts}
\usepackage{amssymb}
\usepackage{amsthm}
\usepackage{dsfont}
\usepackage{graphicx}
\usepackage[dvipsnames]{xcolor}
\usepackage{enumitem}
\usepackage{hyperref}
\hypersetup{hidelinks}
\usepackage{physics}
\usepackage{bbm}
\usepackage{indentfirst}
\usepackage{thmtools}
\usepackage{thm-restate}
\usepackage{soul}
\usepackage{comment}
\setstcolor{red}
\usepackage{breakurl}
\usepackage{mathtools}
\usepackage{titletoc}
\usepackage{chngcntr}
\usepackage[normalem]{ulem}

\usepackage{titletoc}

\titlecontents{section}
  [1.5em]
  {}
  {\contentslabel{2em}}
  {}
  {\hfill\contentspage}

\titlecontents{subsection}
  [3.5em]
  {}
  {\contentslabel{2.5em}}
  {}
  {\hfill\contentspage}

\usepackage{bibunits}
\defaultbibliographystyle{apsrev4-2_Truncate}
\defaultbibliography{refs}

\usepackage{algorithm,algorithmic}

\usepackage{setspace}

\newtheorem{theorem}{Theorem}
\newtheorem{remark}{Remark}
\newtheorem{definition}{Definition}
\newtheorem{corollary}{Corollary}[theorem]
\newtheorem{lemma}{Lemma}

\makeatletter
\renewcommand\footnoterule{%
  \kern -3pt
  \hrule width 1.6in height 0.4pt
  \kern 3pt
}
\makeatother

\newcommand{\DN}{\mathcal{D}^{\otimes n}_\lambda}
\newcommand{\DNINV}{\mathcal{I}^{\otimes n}_{\lambda'}}
\newcommand{\DNT}{\mathcal{D}^{\otimes 3n}_\lambda}
\newcommand{\swap}{\textnormal{\textsf{SWAP}}}

\newcommand{\nbqp}{{\textnormal{\textsf{NBQP}}}}

\newcommand{\bqp}{{\textnormal{\textsf{BQP}}}}

\makeatletter
\newcommand{\bigboxplus}{\mathop{\vphantom{\sum}\mathchoice
  {\vcenter{\hbox{\Large$\m@th\boxplus$}}}
  {\vcenter{\hbox{\large$\m@th\boxplus$}}}
  {\vcenter{\hbox{\normalsize$\m@th\boxplus$}}}
  {\vcenter{\hbox{\scriptsize$\m@th\boxplus$}}}
}\displaylimits}
\makeatother

\begin{document}
\addtocontents{toc}{\protect\setcounter{tocdepth}{-1}} 

\title{Fault-tolerant quantum processing of physical experiments}
\author{Ishaan Kannan}
\thanks{Equal contribution}
\affiliation{Harvard Quantum Initiative, 60 Oxford St, Cambridge, MA 02138}
\author{Harald Putterman}
\thanks{Equal contribution}
\affiliation{Department of Physics, Harvard University, Cambridge, MA 02138 USA}
\author{Jordan Cotler}
\thanks{Corresponding author}
\email{jcotler@fas.harvard.edu}
\affiliation{Harvard Quantum Initiative, 60 Oxford St, Cambridge, MA 02138}
\affiliation{Department of Physics, Harvard University, Cambridge, MA 02138 USA}

\begin{abstract}
Quantum computers may reveal features of Nature inaccessible to conventional experiments, but manipulating raw quantum data introduces noise that degrades inference even when the processor is fault-tolerant, creating a data-input bottleneck for robust quantum learning. Here we show that quantum fault tolerance can substantially improve the sample complexity of learning from noisy experiments. We encode unknown quantum states from physical experiments into protected quantum memory, enabling fault-tolerant implementations of quantum learning algorithms otherwise degraded by errors. Using this \textit{quantum uploading} procedure, we prove that noisy randomized measurement and multi-copy learning tasks can be performed exponentially faster than by any adaptive strategy that does not immediately encode physical states into error-corrected memory. These separations are not simply due to a reduced effective noise rate: they hold even when uploading is substantially noisier than the bare experimental interface, rigorously establishing immediate encoding as the optimal approach to noise-robust learning. We numerically illustrate the speedups in astronomical imaging, where quantum processing of uploaded photons locates an exoplanet obscured by a bright star using orders of magnitude fewer shots than unencoded baselines. Our results establish a robust interface between quantum computers and natural systems, enabling powerful and practical quantum-enhanced experiments.
\end{abstract}

\maketitle

\begin{bibunit}
\setcounter{tocdepth}{-1} 

\section{Introduction}
\vspace{-2mm}

Quantum information processing offers new ways to learn from physical experiments, motivating recent progress in quantum metrology, learning, and computation. Quantum learning theory assumes that an experiment supplies access to an unknown quantum system, and has demonstrated that one can achieve dramatic advantages in time and sample complexity over conventional experimental methods by directly processing the system on a quantum computer \cite{Cotler_2019, chen2021exponentialseparationslearningquantum, Aharonov_2022,  Huang_adv_2022, cotler2026quantumadvantagesensingproperties,kannan2026exponentialquantumadvantagelearning}. 

However, realizing these advantages in physical experiments, especially in the presence of noise, remains difficult \cite{chen2023complexity, cotler2026noisyquantumlearningtheory, depradenne2026restrictionsnoncliffordfaulttolerance}. In astronomy, condensed-matter physics, chemistry, and analog many-body experiments, for instance, the goal is to learn properties of a bare quantum state or process produced or detected in the laboratory.  In all cases, the unknown quantum system is supplied by Nature, and does not arrive inside an error-correcting code on quantum hardware. As such, any coherent learning protocol must shuttle information between protected logical degrees of freedom on the processor and unprotected physical degrees of freedom in the experiment. Recent work formalized the cost of coupling fault-tolerant processors to unprotected experiments, showing that once uncorrected physical noise accompanies every logical operation on the experimental state, the multi-copy and randomized-measurement primitives behind many known speedups can become exponentially less informative, even when the processor itself is fault-tolerant \cite{cotler2026noisyquantumlearningtheory}. The challenge of interfacing quantum data with fault-tolerant processors thus presents a core obstruction to using quantum computers to learn from rich quantum data in Nature.

In this work, we eliminate this bottleneck by porting techniques from quantum fault tolerance to quantum learning theory and presenting a procedure to lift unknown quantum data into encoded quantum memory. Operationally, the procedure has two steps. First, an experiment-dependent \textit{transduction} maps the physical state into unencoded quantum memory. Second, utilizing techniques for universal quantum computation via magic-state injection~\cite{Fowler_2012,Horsman_2012,Li2015MagicStateFidelity, Gidney2023CleanerMagicStates}, the memory state is \textit{injected} into an error-correcting code, where it can then be processed fault-tolerantly.
The code can be grown to an arbitrarily large distance at constant error overhead. We refer to the combined procedure as \textit{quantum uploading}. We prove that uploading-enabled fault-tolerant algorithms can then learn from noisy quantum experiments using exponentially fewer measurements than any protocol that does not embed the experimental state into an error-correcting code, demonstrating the utility of quantum fault tolerance in quantum learning theory. 

\begin{figure*}[ht!]
\includegraphics[width=\textwidth]{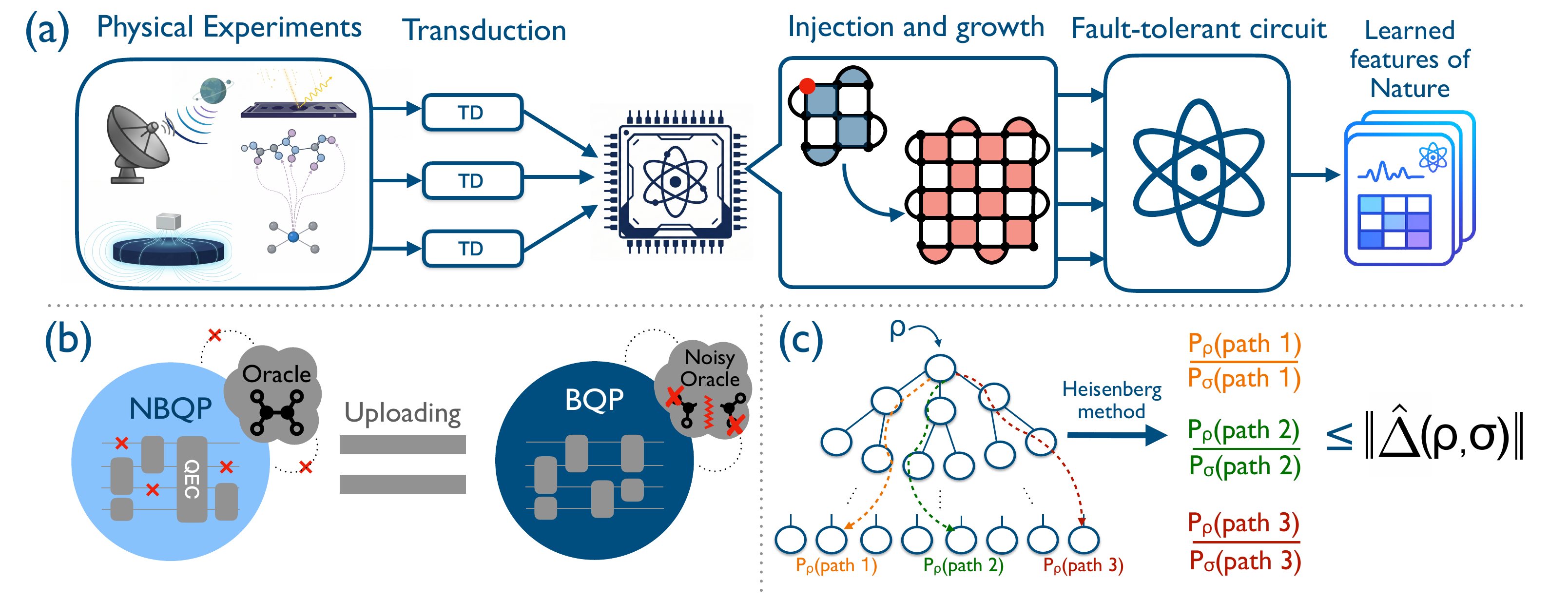}
 \caption{\textbf{Fault-tolerant processing of noisy experiments.} (a) \textit{Schematic of quantum uploading}. Quantum data is collected from physical experiments and transduced onto a quantum processor, which injects the data into a high-distance code patch. After a one-time noise cost, this is processed fault-tolerantly, accurately revealing features of Nature. (b) \textit{Complexity-theoretic implications of uploading}. While the class $\textsf{NBQP}^O$ is superpolynomially separated from $\textsf{BQP}^O$, the existence of uploading implies that a single layer of noise composed with any $\textsf{BQP}$ state oracle makes the classes equivalent. (c) \textit{Heisenberg learning tree method}. Each path along a learning tree is associated with an outcome probability, and hypothesis-testing success in distinguishing two states $\rho$ and $\sigma$ is controlled by the likelihood ratio over typical paths. While controlling this ratio directly is challenging, our Heisenberg method bounds all path ratios by the norm of a fixed distinguishability operator $\hat{\Delta}$; see Supplementary \ref{app:Heisenberg_learning_tree}.}
\label{fig:schematics}
\vspace{-5mm}
\end{figure*}

Uploading changes the learning problem in a simple way. Without uploading, every coherent interaction with the experimental state introduces uncorrected physical noise; with uploading, this exposure is compressed into a one-time effective input channel, while the rest of the learning protocol proceeds fault-tolerantly. We show that this architectural change yields exponential sample-complexity speedups for two well-studied classes of quantum learning primitives, randomized preprocessing and multi-replica measurements, as exemplified by classical shadow tomography~\cite{Huang_2020} and by estimation of third-order moments or cubic observables of general quantum states. These separations hold against any adaptive strategy that does not upload the state, and they are not simply the result of a reduced effective noise rate: the speedups persist even when transduction and injection incur a substantially larger noise rate than the bare experimental interface. Our lower bounds thus circumscribe any error-mitigation strategy that forgoes uploading, and establish that the optimal approach to quantum learning in noisy settings is to immediately upload quantum data rather than to tailor the learning protocol to the noise at hand. We ground these results in an astronomical imaging application based on quantum processing of individual photons~\cite{mokeev2025enhancingopticalimagingquantum}, where uploading requires orders of magnitude fewer shots than unencoded baselines to locate an exoplanet near a bright star.

Information-theoretic lower bounds for quantum learning problems typically use the learning-tree formalism \cite{chen2021exponentialseparationslearningquantum, weiyuan_Paulis, cotler2026noisyquantumlearningtheory}, which has yielded exponential lower bounds when the number of copies of the unknown state is the limited resource. In our setting, copies are not scarce; what is scarce is the ability to process the state coherently before it is uploaded into a code, since any such processing is subject to noise. We therefore introduce the \textit{Heisenberg learning tree method}, which establishes tight sample-complexity lower bounds when the limited resource is noise or another experimentally motivated constraint rather than copy count. This method reduces the challenging task of proving lower bounds for noisy multi-replica protocols to straightforward computation of the norm of a task-dependent distinguishability operator in the Heisenberg picture. This formalism may prove useful as quantum learning and metrology begin to integrate practical constraints into formal separations.

\vspace{3mm}
\vspace{-2em}
\section{Quantum Uploading}
\vspace{-2mm}
Quantum uploading proceeds in two stages: an experiment-dependent transduction maps the physical state into the processor's quantum memory, and a generic injection places the memory state into an error-correcting code, after which all subsequent processing proceeds fault-tolerantly. Injection can be analyzed independently of the experiment at hand. In this section, we provide an explicit surface-code injection gadget, and use the formalism developed in \cite{He2025ComposableFaultTolerance} to prove that it couples to any fault-tolerant circuit with logical error probability exponentially suppressed in distance.

\vspace{-1mm}
\vspace{-1em}
\subsection{Surface-Code State Injection}
\vspace{-2mm}

A useful uploading gadget must satisfy two requirements. First, it must be independent of the input state, since the state supplied by the experiment is unknown. Second, its error must remain a one-time input error rather than growing with the final code distance or with the depth of the subsequent learning circuit. Indeed, speedups from uploading are non-trivial since generic encoding circuits can accumulate noise with depth, adaptive and shallow protocols can mitigate noise directly \cite{Hu_2025}, and even shallow unprotected noise can erase ideal quantum-learning advantages~\cite{cotler2026noisyquantumlearningtheory}. We now give a surface-code gadget that embeds any single-qubit state into a code patch at constant input-error cost, with residual logical failure decreasing with the code distance.

\begin{theorem}[Quantum uploading via surface-code growth, informal]
\label{thm:uploading_informal}
For every code distance $d$, there is a depth-$O(d)$ surface-code uploading procedure $\textnormal{\textsf{Upload}}_d$ that maps any unknown $n$-qubit state $\rho$ into $n$ distance-$d$ logical patches. Under circuit-level noise below threshold, uploading is equivalent, up to logical error exponentially small in $d$ and other physical deviations which do not result in logical errors of subsequent gates, to first applying a one-time local depolarizing channel of strength $\lambda = \Theta(p)$ to the input and then encoding:
\begin{align}
    \rho \longmapsto
    V_d^{\otimes n}\,
\mathcal{D}_{\lambda}^{\otimes n}(\rho)\,
    (V_d^\dagger)^{\otimes n}.
\end{align}
Moreover, any subsequent learning circuit can be run fault-tolerantly, without additional unprotected noise accumulating beyond this one-time input noise.
\end{theorem}

We work with the surface code \cite{DennisKitaevLandahlPreskill2002TopologicalMemory, Horsman_2012, Fowler_2012, Bombin2007OptimalResources} because it is topological and hence admits practically natural circuits for growing the lattice. A depth-$O(d^2)$ construction of an appropriate injection gadget was given in \cite{He2025ComposableFaultTolerance}; our construction instead grows the patch from an initial constant distance to any distance $d$ in a single step~\cite{Li2015MagicStateFidelity, Lodyga2015SimpleSchemeTopologicalCodes}, making it quadratically more depth-efficient than existing constructions with provable guarantees. We obtain new rigorous guarantees for this practically-studied injection scheme by introducing a notion of a spacetime-dependent code distance to bound the error probability during our growth procedure, which may be of independent interest. 

Related guarantees have also been studied for quantum computation more generally: Ref.~\cite{Christandl2024QuantumIO} demonstrates a version of Theorem \ref{thm:uploading_informal} for concatenated-code fault tolerance schemes that simulate any noisy circuit up to input and output error, while Ref.~\cite{He2025ComposableFaultTolerance} instantiates a framework, used in our proof of Theorem \ref{thm:uploading_informal}, for obtaining threshold guarantees for circuits composed of common error correction subroutines such as magic state distillation or injection. These works provide constructions with sufficient guarantees to realize quantum uploading with asymptotically constant error overhead. Our primary contribution in this work is to demonstrate the utility of fault-tolerant uploading for quantum learning theory; it is to this end that we present an explicit and improved surface-code gadget for the uploading stage in this section, and interface it with our quantum learning algorithms of interest. It is further valuable for quantum learning applications to design improved injection procedures using high-rate codes to optimize the resource overhead of quantum uploading.

\vspace{-1em}
\subsection{Uploading and noisy quantum learning theory}
\vspace{-2mm}

Applied to quantum data resulting from physical experiments, Theorem~\ref{thm:uploading_informal} compresses the experimental interface into an effective input-noise channel, after which the remaining computation can be carried out fault-tolerantly. For learning problems, one may therefore analyze an ideal logical protocol acting on an encoded state that has suffered a one-time effective input error.

This observation clarifies the role of noise in quantum learning theory. Ref.~\cite{cotler2026noisyquantumlearningtheory} introduces a complexity-theoretic model of noisy quantum learning, in which the class $\nbqp_\lambda$ captures quantum computation where the processor is fault-tolerant but remains exposed to constant-strength physical depolarizing noise of strength $\lambda$ below the threshold of known error-correcting codes. The class $\nbqp$ consists of all problems in $\nbqp_\lambda$ for any constant $\lambda >0$. Although $\nbqp = \bqp$ by virtue of quantum error correction, for learning tasks whose data is accessed through an unprotected oracle $O$, Ref.~\cite{cotler2026noisyquantumlearningtheory} proves the oracle separation $\nbqp^O \subsetneq \bqp^O$, establishing an exponential gap between idealized noiseless learning and physically feasible learning with fault-tolerant devices. Uploading sharpens the relationship between $\nbqp$ and $\bqp$ by converting the noisy experimental interface into a one-time logical input channel.
\vspace{-1mm}
\begin{corollary}[Uploaded oracle reduction, informal]
\label{cor:uploaded_oracle_reduction}
Let $O$ be any $n$-qubit state-preparation oracle, and let $\mathcal{D}_\lambda$ denote the single-qubit depolarizing channel for $\lambda > 0$. Define $\mathcal{N}_\lambda(O) := \mathcal{D}_\lambda^{\otimes n} \circ O$. There is a $\lambda^*$ such that
\begin{equation}
\nbqp_\lambda^{O} = \bqp^{\mathcal{N}_{\lambda^*}(O)}.
\end{equation}
\end{corollary}
\noindent The gap between idealized quantum learning and learning in the real world thus collapses to a single, unavoidable layer of noise at the input. A learner that uploads once and then proceeds fault-tolerantly pays this noise cost only at the input stage, whereas a learner that never uploads must contend with physical noise throughout the protocol. While Corollary \ref{cor:uploaded_oracle_reduction} pertains to state-preparation oracles, the equivalence between relativized complexity classes can also be established for channel oracles by defining a ``quantum downloading'' gadget that, in tandem with quantum uploading, allows for coherent oracle interleaving. In the state oracle setting, we next show that uploading and direct physical processing can differ exponentially in sample complexity.

\vspace{-1em}
\vspace{-2mm}
\section{Speedups for Noisy Experiments}
\vspace{-2mm}
We now turn to our main results, which establish the utility of quantum uploading for quantum learning. We compare two modes of experimentation. In the first, the learner may use a fault-tolerant processor to act directly on the noisy physical states supplied by the experiment, with full knowledge of the noise channel to optimize their protocol; we denote the resulting optimal sample complexity by $N_{\rm raw}$ and lower bound it assuming that direct processing of the experimental state incurs a constant physical noise rate $\lambda \in (0,1)$. In the second, the learner first uploads the state via the surface-code procedure of the previous section and then carries out the learning protocol fault-tolerantly without optimizing for noise. We denote its sample complexity by $N_{\rm inj}$, and allow the uploading stage to suffer a larger constant noise rate $\eta > \lambda$, reflecting the fact that injection is not noiseless. This additional noise is absorbed into a one-time encoding cost, after which the relevant information-processing primitives proceed fault-tolerantly.  Our goal is to show that $N_{\rm raw}$ and $N_{\rm inj}$ can differ exponentially.

\vspace{-1mm}
\vspace{-1em}
\subsection{Classical shadow tomography}
\vspace{-2mm}

Classical shadows rely on randomized preprocessing before measurement~\cite{Huang_2020}. For states supplied by an experiment, the key distinction is whether this preprocessing acts on an unprotected physical state or on a state that has first been uploaded into encoded memory. Building on upper bounds for shadows implemented by noisy brickwork circuits \cite{Hu_2025}, we compare uploading-enabled shallow-circuit shadows against the corresponding bare protocols, in which brickwork preprocessing is applied directly to the physical state.

\begin{theorem}[Uploading speedup for noisy shallow shadows, informal]
Fix a $k$-local Pauli observable on an $n$-qubit state. Let $N_{\rm raw}$ denote the sample complexity of the best raw protocol that estimates this observable by applying brickwork-circuit randomized preprocessing directly to the physical state, and let $N_{\rm inj}$ denote the sample complexity of the corresponding uploading-enabled shallow-shadows protocol. For constant physical noise rate $\lambda$ and injection noise $\lambda < \eta < c\lambda$ for a constant $c > 1$,
\begin{equation}
\frac{N_{\rm raw}}{N_{\rm inj}} \geq \exp\!\big(\Omega(\lambda k)\big).
\end{equation}
\end{theorem}

More broadly, our result shows that randomized preprocessing recovers its role as a flexible learning primitive once the experimental state has been promoted to an encoded logical state, extending classical shadows to states produced by Nature rather than by fault-tolerant hardware. In the raw experimental setting, the same primitive is fragile, since the randomizing circuit must itself run on the noisy physical state. Uploading disentangles these two issues by localizing the effect of noise to the input stage, leaving the learning protocol itself to proceed fault-tolerantly.

\vspace{-4mm}
\subsection{Third-moment observables}
\vspace{-2mm}
Our second family of examples concerns observables that depend on three copies of the state, including $\tr(\rho^3)$ and more general quantities of the form $\tr(O\rho^3)$. Such observables probe structure inaccessible to single-copy measurements, and arise naturally in settings ranging from entropy estimation and thermalization in quantum many-body physics \cite{Pasquale_Calabrese_2004, Daley_2012, Kaufman_2016} to error-mitigation primitives such as virtual distillation~\cite{Huggins_2021} and variants of quantum principal component analysis~\cite{Lloyd_2014}. They therefore provide a natural testbed for whether uploading can recover multi-copy quantum advantages that are otherwise lost to physical noise.

\vspace{-1mm}
\begin{theorem}[Uploading speedup for third moments and cubic observables, informal]
Let $N_{\rm raw}$ denote the optimal sample complexity among adaptive raw protocols that estimate $\tr(\rho^3)$ from noisy copies of an $n$-qubit state $\rho$ by applying at least one layer of joint unitary processing to those copies before arbitrary further processing, and let $N_{\rm inj}$ denote the sample complexity of the corresponding uploading-enabled protocol. For constant physical noise rate $\lambda$ and injection noise $\lambda < \eta < c\lambda$ for a constant $c > 1$,
\begin{equation}
\frac{N_{\rm raw}}{N_{\rm inj}} \geq \exp\!\big(\Omega(\lambda n)\big).
\end{equation}
The same scaling holds for the estimation of fixed cubic observables $\tr(O\rho^3)$.
\end{theorem}

The phenomenon is not restricted to these specific observables. Multi-copy primitives are among the most powerful tools in idealized quantum learning theory and among the most fragile under repeated physical noise \cite{cotler2026noisyquantumlearningtheory}: in the raw setting, several noisy copies must be brought together coherently before any protection has been established, and the structure one hopes to exploit is degraded at precisely that stage. Uploading changes the order of operations: the copies are first encoded, and the higher-order observable is then implemented fault-tolerantly. 

This lower bound is proved using a new Heisenberg learning-tree method, which extends the learning-tree formalism~\cite{chen2021exponentialseparationslearningquantum, weiyuan_Paulis, cotler2026noisyquantumlearningtheory} to settings where the bottleneck is unprotected noise rather than the number of available copies. The method reduces noisy adaptive multi-replica lower bounds to bounding the norm of a task-dependent distinguishability operator; see the Supplementary Material. Moreover, this lower bound only requires the raw learner to apply a single unitary preprocessing step directly to the physical states before running any subsequent fault-tolerant learning circuit that can include uploading. As such, this is a stronger result establishing that immediately uploading quantum states is not only a useful primitive, but is the \textit{optimal} strategy for learning from noisy quantum replicas. We turn next to a concrete instantiation of this architecture in an astronomical imaging task.

\vspace{-1em}
\subsection{Applications to physical experiments}
\vspace{-2mm}
Optical sensing and imaging are a natural near-term setting for quantum uploading, since several candidate transduction procedures are already understood~\cite{Chen_2008, Khabiboulline_2019, Huang_2022}. The goal in such experiments is to estimate observables of incoming quantum light under weak-signal constraints; uploading allows the light to be stored, encoded, and then processed using randomized or multi-copy primitives that would be prohibitively noisy on the raw physical state. We illustrate this in an experimental task where the comparison between uploading-enabled and raw noisy protocols yields a concrete gap in the number of shots required.

\begin{figure}[t!]
    \centering
    \includegraphics[width=\linewidth]{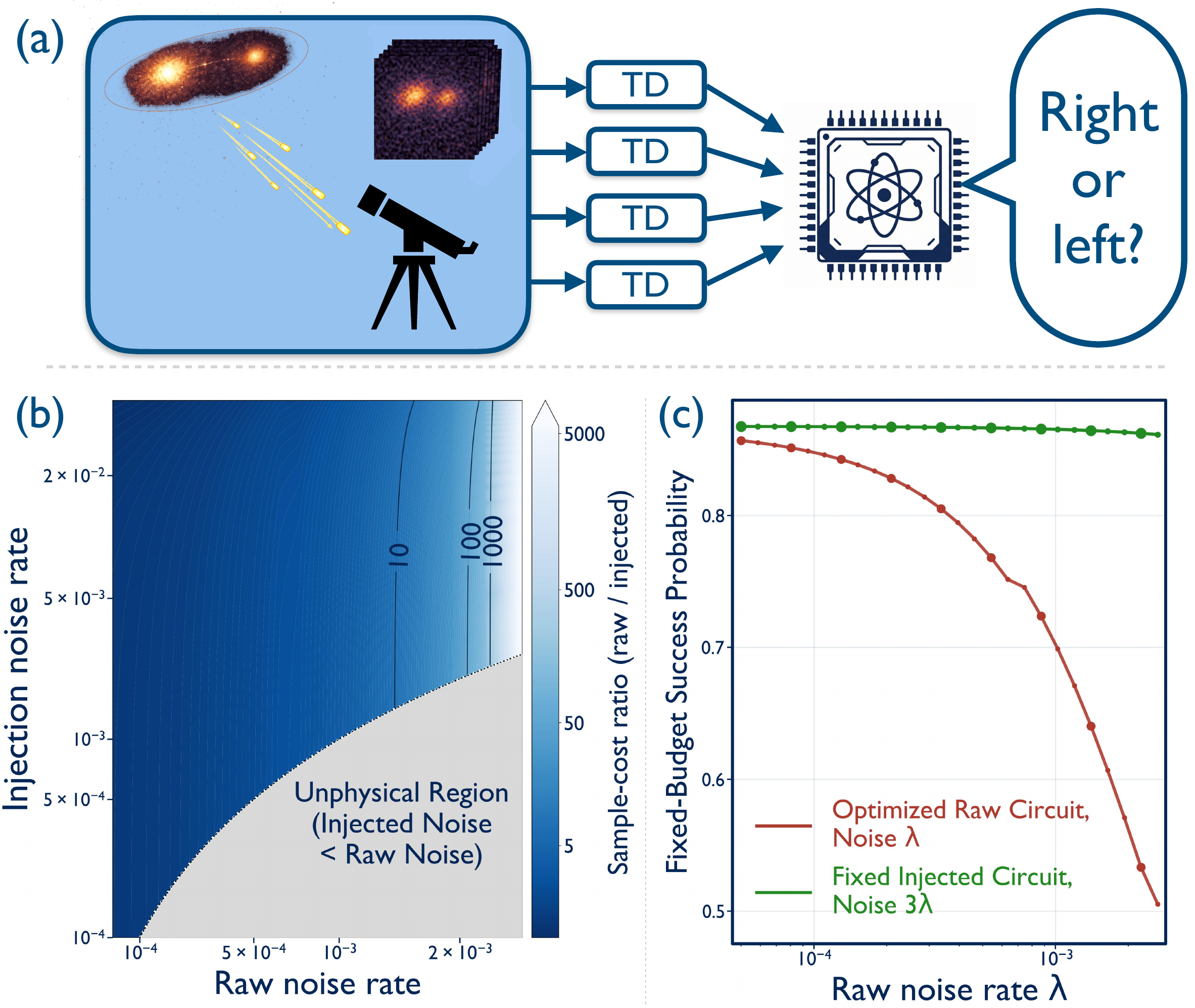}
    \caption{\textbf{Exponential speedup in exoplanet imaging.} (a) Hypothesis-testing formulation of exoplanet imaging: we determine whether a dim source lies spatially to the right or left of a bright star, using the optical imaging algorithm of~\cite{mokeev2025enhancingopticalimagingquantum}. (b) Ratio of sample complexities between raw and uploaded protocols to decide the task with $90\%$ success probability, across various noise rates. (c) Success probability as a function of raw noise rate for a fixed photon budget. The uploaded algorithm is assumed to suffer three times the noise rate of the raw protocol.}
    \label{fig:astronomy}
    \vspace{-6mm}
\end{figure}

The specific task is resolving a faint exoplanet next to a much brighter parent star, formulated as a hypothesis test for whether the dim source lies to the right or the left. Ref.~\cite{mokeev2025enhancingopticalimagingquantum} proposes a quantum algorithm in which a quantum memory stores single-photon density matrices and quantum PCA~\cite{Lloyd_2014} is used to prepare the state's eigenvectors. Combined with interleaved quantum signal processing~\cite{motlagh2024generalizedquantumsignalprocessing}, the procedure sorts quantum noise from the bright background into a branch that can be postselected out. Under a sufficiently weak noise model, \cite{mokeev2025enhancingopticalimagingquantum} shows that this images the exoplanet more efficiently than any classical telescope. However, the protocol requires coherent storage and processing of many photons, so more realistic control noise rapidly becomes prohibitive.

Figure~\ref{fig:astronomy} shows direct simulation of the pipeline under both uploading-enabled and raw noise models. At optimistic noise rates of order $10^{-3}$, locating the exoplanet with constant success probability requires thousands of times more shots under raw noise than with uploading, even when the uploading stage suffers three times more noise. Moreover, in practice the total photon budget is effectively fixed, since learning must complete before atmospheric noise shifts. Constrained to a budget, the raw strategy's success probability decays exponentially with physical noise rate while the uploaded strategy's does not.  In the raw simulation, we optimize the circuit to achieve the best hypothesis testing accuracy at each noise rate. This information is not available in practice; our simulations give the raw learner substantially more information than its uploaded counterpart, which runs a fixed circuit at every noise rate. This illustrates how non-adaptive, uploading-enhanced learning substantially outperforms adaptive and error-mitigated raw quantum learning, even when the uploading stage incurs much larger error than the bare interface. Details are provided in the Supplementary Material.

\vspace{-1em}
\section{Discussion}
\vspace{-2mm}

The ability to learn from rich datasets underpins much of modern computation, and quantum computers equipped with quantum data may be able to probe facets of Nature beyond our current reach. Unlike in the classical world, where loading diverse data sources into memory is comparatively cheap, fault-tolerant quantum computers cannot naively manipulate fragile quantum data without incurring severe resource overheads. Our work bridges that gap, showing that any dataset of noisy quantum states can be integrated into fault-tolerant quantum computation, opening the door to the full suite of modern quantum learning strategies.

In particular, our results draw a sharp distinction between two ways of learning from a noisy quantum experiment. In one, the unknown state remains in unprotected physical degrees of freedom throughout the protocol, and the learning strategy is optimized to minimize the effects of noise. In the other, the state is first transferred into an encoded form, paying a one-time noise penalty, and only then processed fault-tolerantly. The difference between these two architectures can be exponential, and the advantage comes not from removing noise, but from concentrating its effect into a one-time cost.

Earlier no-go results show that even a single noisy layer can destroy an ideal exponential separation~\cite{cotler2026noisyquantumlearningtheory}. A noisy uploading interface could therefore have erased any later benefit of fault-tolerant processing, but the separations proved here are robust: uploading remains exponentially preferable to direct processing even when the uploading stage itself suffers a noise rate larger than the raw experimental interface by a substantial constant factor. 

Many scientifically interesting experiments produce quantum states that are too fragile to be processed coherently in situ, yet too structured for classical postprocessing alone to extract everything of interest. Our results suggest that fault tolerance becomes useful precisely in this intermediate regime, since the state produced by the experiment can be brought into an encoded form before the learning task is carried out. The longer-term significance of quantum uploading may lie less in the particular examples treated here than in the architectural point they illustrate: once physical outputs can be converted into encoded logical inputs at constant cost, the asymptotic tools of quantum learning theory become directly relevant to scientific measurement.

\vspace{-1em}
\section*{Data availability}
\vspace{-2mm}
Code and data for Figure~\ref{fig:astronomy} can be found at \href{https://github.com/ishaan-kannan/QSP_QPCA_Optical_Imaging}{this link}.

\vspace{-1em}
\section*{Acknowledgements}
\vspace{-2mm}

We thank Johannes Borregaard, Mikhail Lukin, Rohan Mehta, and Quynh Nguyen for helpful conversations. We acknowledge the use of GPT 5.4 Codex to assist in writing the code used to generate Figure \ref{fig:astronomy}; the code was reviewed, tested, and validated by the authors.  IK is supported in part by the Nobile Research Initiative.  HP is supported by the Department of Defense through the National Defense Science and Engineering Graduate (NDSEG) Fellowship Program. JC is supported by an Alfred P.~Sloan Fellowship.

\putbib
\end{bibunit}
\clearpage
\onecolumngrid
\appendix

\renewcommand{\thesection}{S\arabic{section}}
\renewcommand{\thefigure}{S\arabic{figure}}
\renewcommand{\thetable}{S\arabic{table}}
\renewcommand{\theequation}{S\arabic{equation}}
\setcounter{equation}{0}
\setcounter{figure}{0}
\setcounter{table}{0}
\setcounter{tocdepth}{2}

\begin{bibunit}
\clearpage
\onecolumngrid
\startcontents[appendix]

\counterwithout*{equation}{section}
\setcounter{equation}{0}
\renewcommand{\theequation}{S\arabic{equation}}

\renewcommand{\thesection}{S\arabic{section}}
\renewcommand{\thefigure}{S\arabic{figure}}
\renewcommand{\thetable}{S\arabic{table}}

\allowdisplaybreaks

\vspace*{1em}
\begin{center}
  {\large\bfseries Supplemental Material for ``Fault-tolerant\\ quantum processing of physical experiments''\par}
\end{center}
\vspace{1em}

\begingroup
  \hypersetup{hidelinks}

  \begin{center}
    \textbf{Contents}
  \end{center}
  \vspace{0.5em}

  \printcontents[appendix]{}{1}{}
\endgroup

\section{Related Work}
\label{sec:related}

\noindent\textbf{State injection and fault-tolerance with arbitrary input.} Our introduction of uploading to the quantum learning literature is inspired by early work on universal fault-tolerant quantum computation (FTQC) with the surface codes. The surface code was introduced in Ref.~\cite{DennisKitaevLandahlPreskill2002TopologicalMemory} and the rotated surface code variant we use in Ref.~\cite{Bombin2007OptimalResources, Horsman_2012}. The concept of growing a small logical patch into a larger one was the standard practical proposal for increasing coherent circuit depth since early work on FTQC~\cite{Fowler_2009, Fowler_2012}. Several works, e.g.~\cite{Horsman_2012,  Lodyga2015SimpleSchemeTopologicalCodes, Li2015MagicStateFidelity}, then instantiated explicit circuits to implement this idea in specific versions of the surface code, enabling FTQC using magic state injection and distillation. These protocols generally involve lifting a single qubit into a constant-distance logical qubit with an explicit circuit, then using stabilizer measurements to grow to larger distance. 

More recently, rigorous treatments of FTQC by chaining fault-tolerant logical gadgets equipped heuristic injection protocols with formal guarantees, instantiating them as standalone gadgets that can couple into fault-tolerant circuits when non-Clifford resources are required. Ref.~\cite{Christandl2024QuantumIO} extended fault-tolerance to circuits with generic quantum input and output, showing that noisy implementations can realize the ideal circuit up to controlled noise channels acting on the boundary systems. As such, their results generalize quantum uploading in an abstract sense. Ref.~\cite{He2025ComposableFaultTolerance} makes this style of reasoning composable, giving a library of gadgets and a framework within which they can be chained together to obtain fault-tolerance guarantees for general circuits. Their results include a proof that injection of arbitrary states at bounded error cost is theoretically possible with any quantum code, and an instantiation of an injection gadget for the surface code that requires $O(d_2^2)$ rounds of measurements, where $d_2$ denotes the distance of the final code patch. This quadratic overhead is due to growing the distance of the code patch incrementally and performing $O(d_2)$ measurements during each growth step. In contrast, our gadget uses only $O(d_2)$ rounds of measurement in total by growing from an initial distance $d_1$ to $d_2$ immediately. Our gadget retains the relevant composable fault-tolerance guarantees due to the space-time dependent ``avoiding sets" we introduce in our argument; see Appendix \ref{app:uploading_proofs} for details. State injection has primarily been seen as a practical tool for FTQC, where the injected qubits are in known quantum states. In this work, we reinterpret injection as a tool for enabling fault-tolerant processing of quantum experiments when the experimental protocol is conducted outside of a quantum processor.

\vspace{2em}
\noindent\textbf{Quantum-enhanced experiments that interface with fault-tolerance schemes.} Existing proposals that coherently couple experimental signals to protected quantum hardware are primarily based in astronomy and metrology. On the astronomy side, Ref.~\cite{Huang_2022} proposes capturing starlight into atomic memories and then encoding it into a quantum error-correcting code so that subsequent imaging operations are noise-protected. Ref.~\cite{Khabiboulline_2019} proposes optical interferometry assisted by quantum networks, where the incoming photon state and an arrival-time index are stored in binary qubit codes and retrieved nonlocally through entanglement-assisted parity checks. More recently, Ref.~\cite{mokeev2025enhancingopticalimagingquantum} shows how weak optical signals can be coherently encoded into qubit registers and then processed with quantum principal component analysis, quantum signal processing, and block encoding, with exoplanet imaging as the motivating example. While their work does not discuss the use of quantum error-correction, we numerically implement their algorithm under both raw and uploaded noise models to demonstrate that quantum uploading can protect the efficacy of this quantum-enhanced experiment when noise would otherwise destroy it. 

On the metrology side, several recent works consider applying quantum information processing and error correction to sensors such as Nitrogen-Vacancy color centers in diamond~\cite{Unden_2016, wang2025noiseresilientquantummetrologyquantum}. Moreover, Ref.~\cite{marrero2026encodedquantumsignalprocessing} recently showed that for certain structured noise models, quantum error correction and signal processing can be unified to enable sensing precision beyond the standard quantum limit. While the idea of using error correction in metrology is well-explored, it remains unclear how to use error-correction to recover speedups under general noise. This is because in sensing, the signal is accumulated on the probe coherently rather than supplied as repeated state-preparation access to an unknown quantum state. As a result, errors can continuously accumulate, and it is often impossible to completely eliminate them without annihilating the signal. Moreover, because sensing occurs at the physical-qubit level, it is unclear how to utilize quantum error-correcting codes to achieve Heisenberg-limited sensitivity unless the signal Hamiltonian is known to act logically on a particular codespace. However, when sensing tasks can be reframed as state-learning problems, such as in precision interferometry \cite{cotler2026quantumadvantagesensingproperties} or astronomy, paying a one-time uploading cost can indeed recover large speedups under general noise. Our results therefore suggest that it is fruitful to reframe sensing tasks as learning from \textit{shots} of quantum data over time rather than amplifying coherent accumulation of signal. It would be valuable to define state-learning problems of relevance in other scientific disciplines that may one day benefit from FTQC.

\vspace{2em}
\noindent\textbf{Noise-aware quantum learning.} Ref.~\cite{cotler2026noisyquantumlearningtheory} instantiates a rigorous study of quantum experiments coupled to fault-tolerant quantum devices where coupling to an unknown system can be a noisy process. This work introduces the complexity class $\textsf{NBQP}$ and proves the relativized separation $\textsf{NBQP}^O \subsetneq \textsf{BQP}^O$, indicating that fault-tolerant learning from noisy, real-world experiments can be exponentially more challenging than idealized quantum learning models assume. In Corollary \ref{cor:uploaded_oracle_reduction} we refine this picture, showing that the separation becomes an equality once a single layer of noise is applied to each $\textsf{BQP}$ oracle query, because injection and FTQC allow all subsequent processing in the $\textsf{NBQP}$ model to be performed effectively noiselessly. Ref.~\cite{cotler2026noisyquantumlearningtheory} also suggests exploiting noise-robust structure intrinsic to some physical systems, such as quantum many-body systems with thermalizing dynamics \cite{venkatesh2026quantumthermalizationachievesoptimal} or systems that are well-modeled by tensor networks with hierarchical structure \cite{Pastawski_2015}; determining whether these physical sources of error correction can be interfaced with uploading or other engineered fault-tolerance schemes is a valuable direction for future work on quantum-enhanced experimentation. 

Several recent works in quantum learning and sensing also consider applications of quantum learning primitives to concrete experiments. Ref.~\cite{Oh_2024} proves an exponential entanglement-assisted separation for learning bosonic random displacement channels in the continuous-variable setting that uses bosonic squeezing to partially mitigate the effects of loss. Ref.~\cite{cotler2026quantumadvantagesensingproperties} considers a generalized version of this bosonic learning task for quantum sensing of classical fields. This work studies estimation of a wide variety of field properties beyond the characteristic function, using squeezing to similar efficacy and proving quantum speedups that operate in the presence of stochastic, noisy backgrounds. Complementarily, Ref.~\cite{romanov2026learningarbitrarylindbladiansquantum} studies ansatz-free learning of quantum Lindbladians by using active quantum error correction, which is a task where both coherent and dissipative signals are learned simultaneously in situ. Control imperfections can meaningfully hamper the performance of the quantum algorithms presented in these works. This is because noise can remove the ability to process experimental data coherently before the relevant information is lost. Our work suggests that future fault-tolerant implementations of such proposals may recover most of the gains of certain experimental protocols present in the noiseless setting.

\section{Preliminaries}

\subsection{Quantum information and quantum complexity theory}
\label{Appendix:basic_quantum_prelim}

Here we first collect standard definitions of quantum measurements used in quantum information theory. Then we recall recent complexity-theoretic definitions from Ref. \cite{cotler2026noisyquantumlearningtheory} relevant to quantum learning tasks that leverage fault-tolerant quantum computers in the presence of realistic noise. Throughout, we use $\mathds{1}_n$ to denote the $2^n \times 2^n$ identity matrix.

\begin{definition}[POVM]
An $n$-qubit Positive Operator-Valued Measure (POVM) is given by a set of matrices $\{F_s\}$ such that all $F_s$ are positive semi-definite and $\sum_s F_s = \mathds{1}_{n}$. Given a density matrix $\rho$, when we say we measure $\{F_s\}$ on $\rho$, we obtain the classical outcome $s$ sampled from the distribution $\textnormal{Pr}[s] = \tr(F_s\rho)$.
\end{definition}
\noindent POVMs are the most general class of quantum measurement relevant to this work, and in general may include measurements performed using deep circuits with a large number of ancillas. A well-known fact is that the classical outcome distribution of an arbitrary POVM can be simulated by a POVM consisting only of rank-1 matrices:
\begin{lemma}[Simulating arbitrary POVMs with rank-1 POVMs, e.g.~Lemma 4.8 in \cite{chen2021exponentialseparationslearningquantum}]
If we neglect the post-measurement quantum state, the outcome distribution of any arbitrary $k$-qubit POVM can be simulated (using classical postprocessing) by a POVM of the form $\{w_s2^n\ketbra{\psi}{\psi}\}$, where $\ket{\psi}$ is a pure quantum state and $\sum_s w_s = 1$. 
\end{lemma}
\noindent This lemma tells us that in quantum learning tasks where we discard a state after measurement, we only need to consider rank-1 POVMs.

In this work, we model noise as a local depolarizing channel, given by the following definition.
\begin{definition}[Single-qubit depolarizing channel] The single-qubit depolarizing channel with strength $\lambda\in [0, 1]$ is
\begin{equation}
    \mathcal{D}_\lambda(\rho) = (1-\lambda)\rho + \lambda\, \frac{I}{2}
\end{equation}
for any single-qubit density matrix $\rho$. The depolarizing channel is self-adjoint with respect to the Hilbert-Schmidt inner product, and on a general 2-by-2 matrix $A$, (the adjoint of) $\mathcal{D}$ acts as
\begin{equation}
    \mathcal{D}_\lambda(A) = (1-\lambda)A + \lambda\,\frac{\tr(A)\,I}{2}\,.
\end{equation}
\end{definition}
\noindent We assume that the channel $\mathcal{D}_\lambda^{\otimes m}$ acts on our full register of $m = \textnormal{poly}(n)$ qubits after each depth-1 unitary operation. We work in the circuit model where depth-1 unitaries are comprised of arbitrary parallel $2$-qubit unitary gates. With this noise model established, we turn to definitions of complexity classes for noisy quantum learning. First we recall the most general of an $n$-qubit quantum oracle.

\begin{definition}[$n$-qubit quantum oracle]
An $n$-qubit quantum oracle $O$ is a completely positive trace preserving map from $\mathcal{L}((\mathbb{C}^2)^{\otimes m})\rightarrow \mathcal{L}((\mathbb{C}^2)^{\otimes n})$ with $m \leq O(
\textnormal{poly}(n))$.
\end{definition}
\noindent There are two special cases to note. 
\begin{definition}[State-preparation oracle]
An $n$-qubit state preparation oracle $O$ is a quantum oracle from input domain $1\rightarrow \mathcal{L}((\mathbb{C}^2)^{\otimes n})$. That is, it loads a fixed quantum state $\rho$ onto an $n$-qubit register.
\end{definition}
\noindent Another special case which will not be utilized in this work, but is useful for understanding the formal complexity-theoretic language of \cite{cotler2026noisyquantumlearningtheory}, is the following.
\begin{definition}[Quantum instantiation of classical oracle]
A classical oracle $O_C$ is a function from $\{0, 1\}^n\rightarrow \{0, 1\}^m$ for $n,m \in \mathbb{N}$. The quantum instantiation of $O_C$ is the unitary $U_{O_C}$ acting on computational basis states $\ket{x}\!, \ket{y}$ as $U_{O_C}\ket{x}\ket{y} = \ket{x}\ket{y\oplus O_C(x)}$.
\end{definition}
\noindent This definition allows for nontrivial comparisons between quantum and classical complexity classes, and further enables \cite{cotler2026noisyquantumlearningtheory} to prove relativized separations between quantum complexity classes using classical decision problems. Now the following definition instantiates the computational model used within the $\nbqp$ complexity class, which allows the algorithm to perform mid-circuit measurements and state preparation that is sufficient for fault-tolerant quantum computation.

\begin{definition}[Noisy quantum algorithm with oracle access]
\label{def:NBQP_circ}
Let $O$ be an $n$-qubit quantum oracle. A $\lambda$-noisy quantum algorithm $Q_\lambda^O$ with access to $O$ is a uniform family of quantum channels $\{C_n^{O}\}_{n\in\mathbb{N}}$, where for each input length $n$ the channel $C_n^{O}$ acts on $n' \leq \mathrm{poly}(n)$ qubits. We refer to $n'$ as the total number of qubits used by the algorithm on inputs of size $n$.  Each $C_n^{O}$, which we call a $\lambda$-noisy quantum circuit, has the form
\begin{equation}
C_n^{O}[\rho] = V_{k,n}\,\mathcal{D}_\lambda\bigl( V_{k-1,n}\,\mathcal{D}_\lambda\bigl(\cdots V_{2,n}\,\mathcal{D}_\lambda\bigl(V_{1,n}\,\rho\,V_{1,n}^\dagger\bigr)V_{2,n}^\dagger \cdots \bigr) V_{k-1,n}^\dagger \bigr) V_{k,n}^\dagger,
\end{equation}
for some integer $k \leq \mathrm{poly}(n)$. For each $t = 1,\dots,k$, the unitary $V_{t,n}$ is either \textnormal{(i)} a depth-1 unitary on the $n'$ qubits, or \textnormal{(ii)} an oracle layer of the form $O \otimes I_{2^{n'-m}}$, where $O$ acts on some chosen subset of $m$ of the $n'$ qubits and the identity acts on the remaining $n' - m$ qubits.  For fixed $n$, the output state of $Q_\lambda^O$ is $C_n^{O}\left[\bigl|0^{n'}\rangle\langle0^{n'}\bigr|\right]$, and the runtime of the algorithm on inputs of length $n$ is $\Theta(k)$.
\end{definition}
\noindent We can now state the formal definition of $\nbqp$ as a functional complexity class.
\begin{definition}[\textsf{NBQP}$^O$ complexity class]\label{def:NBQP}
Let $f : \{0,1\}^\star \rightarrow \{0,1\}^\star$ be a (total) function. We say that $f$ is in \textnormal{\textsf{NBQP}}$^O$ if there exist a constant $\lambda > 0$, a polynomial $p(\cdot)$, and a $\lambda$-noisy quantum algorithm $Q_\lambda^O$ with access to $O$ (as in Definition~\ref{def:NBQP_circ}) such that, for every input $x \in \{0,1\}^\star$ of length $n = |x|$:
\begin{enumerate}
    \item The algorithm $Q_\lambda^O$ on input $x$ acts on $n' \leq p(n)$ qubits, uses at most $p(n)$ layers (i.e.~$k \leq p(n)$ in Definition~\ref{def:NBQP_circ}), and hence has total running time at most $p(n)$;
    \item Measuring all qubits of the output state $C_n^{O}\!\left[\bigl|0^{n'}\rangle\langle0^{n'}\bigr|\right]$ in the computational basis yields a classical bit string of length at most $p(n)$;
    \item The resulting measurement outcome equals $f(x)$ with probability at least $2/3$.
\end{enumerate}
\end{definition}
\noindent One can also work with $\nbqp_\lambda^O$ with $\lambda > 0$ to select an explicit noise rate. If $\lambda$ were set to $0$, the previous definition would reduce to exactly the $\bqp^O$ complexity class. As done in \cite{cotler2026noisyquantumlearningtheory}, problems of learning properties of quantum states can be naturally fit into the functional notion of $\nbqp^O$.

\begin{remark}[Property testing with state-preparation oracles]
In Definition~\ref{def:NBQP_circ}, we can specifically choose $O$ to be a state-preparation oracle $1\rightarrow \rho_O$. Now, for a two-property testing problem, suppose that for each $n$ there are two disjoint classes of states $\mathcal{P}_0(n)$ and $\mathcal{P}_1(n)$, and we are promised that $\rho_O$ belongs to exactly one of them. For each such oracle $O$, define a function $f^O : \{0,1\}^\star \to \{0,1\}$ by
\begin{align}
f^O(1^n) =
\begin{cases}
0 & \textnormal{if } \rho_O \in \mathcal{P}_0(n)\\
1 & \textnormal{if } \rho_O \in \mathcal{P}_1(n)
\end{cases}\,,
\end{align}
where the input is simply the unary string $1^n$ encoding the system size. An \textnormal{\textsf{NBQP}}$^O$ algorithm for this property-testing task is then a $\lambda$-noisy quantum algorithm $Q_\lambda^O$ that, on input $1^n$ and with access to $O$, outputs $f^O(1^n)$ with probability at least $2/3$ using only polynomial resources. 
\end{remark}

\subsection{Minimax learning lower bounds}
The standard strategy for proving lower bounds on the sample complexity of a learning problem is via reduction to a distinguishing task. In short, given an algorithm that can learn a particular property of a quantum system to precision $\epsilon$, the same algorithm will be able to distinguish two systems which differ by $> 2\epsilon$ in the property. As a result, a lower bound on the sample complexity of the reduced distinguishing problem implies a lower bound on the learning task. In this work, we will use the following kinds of distinguishing problem in lower bounds.

\begin{definition}[One-vs-one distinguishing problems]
Consider two quantum states $\rho_0, \rho_1$, and select the ground truth state $\rho = \rho_0$ or $\rho_1$ with equal probability. A one-vs-one distinguishing task is simply to determine the ground truth given copies of $\rho$.

\end{definition}
\begin{definition}[Many-vs-many distinguishing problems]
Consider two ensembles of quantum states $\mathcal{E}_0, \mathcal{E}_1$, and distributions $D_0(\mathcal{E}_0), D_1(\mathcal{E}_1)$ over them respectively. We then choose a ground truth distribution $D$ equiprobably at random to be either $D_0(\mathcal{E}_0)$ or $D_1(\mathcal{E}_1)$. Then given copies of a fixed state $\rho \sim D$, a many-vs.-many distinguishing task is to decide whether $D = D_0$ or $D_1$. 
\end{definition}

Once an appropriate distinguishing problem is presented, the next step is to impose the desired constraints on the measurements we can perform. These constraints will control the classical probability densities that govern our solutions to the hypothesis testing problem, taking us from quantum state access to classical success probabilities. For this, one interprets an adaptive experiment as a rooted tree.

\begin{definition}[Tree representation for a quantum learning algorithm]\label{def:general_learning_tree}
A quantum learning algorithm with access to a fixed $n$-qubit state $\rho$ and $\textnormal{poly}(n)$ classical memory can be represented as a rooted tree $\mathcal{T}$ with the following properties:
\begin{itemize}
    \item Each node $u$ in $\mathcal{T}$ is associated with a POVM $\{M_s^u\}$ on $k \leq n$ qubits, which may be drawn from some subset of all possible $n$-qubit POVMs.
    \item Each node $u$, at depth $d$ in the tree, has an associated probability $p_\rho(u)$ denoting the probability that the state of the algorithm is represented by $u$ after $d$ measurements.
    \item Each non-leaf node $u$ is joined to its children by edges $e_{u, s}$, where $s$ corresponds to the classical outcome of the measurement performed at $u$. For a child node $v$, the transition rule is given by
    \begin{equation}
        p_\rho(v) = p_\rho(u) \,\tr(\rho(M_s^u\otimes \mathds{1}_{n-k}) )\,,
    \end{equation}
    where the identity acts on any qubits not measured under the POVM. 
    \item Every root-to-leaf path has $T$ edges. 
\end{itemize}
To specify a learning tree, we specify the size of the quantum register and the set of allowed POVMs at each node.
\end{definition}
This definition allows us to embed a wide variety of constraints. For instance, a protocol which measures $k$ copies of an $m$-qubit quantum state $\sigma$ at a time can set $n = mk$ in defining the tree, and set $\rho = \sigma^{\otimes k}$. If noise or other constraints are present, these can be incorporated directly into the set of allowable POVMs at each node. We will also use the following notation in our lower bound.

\begin{definition}[$\mathcal{M}_{k ,T}$ POVMs] \label{def:m_k,t_POVMs}
Let $\mathcal{H}_n$ denote the $n$-qubit Hilbert space, where $n$ is set by the problem at hand. The set of POVMs $\mathcal{M}_{k ,T}$ consists of all POVMs described by a depth-$T$ learning tree in which each non-leaf node is assigned an arbitrary POVM acting on $\mathcal{H}_n^{\otimes k}$ (the Hilbert space obtained from $k$ copies of the $n$-qubit Hilbert space), with the choice of POVM allowed to depend on all previous classical outcomes.
\end{definition}

With the tree definition encoding our problem constraints and providing us with a mapping from quantum states to potential classical outcome distributions, the key tool in proving an information-theoretic lower bound is Le Cam's two-point method.

\begin{lemma}[Le Cam's Two-Point Method]\label{lemma:le_cam}
Given a two-hypothesis distinguishing problem with hypothesis states $\rho \sim D_0, \sigma\sim D_1$ and learning tree $\mathcal{T}$ representing a quantum algorithm for this problem, the probability that the algorithm selects the correct hypothesis is upper bounded by
\begin{equation}
    \frac{1}{2} + \frac{1}{2}\sum_{\ell\in \textnormal{leaf}(\mathcal{T})} \Big|\mathbb{E}_{\rho \sim D_0}[p_{\rho}(\ell)] - \mathbb{E}_{\sigma \sim D_1}[p_{\sigma}(\ell)]\Big|
\end{equation}
The divergence $\frac{1}{2}\sum_{\ell\in \textnormal{leaf}(\mathcal{T})} \Big|\mathbb{E}_{\rho \sim D_0}[p_{\rho}(\ell)] - \mathbb{E}_{\sigma \sim D_1}[p_{\sigma}(\ell)]\Big|$ is the total variation distance $d_{\text{\rm TV}}(\mathbb{E}_\mathcal{D_0}[p_{\rho}(\ell)],\mathbb{E}_\mathcal{D_0}[p_{\sigma}(\ell)])$.
\end{lemma}

\noindent Because the learning tree permits adaptivity, Lemma \ref{lemma:le_cam} allows us to bound the success probability of any (possibly adaptive) protocol for the hypothesis testing problem in terms of the tree's depth. This imposes a lower bound on tree depth, and thereby sample complexity, when one requires constant success bias, because the tree must be sufficiently deep that the total variation remains constant rather than decaying with system size.

Our goal will be to argue that for the total variation bound in Lemma \ref{lemma:le_cam} to become appreciably large, the tree depth $T$ must grow exponentially in $n$. To do so, it suffices to bound the one-sided likelihood ratio.  Let us first define this ratio, and then subsequently explain why bounding it is useful for us.
\begin{definition}[One-sided likelihood ratio]
Assume a learning tree $\mathcal{T}$ and a one-vs.-one distinguishing problem with hypothesis states $\rho$ and $\sigma$. For every $\ell\in \textnormal{leaf}(\mathcal{T})$, the one-sided likelihood ratio is
\begin{equation}
    L(\ell) = \frac{p_{\rho}(\ell)}{p_\sigma(\ell)}\,.
\end{equation}
Moreover, we define the intermediate edge likelihood ratio:
\begin{equation}
    L(s|u) = \frac{p_{\rho}(s|u)}{p_\sigma(s|u)}
\end{equation}
given by the outcome distributions of a POVM at node $u$.
\end{definition}
Previous works have used ensemble-averaged likelihood ratios to contend with many-vs.-one hypothesis testing problems such as Pauli shadow tomography \cite{weiyuan_Paulis, cotler2026noisyquantumlearningtheory}. However, in the Heisenberg learning tree method, we will first reduce the learning tree for a many-vs.-many problem to a learning tree for a one-vs.-one problem, so that only the above likelihood ratio is needed. To then leverage the likelihood ratio to obtain a lower bound, we will use the following lemma, which is a simplified version of the martingale technique introduced in \cite{weiyuan_Paulis}. We provide a proof of this simpler case for completeness.

\begin{lemma}[Martingale lemma for one-vs.-one hypothesis testing.]\label{lemma:martingale}
    Suppose $\mathcal{T}$ is a learning tree for a one-vs.-one distinguishing problem. If the condition
    \begin{equation}
      \mathbb{E}_{s\sim p^{\sigma}(s|u)} \left[\left(L(s|u) - 1\right)^2\right]\leq \delta
    \end{equation}
    is satisfied for all $u$ in $\mathcal{T}$, then $\mathcal{T}$ must have depth $T >\Omega(1/\delta)$ for the algorithm to succeed with probability $>2/3$. 
\end{lemma}
\begin{proof}
    Let $H_t = (s_1, s_2,...,s_t)$ denote a sequence of outcomes over $t$ measurements, i.e. a length-$t$ path starting at the root of $\mathcal{T}$. Then, let $P^{(t)}(H_t), Q^{(t)}(H_t)$ denote the classical outcome probability densities over values of $H_t$ at tree depth $t$, under hypothesis state $\rho$ and $\sigma$ respectively. Define the random variable
    \begin{equation}
        Z_t = \frac{P^{(t)}(H_t)}{Q^{(t)}(H_t)}  = \prod_{i=1}^t L(s_i|t_i).
    \end{equation}
    Letting $H_t = (\Vec{h}, s_t)$, it follows that  $P^{(t)}(h) = P^{(t-1)}(\Vec{h})P(s|u)$, where we define $P(s|u)$ to be the single-outcome distribution given by $\tr(F_s^u \rho)$. The same holds for $Q$. As such, we have
    \begin{equation}
        Z_t(H_t) = \frac{P^{(t-1)}(\Vec{h})P(s|u)}{Q^{(t-1)}(\Vec{h})Q(s|u)} = Z_{t-1}(H_{t-1}) L(s|u) \ .
    \end{equation}
    It follows that under the ground truth density $Q$, $\mathbb{E}_Q[Z_t | H_{t-1}] = Z_{t-1}\mathbb{E}_{Q_u}[L(s|u)] = Z_{t-1}$, because the expectation of the likelihood ratio is $1$ by definition. As such, $Z_t$ is a martingale under $Q$. By the same reasoning, \begin{equation}
        \mathbb{E}_Q[Z_t^2|H_{t-1}] = Z_{t-1}^2 \mathbb{E}_{Q_u}[L(s|u)^2] = Z_{t-1}^2 \, \mathbb{E}_{Q_u}[(L(s|u) - 1)^2 + 1]  \ ,
    \end{equation}
    where in the last step we expand the square and use that the expectation of the likelihood ratio is always $1$. Applying the assumed concentration condition and iterating for $T$ rounds, we obtain
    \begin{equation}
        \mathbb{E}_{Q}[Z_T^2] = (1+\delta)^T \ .
    \end{equation}
    Now since the $\chi^2$-divergence is exactly defined as $\mathbb{E}_Q[(Z_T-1)^2] = \mathbb{E}_Q[Z_T^2] - 1$, it follows by Pinsker's inequality that
    \begin{equation}
        d_{\rm TV}(P^{(T)}, Q^{(T)}) \leq \sqrt{\frac{1}{2}((1+\delta)^T - 1)}\leq \frac{1}{2}(e^{\delta T} - 1) \ .
    \end{equation}
    For this to be $\Theta(1)$ (specifically, it must be at least $1/6$ to achieve success probability at least $2/3$), we require $T \geq \Omega(1/\delta)$.
\end{proof}
\noindent Via Lemma \ref{lemma:martingale}, a nodewise likelihood ratio upper bound will translate to a sample complexity lower bound. Our proofs will obtain this lower bound by first translating the many-vs.-many task into a one-vs.-one task with different hypothesis states, then controlling the likelihood ratio for this new learning tree using a Heisenberg-picture operator. 

\subsection{Heisenberg learning tree method} \label{app:Heisenberg_learning_tree}
In this section, we detail the proof structure of the Heisenberg learning tree method. We only review the high-level idea behind existing learning tree proof strategies here; for a detailed description, see \cite{chen2021exponentialseparationslearningquantum}. 

Broadly speaking, the learning tree formalism has thus far consisted of three approaches of increasing complexity. The "tree-based" approach is essentially equivalent to direct application of Le Cam's two-point lemma. One uses the elementary fact that the one-sided likelihood ratio directly controls total variation, and obtains a tight likelihood ratio bound at every level of the tree in conjunction with a simple union bound. This requires a close match in the first moments of the measurement transcript distributions, which is often a stronger requirement than can be proven in interesting cases.

The first method that uses higher moments of the tree structure is the edge-based method. Here, one requires a uniform upper bound on the variance of the likelihood ratio at every intermediate node in the tree. Similarly to Lemma \ref{lemma:martingale}, this results in a tighter bound on the total variation between leaf distributions than is obtained from first moments alone. However, for some problems it is possible that in rare cases the likelihood ratio of a particular edge in the tree is far from $1$, even when the vast majority of edges remain close to $1$. This behavior is only captured by higher moments of the tree, causing the edge-based method to fail. One then relies on the path-based method, in which one focuses on a typical root-to-leaf path in the tree and probabilistically bounds its likelihood ratio. By controlling the number of "bad" paths along which the likelihood ratio strays from $1$ while bounding the "good" paths using higher moments of the transcript distribution, one again obtains a total variation bound. While this is the most powerful bound one can in principle obtain, as it uses all accessible moments of the tree, it is often very unwieldy to handle higher moments of tensor product distributions over states or unitaries. This makes it difficult to use in cases where the limited resource is not the number of copies per experiment, but another resource such as coherence. 

To circumvent some limitations of the edge-based and path-based approaches, we utilize a combined approach which we call the Heisenberg learning tree method for many-vs.-many (and many-vs.-one) hypothesis testing problems. The core component of the method is the following lemma.

\begin{lemma}[Heisenberg bound for one-vs.-one distinguishing] \label{thm:Heisenberg} Consider a one-vs.-one distinguishing problem with hypothesis states $\sigma_p, \sigma_q$. Further suppose that in each experiment, we consider POVMs acting on adaptively post-processed states $\sigma_p^u = \Phi_u[\sigma_p], \sigma_q^u = \Phi_u[\sigma_q]$, where $\Phi_u$ is a quantum channel that can include adaptive unitary controls and noise. Suppose $\mu \mathds{1} \preceq \sigma_q^u$ for some scalar $\mu > 0$. Let $\mathcal{T}$ denote the learning tree in which arbitrary POVMs are made on the post-processed states at each node. Finally, define $\Delta^u \coloneqq \sigma_p^u - \sigma_q^u$. Then it follows that
\begin{equation}
    \mathbb{E}_q[(L(s|u) - 1)^2] \leq \mu^{-1} \tr\big((\Delta^u)^2\big)
\end{equation}
for all nodes $u$ in $\mathcal{T}$. As such, for success probability at least $2/3$, $\mathcal{T}$ must have depth at least $\Omega\left(\mu/ \tr\big((\Delta^u)^2\big)\right)$.
\end{lemma}
\begin{proof}
    At any node $u$ with POVM $\{F_s\}_u$, we have
    \begin{align}
        \mathbb{E}_q[(L(s|u) - 1)^2] &= \sum_s \tr(F_s \sigma_q^u)\left(\frac{\tr(F_s\sigma^u_p) - \tr(F_s\sigma_q^u)}{\tr(F_s\sigma_q^u)}\right)^2 \\
        &=\sum_s \tr(F_s \sigma_q^u)\left(\frac{\tr(F_s\Delta)}{\tr(F_s\sigma_q^u)}\right)^2 
        \\
        &=\sum_s \frac{\tr(F_s\Delta^u)^2}{\tr(F_s\sigma_q^u)} \leq \mu^{-1} \sum_s \frac{\tr(F_s\Delta^u)^2}{\tr(F_s)}
    \end{align}
    where in the last step we use $\sigma_q^u \succeq \mu \mathds{1}$. Rewriting $\tr(F_s\Delta^u)^2 = \Tr(\sqrt{F_s}\sqrt{F_s}\Delta^u)^2$ and applying Cauchy-Schwarz in the Hilbert-Schmidt inner product (noting that $F_s$ and $\Delta$ are both Hermitian), we have
    \begin{equation}
        \tr(F_s\Delta^u)^2 = \tr(\sqrt{F_s}\sqrt{F_s}\Delta^u)^2 \leq \tr(F_s)\tr(F_s(\Delta^u)^2) \ .
    \end{equation}
    Substituting, we obtain 
    \begin{equation}
        \mathbb{E}_q[(L(s|u) - 1)^2] \leq \mu^{-1}\sum_s \tr(F_s(\Delta^u)^2) = \mu^{-1}\tr((\Delta^u)^2) \ .
    \end{equation}
    Applying Lemma \ref{lemma:martingale} completes the proof.
\end{proof}

\noindent Theorem \ref{thm:Heisenberg} can be used to incorporate constraints like noise into tight learning lower bounds. The proof strategy is composed of three steps. Given a many-vs.-many distinguishing problem between states $\rho$ and $\sigma$ with additional constraints such as noise: 
\begin{enumerate}
    \item  Find a representation of the learning problem as one-vs.-one hypothesis testing, absorbing adaptive and noisy control into channels $\Phi_u$ and defining $\sigma_p^u = \Phi_u[\rho], \sigma_q^u = \Phi_u[\sigma]$.
    \item Find a uniform upper bound on $\tr\big((\Delta^u)^2\big)$ and a operator lower bound on $\sigma_q^u$.
    \item Apply Lemma \ref{thm:Heisenberg} to obtain a sample complexity lower bound.
\end{enumerate}
The first stage of this method can vary from problem to problem. Ref. \cite{cotler2026noisyquantumlearningtheory} begins with the path-based approach for many-vs.-one distinguishing, then rewrites the full-path likelihood ratio as a likelihood ratio of distinguishing two \textit{specific} states, removing the ensemble-dependance at each node. In this work, we use results from \cite{ye2025exponentialadvantagereplicaestimating} that achieve the same effect using direct computations that leverage structure within the hypothesis ensembles. The second step is often simple, since the channels $\Phi_u$ are almost always unital so that it is easy to obtain expressions of the form $\mu \mathds{1} \preceq \sigma_q^u$, while the $\Delta^u$ operator can be cumbersome but is a concrete operator whose squared trace can be manually calculated. This is precisely the power of the Heisenberg method: by absorbing complex constraints into a redefinition of the learning tree task, the challenge of obtaining a tight lower bound reduces to computation of simple operator traces.

\section{Uploading Protocol}
\label{app:uploading_proofs}

In this section we introduce the ``one-step'' state injection procedure and prove how it can be used to couple physical states to a fault-tolerant computation.  We use the composable fault-tolerance framework introduced in \cite{He2025ComposableFaultTolerance}.
\subsection{Definitions and lemmas}

We briefly review the composable fault-tolerance framework of~\cite{He2025ComposableFaultTolerance}, keeping only the ingredients needed for the uploading proof. The main object is a \emph{fault family}: a collection of fault sets whose occurrence may cause a gadget to fail. A gadget is certified by showing that every fault set outside this family produces the intended logical action and a controlled, correctable deviation compatible with subsequent gates. We begin with this abstract language and then specialize it to surface-code decoding.

\begin{definition}[Fault locations and fault sets]
For a fixed circuit gadget, let $\Omega$ denote the finite set of elementary fault locations. A location may correspond, for example, to a data-qubit error, a measurement error, an idle error, or any other elementary event allowed by the chosen noise model. A fault set is a subset $X\subseteq \Omega$, and its weight is its cardinality,
\begin{equation}
|X| = \#\{\omega\in X\}.
\end{equation}
Under local-stochastic noise of rate $p$, the probability that a particular fault set $X$ occurs is bounded by $p^{|X|}$.
\end{definition}
\begin{definition}[Avoiding set, Definition 2.1 of Ref.~\cite{He2025ComposableFaultTolerance}]
\label{def:avoiding_set}
Given a set of fault locations $\Omega$ and a family $\mathcal{F}\subseteq 2^\Omega$, a fault set $X\subseteq \Omega$ is $\mathcal{F}$-avoiding if it contains no bad set:
\begin{equation}
F \not \subseteq X \quad \textnormal{for every } F \in \mathcal{F}.
\end{equation}
\end{definition}

\noindent Avoiding sets capture the combinations of errors which can lead to logical errors and problematic deviations in the output of our gadgets. We will argue that as  long as error sets within the avoiding set do not occur, our gadget will have output whose deviation from a logical codeword does not propagate through a fault-tolerant circuit. Next, we capture the distribution of errors in our gadgets using the concept of the weight enumerator.

\begin{definition}[Weight enumerator, Definition 2.3 of Ref.~\cite{He2025ComposableFaultTolerance}]
\label{def:weight_enumerator}
For a family of bad fault sets $\mathcal{F}\subseteq 2^\Omega$, the weight
enumerator is
\begin{equation}
\mathcal{W}(\mathcal{F},x) = \sum_{F\in\mathcal{F}} x^{|F|} = \sum_{w\geq 0} A_w x^w,
\end{equation}
where $A_w$ is the number of bad sets $F\in\mathcal{F}$ of weight $w$.
\end{definition}
\noindent We will use the operation $\boxplus$ to indicate taking the union of avoiding sets, including ones in different gadgets. It follows that for two avoiding sets, $\mathcal{W}(\mathcal{F}_1\boxplus \mathcal{F}_2, x) \leq \mathcal{W}(\mathcal{F}_1, x) + \mathcal{W}(\mathcal{F}_2, x)$. Next, we need a formal way to define how far the output of a logical gadget is from the input.
\begin{definition}[Deviated state, Definition 4.1 of Ref.~\cite{He2025ComposableFaultTolerance}] \label{def:deviated_state} Let $A$ be a set of physical qubits and let $B\subseteq 2^A$ be a family of bad supports. A state $\tilde{\sigma}$ on $A$ is $B$-deviated from $\sigma$ if there exists a $B$-avoiding support $S\subseteq A$ and a quantum channel $\mathcal{E}_S$ supported on $S$ such that $\tilde{\sigma} = \mathcal{E}_S(\sigma)$. \end{definition}
We will usually take $\sigma$ to be an ideal encoded state and interpret the support $S$ as the residual physical error left after decoding. With these tools, we can state the definition of a fault-tolerant gadget.

\begin{definition}[Fault-tolerant gadget informal, definition 4.8  of Ref.~\cite{He2025ComposableFaultTolerance}]
\label{def:circuit_gadget}
     A circuit $C$ with avoiding set $\mathcal{F}_G$ is a fault-tolerant gadget for the logical gate $G$ mapping from input code type ($U_{\mathrm{in}}$, $B_{\mathrm{in}}$) to output code type ($U_{\mathrm{out}}$, $B_{\mathrm{out}}$) if
    \begin{enumerate}
        \item The gadget accurately simulates the action of the gate: for every fault which is $\mathcal{F}_G$-avoiding $C$ maps $B_{\mathrm{in}}$-deviated encodings of $\rho$ to $B_{\mathrm{out}}$-deviated encodings of $G(\rho)$.
        \item The gadget is friendly:  for every fault which is $\mathcal{F}_G$-avoiding regardless of the input to the gadget the output is $B_{\mathrm{out}}$-deviated from some logical state.  
    \end{enumerate} 
    Recall that in the tuples for code type, $U$ is the ideal decoding unitary and $B$ the set capturing allowable deviation.  
\end{definition}
\noindent The crucial property of this definition is the ability to compose gadgets and directly obtain fault-tolerance guarantees for the composite circuit.
\begin{theorem}[Gadget composition, theorem 4.9  of Ref.~\cite{He2025ComposableFaultTolerance}]
\label{thm:gadget_composition}
      Consider fault tolerant gadgets $G_1$ and $G_2$ with avoiding sets $\mathcal{F}_1$ and $\mathcal{F}_2$.  Parallel composition gives a new fault tolerant gadget with avoiding set $\mathcal{F}_1 \boxplus \mathcal{F}_2$.  If one sequentially composes the gadgets then the resulting gadget $G_2 G_1$ is also fault-tolerant if the output type of the first gadget is compatible with the input type of the second gadget.  Again the avoiding set of the combined gadget is $\mathcal{F}_1 \boxplus \mathcal{F}_2$.  Friendliness is generically preserved for parallel composition and is preserved for sequential composition if the first gadget is friendly.  
\end{theorem}
\noindent We now specialize this abstract gadget language to CSS codes and surface-code decoding. A CSS code has check matrices $H_X$ and $H_Z$ which detect $Z$ and $X$ errors, respectively. Because the two Pauli sectors decouple, error correction can be analyzed independently in the $X$ and $Z$ sectors. In each sector, the decoding problem is represented on a graph whose edges are elementary error locations and whose vertices are detection events. The decoder then selects a minimum-weight set of edges consistent with the observed syndrome.  The graph corresponding to a check matrix $M$ is denoted by $G[M]$.  We assume a minimum weight decoder defined below.
\begin{definition}[Min weight decoder]
    \label{def:min_weight_decoder}
Given $a\in \{X,Z\}$, a set of detection events $s^a \in C_0^a$, and a mapping $\partial$ from error to syndrome, the min weight decoder returns a configuration of errors $\hat{e}^a\in C_1^a$ such that $\partial \hat{e}^a = s^a$ and $\hat{e}^a = \operatorname{argmin}_{e^a} \{|e^a|\ |\ \partial e^a = s^a\}$.  In other words it returns the minimum weight error configuration consistent with the syndrome.  
\end{definition}
\noindent To bound the probability of decoder failure, we count connected clusters of fault locations. We use these cluster-counting estimates to bound weight-enumerators for fault sets.

\begin{definition}
[Clustering sets: 6.4  of Ref.~\cite{He2025ComposableFaultTolerance}]
\label{def:clustering_sets}
For a graph $G = (V,E)$ the (d,t)-clustering sets, denoted $CG_{d,t}\subseteq 2^V$, are given by 
    \begin{align}
        CG_{d,t} = \{W \subseteq U \subseteq V \mid |U| = d, U\text{ is connected in }G, |W| = t\}
    \end{align}
\end{definition}
\noindent and cluster covers as 
\begin{definition}[Cluster cover: 6.5  of Ref.~\cite{He2025ComposableFaultTolerance}]
\label{def:cluster_cover}
     For a graph $G = (V,E)$ and a subset of the vertices $S \subseteq V$ a cluster cover $CO(m, S)$ is a set of vertices $C = \bigcup_i C_i \subseteq V$ where for every $C_i$ we have that $C_i \cap S \neq \emptyset$ and $|C|=m$.
\end{definition}
\noindent We need two lemmas which allow us to count the number of connected clusters and cluster covers for a graph.  
\begin{lemma}[Cluster cover count, 6.5  of Ref.~\cite{He2025ComposableFaultTolerance}]
\label{lemma:cluster_cover_count}
     The number of cluster covers $CO(m,S)$, denoted by $M(m,S)$ is bounded as
    \begin{align}
        M(m,S) \leq e^{|S|-1} (\Delta e)^{m-|S|}
    \end{align}
\end{lemma}
\noindent Note that in the case where $|S|=1$ this counts the number of connected clusters with a bound given by $(\Delta e)^m$.  Next using this we can get a bound on the number of possible $CG_{d,t}$ for a given graph and hence the weight enumerator  
\begin{lemma}[Weight enumerator bound, 6.6  of Ref.~\cite{He2025ComposableFaultTolerance}]
\label{lemma:weight_enumerator}
 The weight enumerator of $CG_{d,t}$ for a graph $G= (V,E)$ with maximum degree $\Delta$ is bounded by
    \begin{align}
        \mathcal{W}(CG_{d,t}, x) \leq |V| \binom{d}{t-1} (\Delta e)^{d-1}x^t.
    \end{align}
\end{lemma}
\noindent Lastly we need to define the space-time version of the code that we use.  Here we want to capture the fact that we perform repeated stabilizer measurements and form detectors comparing the value of measurements in subsequent rounds.  We provide a brief informal version and the full version should be found in \cite{He2025ComposableFaultTolerance}
\begin{definition}[Space-time code, 6.14  of Ref.~\cite{He2025ComposableFaultTolerance}]
\label{def:space_time_code}
     Recall that the spatial code $H_Z\in \mathbb{F}_2^{r_Z \times n}$ has $r_Z$ parity checks on n qubits.  The space-time code $H\in \mathbb{F}_2^{Tr_Z \times (T n+T r_Z)}$ acts on a vector $x$ of length $T n + T r_z$ where the $T n$ corresponds to the data qubit errors in each round and $T r_Z$ corresponds to the measurement errors in each of the $T$ rounds.  The space-time code captures that the syndrome flips in two subsequent rounds for a measurement error and once for a data qubit error in a given round.  The matrix $H$ should be understood as mapping data qubit and measurement errors to detectors which are consistency checks between rounds of stabilizer measurement.  
\end{definition}
\noindent We next introduce notation for restrictions of a spacetime error. A spacetime error is a vector
\begin{equation}
x\in \mathbb{F}_2^{Tn+Tr_Z},
\end{equation}
whose coordinates are partitioned into data-error coordinates $J_d^1,\ldots,J_d^T$ and measurement-error coordinates $J_s^1,\ldots,J_s^T$. We write $x[J_d^t]\in\mathbb{F}_2^n$ for the restriction of $x$ to data-qubit errors in round $t$, and $x[J_s^t]\in\mathbb{F}_2^{r_Z}$ for the restriction to measurement errors in round $t$. All sums below are over $\mathbb{F}_2$. The rows of $H$ are denoted $I_1,\ldots,I_T$ and correspond to detectors comparing adjacent syndrome-measurement rounds.
\begin{definition}[Flattening map]
\label{def:flattening_map}
The flattening map sends a spacetime error to the net spatial data error at the
output of the gadget:
\begin{equation}
\phi(x) = \sum_{t=1}^T x[J_d^t]\in \mathbb{F}_2^n.
\end{equation}
\end{definition}
\noindent We also use the extended spacetime complex with an additional output data-error slice $J_d^{T+1}$ and a corresponding final detector layer $I_{T+1}$.

We will need several lemmas that allow us to reason about the space-time complex.  Firstly we need the following lemma which allows us to reason about the space-time support of a fault pattern corresponding to some residual error 
\begin{lemma}[6.19  of Ref.~\cite{He2025ComposableFaultTolerance})]
\label{lemma:residual_error}
     Consider a bit-string $x\in \mathbb{F}^{Tn + T r_z}$ with space-time syndrome $H x = 0$ (this includes all the checks that are included in the bulk).  Let $y\subseteq J_{T+1}^{d}$ be the minimum weight operator that agrees with the projected error, formally $y = \operatorname{argmin}_{y: H_Z \phi(x) = H_Z y} |y|$.  If $y$ contains some set $S\subseteq J_{T+1}^d$, then there must exist a cluster cover of $R$ of $S$, $R\subseteq x\sqcup y$, such that for $|R| \geq 2|S|$ and for every cluster $C$ of $R$, we have $|x \cap C|\geq |C|/2$.
\end{lemma}
\noindent Additionally we need
\begin{lemma}[6.20  of Ref.~\cite{He2025ComposableFaultTolerance}]
\label{lemma:error_to_large_connected_cluster}
     Consider a bit-string $x\in \mathbb{F}^{Tn + T r_z}$ with space-time syndrome $H x = 0$ (this includes all the checks that are included in the bulk).  Let $y\subseteq J_{T+1}^{d}$ be the minimum weight operator that agrees with the projected error, formally $y = \operatorname{argmin}_{y: H_Z \phi(x) = H_Z y} |y|$.   If $\phi(x) -y \notin \operatorname{rowsp}(H_X)$ then there is a connected cluster of vertices $S$ of size at least $d$ such that $|x\cap S|\geq |S|/2$.
\end{lemma}
\noindent The surface code family is defined as follows.
\begin{definition}[Surface-code type]
\label{def:surface_code}
We use $SC(d)$ to denote the rotated surface code~\cite{Bombin2007OptimalResources, Horsman_2012} with
parameters $[[d^2,1,d]]$. We write
\begin{equation}
SC(d)_c = (U,B_c)
\end{equation}
for the corresponding code type. Here $U$ is the ideal purified decoding unitary, and $B_c$ is the family of residual supports regarded as too large to be safely passed to later gadgets. Explicitly,
\begin{equation}
B_c = B_c^{(X)}\cup B_c^{(Z)},
\end{equation}
where
\begin{equation}
B_c^{(X)}=\bigcup_{m=d}^{d^2} CG_{m,\lceil cm/2\rceil}^{(X)}, \qquad B_c^{(Z)}=\bigcup_{m=d}^{d^2} CG_{m,\lceil cm/2\rceil}^{(Z)}.
\end{equation}
\end{definition}
\noindent Intuitively, $SC(d)_c$ allows small, spatially clustered residual errors but excludes connected clusters of size at least the code distance that contain a sufficiently large fraction of faults. Such clusters are precisely the objects that can become dangerous when gadgets are composed.

Going forward, for our surface code gadgets we choose the residual family to be $B^{(R)}=B_{1/5}$ for both the $X$ and $Z$ sectors.  A surface code error correction gadget can be shown to have the following properties
\begin{lemma}[Surface code error correction gadget, 7.3  of Ref.~\cite{He2025ComposableFaultTolerance}]
\label{lemma:surface_code_ec_gadget}
     There is a fault-tolerant error correction gadget $EC^{SC(d)}$ with avoiding set $\mathcal{F}_{EC}^d$ which implements the identity and has the input code type $SC(d)_{1/2}$ and output code type $SC(d)_{1/5}$.  The weight enumerator of the gadget can be bounded as 
    \begin{align}
        \mathcal{W}(\mathcal{F}_{EC}^d, x) \leq \operatorname{poly}(d) e^{-\beta d}
    \end{align}
    for $x \leq \epsilon_*$.  
\end{lemma}
\noindent Furthermore, \cite{He2025ComposableFaultTolerance} presents a universal set of operations which can be used to realize Clifford + T computation. We now present the details relevant to this work.
\begin{lemma}[Surface code universal gateset gadgets, 7.3, 7.4, 7.10, 7.11  of Ref.~\cite{He2025ComposableFaultTolerance}]
\label{lemma:surface_code_universal_gadgets}
     There exists $\epsilon_* > 0$ such that every operation in a universal Clifford + T gateset can be implemented by a friendly surface code gadget with input/output code types $SC(d)_{1/5}$ and avoiding set $\mathcal{F}_g^d$ which satisfies 
    \begin{align}
        \mathcal{W}(\mathcal{F}_g^d, x) \leq \operatorname{poly}(d)e^{-\beta d}
    \end{align}
    for all $x\in [0,\epsilon_*]$.  The overhead of the gadgets is polylogarithmic in the circuit size and target error rate once $d$ is chosen so that $\operatorname{poly}(d)e^{-\beta d}\leq \epsilon$ for a target error probability $\epsilon$.  
\end{lemma}
\noindent Using this universal set of gadgets, Ref.~\cite{He2025ComposableFaultTolerance} proves the following theorem concerning the fault-tolerant implementation of a Clifford + T circuit.  
\begin{theorem}[Surface code threshold theorem 7.11 from \cite{He2025ComposableFaultTolerance}]
\label{thm:basic_surface_code_threshold}
    . There exists a constant $\epsilon_*\in (0,1)$ such that for any Clifford + T circuit $C$ of width $W$ and depth $D$ and any $\epsilon \in (0,1)$, there exists a circuit $\bar{C}$ that is a fault-tolerant gadget for $C$ with weight enumerator $\mathcal{W}$.  The width $\bar{W}$ and depth $\bar{D}$ of $\bar{C}$ are at most $polylog(V)$ larger than the original circuit $C$ where $V = W D/\epsilon$.  For $x \in [0,\epsilon_*]$ the weight enumerator has the bound $\mathcal{W}(\mathcal{F}, x) \leq \epsilon$.  
\end{theorem}
\noindent This theorem can be applied to circuits where the input types are $SC(d)_{1/5}$. 

\subsection{Surface code quantum uploading}
In this section we establish the growth procedure for quantum uploading via the surface code and establish that the procedure can be used interface physical qubits with  universal fault-tolerant quantum circuits. Surface code growth has mostly been studied in the context of magic state injection, enabling  fault-tolerant syndrome measurement of an outer error correcting code \cite{Li2015MagicStateFidelity}.  This is treated more formally in Ref.~\cite{He2025ComposableFaultTolerance}, where a growth procedure that increases the code distance iteratively is introduced and coupled to other fault-tolerant gadgets. Ref. ~\cite{Christandl2024QuantumIO} establishes more general and abstract conclusions for connecting these input-output gadgets to a fault tolerant computation. In the present work, our goal is to provide a fault-tolerance proof for the more practical style of growth gadget presented in \cite{Li2015MagicStateFidelity}. Our guarantees will produce a quadratically more depth-efficient gadget than the one presented in Ref.~\cite{He2025ComposableFaultTolerance}while retaining the relevant fault-tolerance guarantees.

We will now provide a detailed technical description of how the growth gadget is realized in practice, including how stabilizer measurement data obtained during the growth process is used for decoding. Our discussion necessarily utilizes technical terminology specialized to syndrome measurement and decoding of the surface code, and readers unfamiliar with this language or primarily interested in the implications that our growth gadget has for downstream implementation of fault-tolerant learning circuits can skip to Theorem \ref{thm:arb_input_threshold_theorem}.

We begin with an informal sketch of our technical goal. To rigorously establish fault-tolerance guarantees, we begin by establishing that the logical error of surface code growth decomposes into two terms. The first is a noise floor that depends only on $d_1$ and is bounded by a constant multiple of the physical error rate. The second is a $\operatorname{poly}(d_1, d_2) e^{-\alpha d_2}$ contribution which decays exponentially with the outer code distance. As such, the logical error rate actually decreases as one grows the unknown state to larger code patches, down to an asymptotic floor that is linear in the physical error rate.

Next, we provide a pedagogical description of the code growth gadget, which will lead to a formal definition that we treat in the subsequent proofs. Our growth gadget relies on the rotated surface code from Definition  \ref{def:surface_code}, which is a $[[d^2,1,d]]$ CSS quantum error-correcting code depicted in Fig. \ref{fig:surface_code_growth_schematic}. Data qubits reside on the vertices of a two-dimensional square lattice. The codespace is the common $+1$ eigenspace of the stabilizer subgroup $\langle S_1, \cdots S_{d^2-1}\rangle$. This stabilizer subgroup is generated by the white and grey plaquettes in Fig. \ref{fig:surface_code_growth_schematic}, which correspond to $Z$ and $X$ stabilizers respectively.

The logical operators of this code are Pauli operators which commute with the stabilizers but are not in the stabilizer group ($N(S)\setminus S$). These are $X$ and $Z$ logical strings which connect the horizontal and vertical boundaries respectively. For a rotated surface code of a given distance $d$, we choose a representative $X$ logical operator denoted by $X_L^{(d)}$, which is a Pauli string connecting the X boundaries. Similarly, we select a representative $Z$-logical operator $Z_L^{(d)}$ which connects the $Z$ boundaries.  These are indicated by the colored strings in Fig.~\ref{fig:surface_code_growth_schematic}(a).  We will henceforth refer to the topmost boundary with two-body $X$ stabilizers as the $X$ boundary and the leftmost boundary with two-body $Z$ stabilizers as the $Z$ boundary.  
\begin{figure}
    \centering
    \includegraphics[width=\linewidth]{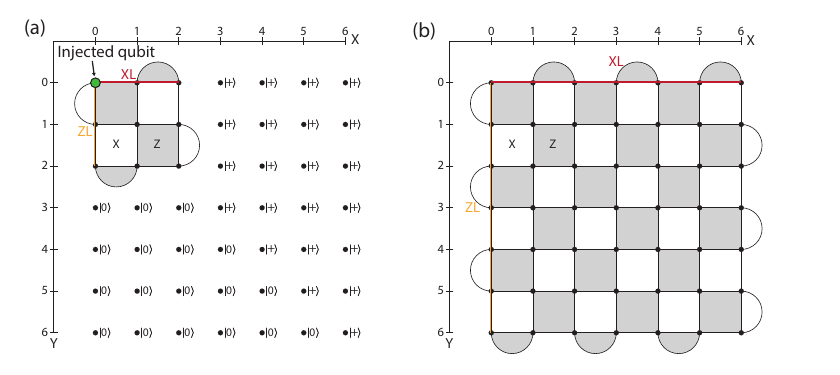}
    \caption{\textbf{Surface code growth schematic: }Qubits are located on the vertices and stabilizers are the plaquettes.  X stabilizers are indicated in gray and Z stabilizers in white.  In (a) we show the distance-3 surface code before the growth along with the initial states of the new qubits.  The states are chosen to ensure that no logical information is learned in the process of growing the surface code.  This figure focuses on growth of a distance-3 code to a distance-7 code but in practice one would use a constant size initial encoding procedure to map a single physical qubit (shown in green) into the logical state of the distance-$d_1$ code before growing to the higher distance code. In (b) we show the distance-7 stabilizers after the growth.  A representative is indicated in red for the $X$ logical and yellow for the $Z$ logical.  The $X$, $Z$ logicals are formed by physical $X$ and $Z$ on the support indicated in the figure.}
    \label{fig:surface_code_growth_schematic}
\end{figure}

With the goal of growing from a single physical qubit to a distance $d_2$ surface code in mind, our first step is use a known, constant-size circuit to prepare a distance $d_1$ surface code, where $d_1$ is a fixed constant. In practice, this can be achieved with a unitary encoding circuit or stabilizer measurements and gauge-fixing, and related examples are provided in Refs.~\cite{DennisKitaevLandahlPreskill2002TopologicalMemory, Li2015MagicStateFidelity}.

After this initial circuit, the surface code has $d_1^2$ qubits with positions labeled by $x,y$ coordinates with $0 \leq x,y < d_1$ (Fig. \ref{fig:surface_code_growth_schematic} depicts the case of $d_1=3$).  The initial code further has $d_1^2-1$ stabilizer generators.  To begin the growth procedure we are thus given a distance-$d_1$ surface code in an arbitrary encoded logical state $\rho_L$.  Next, the remaining qubits with coordinates $d_1 \leq x < d_2$ and/or $d_1 \leq y < d_2$ are added into the picture. All qubits closer to the $Z$ boundary are initialized in the state $|0\rangle$, while qubits nearer to the $X$ boundary are initialized in the state $|+\rangle$. 

At timestep $t = 0$ we begin measuring the stabilizers of the larger code yielding measurement values $m_{s,t}$ for stabilizer $S$ at timestep $t$. The boundary between growth qubits in $|0\rangle$ and $|+\rangle$ states lies along the planar diagonal as shown in Fig. \ref{fig:surface_code_growth_schematic}.  Crucially, the first measurements outcomes of new stabilizers bordering the original code are random.  This ensures that no logical information is learned during the growth procedure. These initial random measurement outcomes are treated as gauge data and are used to update the Pauli frame on the grown code after growth has completed, as discussed shortly. 

We perform the measurements of the grown code up to timestep $T \geq d_2$. At this stage, we have obtained a quantum state on the larger code patch. We must next discuss how intermediate stabilizer measurement outcomes are decoded in order to determine the decoder frame update and corresponding logical action. Moreover, we wish to understand how this decoding procedure affects an arbitrary input state. To accomplish this, we formalize the spacetime description of our growth procedure and define detectors. 

A detector, associated with a vertex in the spacetime, is a comparison between deterministic stabilizer measurements in subsequent measurement rounds. Concretely, the detector is $m_{s,t-1} \oplus m_{s,t}$. In the growth step at $t=0$, measurements of new stabilizers can be deterministic or random depending on the initial state of the data qubit relative to the stabilizer. Only deterministic measurements form detectors, while other measurements serve as gauge data. This definition of the space-time is closely related to the spacetime check matrix $H$ defined in  \ref{def:space_time_code}. In an error-free run of the circuit, all stabilizer measurements would be consistent and all of the detection events would return $0$. 

In practice, the circuit will be noisy, resulting in nontrivial stabilizer outcomes. In all stages of the circuit we consider errors on both measurements and data qubits (so-called phenomenological noise). To treat these, we partition the $X$ and $Z$ Pauli errors into two corresponding sectors $a\in \{X,Z\}$ and define $\Omega^a$ as the set of possible error locations in sector $a$. We define $C_1^a\coloneq \mathbb{F}_2^{\Omega^a}$. Then, any error pattern is specified by a vector $e \in C_1^a$. 

Similarly, defining $D^a$ to be the set of all possible detectors in sector $a$, we define $C_0^a \coloneq \mathbb{F}_2^{D^a}$ such that any syndrome configuration is described by a vector $s \in C_0^a $. Using our definition of a detection event, we define the mapping from errors to detection events by $\partial: C_1^a \rightarrow C_0^a$. As such, $\partial$ outputs the detection events induced by any choice of error.  Moreover, we note that any measurement error or data-qubit error will cause either a single or pair of detection events to be non-zero. To depict these, we construct a graph where the vertices correspond to one of the detection locations.  For every possible error, we add an edge connecting the two detection locations which that error causes to be non-zero. If the error only causes one non-zero detector, we add an edge to a boundary vertex. An example of the graph is shown in Fig.~\ref{fig:error_examples}(b). Sampling physical errors can be thought of as sampling edges from the error correction graph; in our case we assume all edges have equal probability as errors are independently and identically distributed at all locations and times.

The decoder is provided the detector information $s^a \in C_0^a$ and returns the minimum-weight error configuration $\hat{e}^a \in C_1^a$ which can generate such a syndrome. More formally, errors in the circuit $e^a$ each correspond to a syndrome $s^a = \partial e_a$. The decoder returns $\hat{e}^a$ with a guarantee that $\partial (e^a + \hat{e}^a) = s^a \oplus s^a = 0$. In the surface code the union of errors and decoder corrections can only form cycles or strings that connect boundaries of the surface code to avoid a non-trivial syndrome (all the detectors have value 0). A decoded logical failure in sector $a\in \{X,Z\}$ is then described by the event of a string connecting two boundaries of the spatial growth.  The properties of the minimum weight decoder we use are formalized by Definition \ref{def:min_weight_decoder}. 

After the growth, we must perform gauge corrections to ensure that the new stabilizers of the grown surface code have $+1$ eigenvalues. This is due to the randomness in the initial growth measurements. We determine the corrected value of the measurement by using the value of the gauge measurements $m_{x,y,0}$ in addition to decoder corrections. If the gauge-corrected measurement,  $\hat{m}_{x,y,0}$, is $1$, we apply a string of $X$ operators to the boundary in order to gauge-fix.  

We summarize the growth gadget with the following definition.

\begin{definition}[Injection gadget]
\label{def:growth_gadget_procedure}
     The procedure $\mathrm{INJECT}_{d}$ defined above takes a physical qubit to an encoded distance $d_2$ surface code, and consists of the following steps.
    \begin{enumerate}
        \item A constant size circuit $\mathrm{ENC}_{d_1}$ is used to prepare a distance $d_1\geq 4$ surface code where $d_1$ is fixed constant. 
        \item Qubits are added according to the pattern shown in Fig.~\ref{fig:surface_code_growth_schematic}.  
        \item At timestep $t=0$, stabilizers of the distance-$d$ surface code ($d$=$d_2$) begin to be measured.  
        \item $T\geq d$ total rounds of measurement and syndrome extraction are performed, and the spacetime is decoded according to the maximum-likelihood decoder.  
        \item Gauge corrections are applied based on the inferred value of growth syndrome measurements. 
    \end{enumerate}
    Overall, we have that 
    \begin{align}
        \mathrm{INJECT}_{d} = \mathrm{GROW}_{d_1 \rightarrow d} \circ \mathrm{ENC}_{d_1}
    \end{align}
    where $\mathrm{GROW}_{d_1 \rightarrow d}$ captures all steps after the initial constant-size encoding circuit.
\end{definition}

Having established the growth procedure, we turn to the proofs of fault-tolerance guarantees. We first prove that the injection circuit acts as an exact surface code encoding circuit in the absence of noise.
\begin{lemma}
\label{lemma:surface_code_growth_noiseless_correctness}
    In the absence of errors the surface code growth procedure \ref{def:growth_gadget_procedure} acts as the logical identity.
\end{lemma}
\begin{proof}
    Consider a qubit $a$ coupled to an $n-1$-qubit reference system $R$, with the joint state described by $\rho_{Ra}$. We will consider injection of qubit $a$ into a distance-$d$ surface code; by changing the label $a$ and treating the other $n-1$ qubits as $R$, this argument shows that applying the injection gadget qubit-by-qubit to an arbitrary $n$-qubit state performs the intended encoding.
    
    Let $E$ denote the isometry consisting of the initial distance-$d_1$ encoding circuit together with the preparation of the growth ancillas in the states $|0\rangle^{A_0}$ and $|+\rangle^{A_+}$. Let $\Pi_{d_1}$ be the projector onto the distance-$d_1$ codespace and define the initialized input projector
    \begin{equation}
        \Pi_{\rm in}
        =
        \Pi_{d_1}
        \left(\prod_{j\in A_0}\frac{I+Z_j}{2}\right)
        \left(\prod_{j\in A_+}\frac{I+X_j}{2}\right),
    \end{equation}
    where $\Pi_{d_1}$ acts trivially on the added ancillas. Then $\Pi_{\rm in}E=E$. Moreover, if $P_a^{(d_1)}$ denotes the logical Pauli on the initial distance-$d_1$ patch, then by definition of the initial encoding circuit, $P_a^{(d_1)}E = E P_a$ for $P_a\in\{X,Z\}$ (suppressing ancillas). Let $M_g$ denote the projector onto a branch $g$ of the stabilizer measurements during growth, before gauge correction. Then
    \begin{equation}
        M_g=\prod_{S\in \mathcal{S}_{\rm grow}}\frac{I+g_S S}{2},
    \end{equation}
    where the $S$ are stabilizers of the final distance-$d$ code and $g_S\in\{\pm1\}$ are the measured signs on that branch. We choose final canonical logical representatives $\bar X_a^c$ and $\bar Z_a^c$ so that they extend the initial logical strings through qubits prepared in the appropriate eigenbasis (for example, the orange and red canonical logicals in Fig.~\ref{fig:surface_code_growth_schematic}). Thus
    \begin{equation}
        \bar X_a^c = X_a^{(d_1)}\prod_{j\in B_X}X_j,
        \qquad
        B_X\subseteq A_+,
    \end{equation}
    and
    \begin{equation}
        \bar Z_a^c = Z_a^{(d_1)}\prod_{j\in B_Z}Z_j,
        \qquad
        B_Z\subseteq A_0.
    \end{equation}
    It follows from the definition of $\Pi_{\rm in}$ that
    \begin{equation}
        \bar P_a\Pi_{\rm in}=P_a^{(d_1)}\Pi_{\rm in}
    \end{equation}
    for $P\in\{X,Z\}$. Since $\bar P_a^c$ is a logical operator of the final code, it commutes with every stabilizer appearing in $M_g$, and hence
    \begin{align}
        \bar P_a^c M_g\Pi_{\rm in}
        &=
        M_g\bar P_a^c\Pi_{\rm in} \nonumber \\
        &=
        M_g P_a^{(d_1)}\Pi_{\rm in}.
    \end{align}

    At this point we have shown that for one canonical choice of logicals, the entertwining relation applies. In order to show that we have mapped the physical state into a logical state we need to show that the same relation holds for all logical representatives, or equivalently, that we can work in a Pauli frame in which all stabilizers have $+1$ eigenvalues. The gauge correction step accomplishes this purpose. We let $\bar{P}_a$ be any logical representative given by $\bar{P}_a = \bar{P}_a \prod_{i\in C} S_i$ for some subset of stabilizers $C \in S$.  We would like to show that $\bar P_a C_g M_g\Pi_{\rm in}= C_g M_g P_a^{(d_1)}\Pi_{\rm in}$.  The gauge correction $C_g$ is a string of Pauli operators which commutes with $\bar P_a^c$ (the canonical logicals) 
    \begin{equation}
        [C_g,\bar X_a^c]=[C_g,\bar Z_a^c]=0.
    \end{equation}
    and has the effect of mapping all of the stabilizers to have $+1$ eigenvalues:
    \begin{align}
        S_i C_g M_g \Pi_{in} = C_g M_g \Pi_{in}.
    \end{align}
    Hence in total we get that
    \begin{align}
        \bar P_a C_g M_g\Pi_{\rm in}
        &=
        \bar P_a^c \prod_{i} S_i C_g M_g\Pi_{\rm in} \nonumber \\
        &=
        \bar P_a^c C_g M_g\Pi_{\rm in} \nonumber \\
        &=
        C_g\bar P_a^c M_g\Pi_{\rm in} \nonumber \\
        &=
        C_g M_g P_a^{(d_1)}\Pi_{\rm in}.
    \end{align}
    as desired. Composing on the right with $E$ and using $\Pi_{\rm in}E=E$ and $P_a^{(d_1)}E=EP_a$, we obtain
    \begin{equation}
        \bar P_a C_gM_gE=C_gM_gEP_a.
    \end{equation}
    The same identity for $P_a = Y$ follows from $Y=iXZ$. Thus, for every branch Kraus operator $K_g=C_gM_gE$, we have
    \begin{equation}
        \bar P_a K_g=K_gP_a
    \end{equation}
    for every single-qubit Pauli $P_a$ and all logical representatives $\bar{P}_a$. To conclude, let the output state of the injection gadget be
    \begin{equation}
        \bar{\rho}_{RA}
        =
        \sum_g (I_R\otimes K_g)\rho_{Ra}(I_R\otimes K_g^\dagger).
    \end{equation}
    For any Pauli $P_RP_a$ on $Ra$, the branchwise intertwining relation implies
    \begin{align}
        \operatorname{tr}\left(P_R\bar P_a\bar{\rho}_{RA}\right)
        &=
        \sum_g
        \operatorname{tr}\left((P_R\otimes \bar P_a)(I_R\otimes K_g)\rho_{Ra}(I_R\otimes K_g^\dagger)\right) \nonumber \\
        &=
        \sum_g
        \operatorname{tr}\left((I_R\otimes K_g)(P_R\otimes P_a)\rho_{Ra}(I_R\otimes K_g^\dagger)\right) \nonumber \\
        &=
        \operatorname{tr}\left((P_R\otimes P_a)\rho_{Ra}\left(I_R\otimes \sum_g K_g^\dagger K_g\right)\right) \nonumber \\
        &=
        \operatorname{tr}\left(P_RP_a\rho_{Ra}\right),
    \end{align}
    where the final equality uses completeness of the measurement branches. Since Pauli operators form a basis, the output is exactly the distance-$d$ encoding of the input state, and the noiseless growth procedure acts as the logical identity.
\end{proof}

\noindent We next argue that when noise is present, logical errors and uncontrolled residuals can be controlled by a careful definition of a spacetime location-dependent avoiding set. Our next building block result addresses one sector of the decoding problem.
\begin{theorem}[Single sector decoding]
\label{thm:single_sector_decoding}
For the growth procedure outlined above and a collection $B^{(R)}$ of bad residual supports, there exists a family of bad fault paths $\mathcal{F}$ such that for $x = e + \hat{e}$ corresponding to the combination of physical and decoded errors and $y$ corresponding to the minimum weight operator satisfying $H_Z y = H_Z \phi(x)$ we have
    \begin{enumerate}
        \item If the bulk error $e$ is $\mathcal{F}$-avoiding, then $y$ is $B^{(R)}$-avoiding.
        \item The output has no logical errors, namely $\phi(x)+y\in \operatorname{rowsp}(H_X)$.  
    \end{enumerate}
    The weight enumerator polynomial for the avoiding set $\mathcal{F}$ is given by
    \begin{align}
        \mathcal{W}(\mathcal{F}, x) &\leq \mathcal{W}(\mathcal{F}_{\rm in}, x) + \mathcal{W}(\mathcal{F}_1, x) + \mathcal{W}(\mathcal{F}_2, x)\nonumber \\
        &\leq |V_{\rm in}| x 
        + \frac{T d_2^2 \eta_x^{d_2}}{1-\eta_x}
        + \frac{\eta_x^{d_1}}{(1-\eta_x)^2}\Bigg[d_1^2 + \frac{(2 d_1-1) \eta_x}{1-\eta_x} + \frac{2\eta_x}{(1-\eta_x)^2}\Bigg] \nonumber \\
        &\quad+ \frac{1}{1- 2D e x^{1/4}}\mathcal{W}\!\left(B^{(R)}, (2De)^2 x^{1/2}\right)
    \end{align}
    which has a $d_1$ dependent floor and all terms depending on $d_2$ are exponentially decaying in $d_2$.  Moreover we have let $\eta_x = 2De\,x^{1/\chi}$ and we take $\chi=4$, as will be explained shortly.
\end{theorem}
\noindent Before proceeding with the proof, we note the two crucial parts of this theorem at an intuitive level.
\begin{enumerate}
    \item If the gadget does not hit the avoiding set, the output state is not ``too" deviated.
    \item If the gadget does not hit the avoiding set, the output state is the correct logical state.
\end{enumerate}
Importantly, to establish this theorem, we we do not place any constraints on the gadget's input state. Rather, we treat the input as perfect, and group any errors that acted on the input state before the growth gadget into the weight enumerator function. The proof follows the cluster-counting strategy of Ref.~\cite{He2025ComposableFaultTolerance}, but we instead use two closely related graphs. 

The detector graph has vertices corresponding to detection events and edges corresponding to elementary fault locations, and is useful for reasoning about the length of logical error chains. The adjacency graph $G_{\rm adj}[H]$, in contrast, has vertices corresponding to fault locations, with two vertices adjacent when the corresponding faults participate in a common spacetime check. This is useful for applying cluster-counting estimates, since the sets $CG_{m,t}$ are vertex sets. Equivalently, $G_{\rm adj}[H]$ is the line graph of the detector graph. Below, when we count clusters of fault locations, we are referring to clusters in $G_{\rm adj}[H]$. Furthermore, we write $D$ for the maximum degree of this adjacency graph, which is constant for the surface-code spacetime complex.

\begin{proof}
We focus on X errors (detected by Z checks); the problem is symmetric between the $X$ and $Z$ cases.  We will refer to physical errors using the character $e$ and decoder errors interchangeably as $\hat{e}$ and $c$. To begin, we define the avoiding sets for our growth gadget.  The errors $x = e + \hat{e}$ which can cause the growth procedure to have a logical error are 
\begin{enumerate}
    \item A string between the spatial boundaries of the code whose length is at least half comprised by errors.
    \item A string between the growth time boundary and the opposing spatial boundary whose length is at least half comprised by errors.
    \item A string between the growth boundary and the temporal boundary at time $T$ whose length is at least half comprised by errors. This can lead to gauge-fixing errors.  
\end{enumerate}

\begin{figure}
    \centering
    \includegraphics[width=\linewidth]{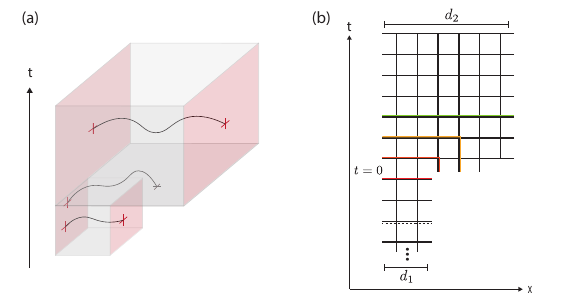}
    \caption{\textbf{Examples of logical errors: } (a) We show schematically the 3D space-time for the growth of an initial surface code to a larger surface code.  Time runs in the vertical direction.  The red boundaries correspond to $X$ boundaries and the strings to possible logical $X$ errors.  We show different ways one can have logical errors including between boundaries of the initial code, boundaries of the grown code, and going between the initial and grown code. (b) Schematic x-t projection of the error correction graph for the surface code growth.  Equivalently this can be viewed as the error correction graph for the growth of a repetition code.  Vertices correspond to detection events and edges to possible errors.  Edges with only one vertex correspond to connections to the boundary.  The growth happens at timestep $t=0$ and the grown stabilizers are measured for $T \geq d_2$.  We assume that all errors that happen before the growth are logical errors. In color are possible logical error strings. }
    \label{fig:error_examples}
\end{figure}
\noindent Examples of strings which cause logical errors are presented in \ref{fig:error_examples} (a).  This decomposition into faults relies on the following lemma which necessitates that a logical error for the surface code is a string whose length is at least half comprised by errors.
\begin{lemma}
    \label{lemma:logical_strings}
    (Logical error strings must comprise at least half the physical faults) Let $e, \hat{e}\in C_1$ satisfying $\partial e = \partial \hat{e}$ and let $z\subseteq e + \hat{e}$ be a nontrivial cycle, i.e. a logical string connecting the boundaries.  Then at least half of the edges in the string corresponding to physical errors $e$ (the other part corresponds to corrections applied by the decoder $\hat{e}$). Formally, 
    \begin{align}
        |z \cap e| \geq |z|/2.
    \end{align}
\end{lemma}
\begin{proof}
    Assume for the sake of contradiction there is an error configuration and nontrivial cycle $z\subseteq e + \hat{e}$ where $|z \cap e| < |z|/2$.  Consider a new choice for the decoder $\hat{e}' = \hat{e} + z$.  We still have that $\partial e = \partial \hat{e}'$ since $\partial z=0$.  But $|\hat{e}'| < |\hat{e}|$ contradicting the fact that $\hat{e}$ was a minimum weight decoding.  
\end{proof}
\noindent  Previous work uses cluster-based avoiding sets for LDPC codes which follow \cite{Kovalev2013BadQLDPC}. In order to be able to couple to other fault-tolerance gadgets which use these definition, we will consider a class of avoiding sets consisting of errors more general than the string-like errors we have discussed. For every spacetime location  $x,y,t$, this avoiding set includes all clusters of size $m\geq d_{x,y,t}$ that touch the coordinate $x,y,t$ with at least $m/\chi$ errors.  Here, $d_{x,y,t}$ is the length of the minimum-length error string which intersects $x,y,t$ and is capable of causing a logical error.  Based on the geometry of the growth procedure. we can see that this spatially-varying distance is given by 
\begin{align}
    d_{x,y,t} = 
    \begin{cases}
        \min(d_1+t, d_2), & \text{if } y \leq d_1, x \leq d_1, t \geq 0,\\[6pt]
        \min(x+t, d_2), & \text{if } y \leq d_1, x > d_1,  t \geq 0,\\[6pt]
        \min(x+t, d_2), & \text{if } y \geq d_1, x> y,  t \geq 0,\\[6pt]
        \min(y+t, d_2), & \text{if } y \geq d_1, x\leq y,  t \geq 0,\\
    \end{cases}
\end{align}
\noindent These cases capture that as time increases and distance grows along the diagonal (where the initialization basis changes), the minimum-length logical operator also grows. Based on this spatially-varying distance, we can define the bulk avoiding sets to be
\begin{equation}
\mathcal{F}_1 = \bigboxplus_{v=(x,y,t)} \bigboxplus_{m\geq d_v} \left\{ E\subseteq C_1: \exists C\in \mathrm{Conn}(v,m)  \textnormal{ such that } |E\cap C|= \lceil m/\chi \rceil  \right\},
\end{equation}
where $C_1$ denotes the set of all edges in the decoding graph and $\mathrm{Conn}(v,m)$ denotes all connected components in the graph of size $m$ which touch $v$. We will take $\chi=4$ below.

Informally, this says that a configuration is ``bad" if there is a cluster touching $v$ of size at least $d_v$ with at least a $1/\chi$ fraction of faults. By construction, this avoiding set also contains the aforementioned string-like events whenever $\chi > 2$. Indeed, if a string $c$ has at least half of its edges in error, then for any point $v$ on the string, the definition of $d_v$ implies that the string has length at least $d_v$. Since $1/2 \geq 1/\chi$ for our choice $\chi=4$, this string is included in $\mathcal{F}_1$. Note the timelike error strings are also captured since they have length $\geq d_2$ and at all spacetime points, $d_{x,y,t}\leq d_2$.

We require two more components in our avoiding set.  The following set $\mathcal{F}_2$, defined in \cite{He2025ComposableFaultTolerance}, allows us to exclude components touching the final time boundary which can cause high residual error:
\begin{align}
    \mathcal{F}_2 = \bigboxplus\limits_{B\subseteq J_{T+1}^D, B\in B^{(R)}}\,\,\bigboxplus\limits_{r\geq 2 |B|}\,\,\bigboxplus\limits_{R\in CO(r,B)}\left\{K\subseteq R: |K| = \lceil|R|/4\rceil, K \subseteq J_{\mathrm{bulk}}\right\}
\end{align}
where $CO(r,B)$ denotes cluster covers of $B$ of size $r$.  Note that the equality $|K| = |R|/4$ suffices  since terms which would have $|K| > |R|/4$ are still part of the avoiding set. Lastly, we define $\mathcal{F}_{\rm in}$. which prohibits any errors in the constant-sized input part of the circuit (which we refer to as $V_{\rm in}$).
\begin{align}
    \mathcal{F}_{\rm in} = \{\{l\}: l \in V_{\rm in}\} .
\end{align}
In total, the avoiding set for the growth gadget is defined to be 
\begin{align}
    \mathcal{F} = \mathcal{F}_{\rm in} \boxplus \mathcal{F}_{1} \boxplus \mathcal{F}_{2}.
\end{align}
This excludes all errors which have any support on the input region.  

With this definition, we can now show that if our error configuration is $\mathcal{F}$-avoiding, then the output of our gadget after decoding satisfies our two criteria. We consider the connected components of $e + c$ which can be decomposed into
\begin{enumerate}
    \item \textbf{Internal components} that do not contain vertices in $J_{1}^d$ or $J_{T}^S$
    \item \textbf{Output components} that contain vertices in $J_{T}^S$ and do not contain vertices in $J_{<T/2}^{S,d}$.
\end{enumerate}
Notably, unlike in \cite{He2025ComposableFaultTolerance}, we have no input components since we include $\mathcal{F}_{\rm in}$ in $\mathcal{F}$.  Furthermore, because we use output padding, there can be no output components that extend from $J_{T/2}^d$ to $J_{T}^s$.  Otherwise, the existence of such components would imply that $e+c$ has a connected component $C$ extending from $J_{T/2}^d$ to $J_{T}^s$ of size $\geq d$ with $|e\cap C|\geq |C|/2 \geq d/2$; this is prohibited by $\mathcal{F}_{1}$.  

We first note that the internal components can have no effect on either criterion of our theorem.  Since these components are fully in the bulk, they can not induce a residual syndrome. Furthermore, by the definition of $\mathcal{F}_1$, these components can not cause a logical error. Thus, we only need to consider residual components and show that they do not compromise either condition.  

Next, we show that the residual error $y$ is sufficiently low, namely that $y$ is $B^{(R)}$-avoiding. Consider one set $B\in B^{(R)}$ and assume that $y$ contains $B$. We can invoke \ref{lemma:residual_error}, noting that $H x = 0$ because the residual component has no support on the initial growth step as discussed previously. This lemma allows us to conclude there must be a cluster cover $V$ of $B$ satisfying
\begin{enumerate}
    \item $V\subseteq (e+c)\sqcup y$ and $|V| \geq 2|B|$
    \item For every cluster $C$ in $V$ we have $|(e+c)\cap C| \geq |C|/2$ (instantiating lemma with $x\rightarrow e+c$).  Necessarily, each cluster $C$ must contain a vertex of $B$.
\end{enumerate}
We have that $|(e+c)\cap C| \geq |C|/2$, and using that the clusters are disjoint, we obtain that 
\begin{equation}
    |(e+c)\cap V| = |(e+c)\cap (c_1 \cup \cdots \cup c_t)| = \sum_{C} |(e+c)\cap C| \geq \sum_{C\in V}|C|/2 = |V|/2 \ .
\end{equation} 
Since $c$ is by definition the minimum-weight correction from the decoder, we have that $|e \cap V|\geq |V|/4$.  Now, denoting $CO(r,B)$ as the set of all cluster covers of $B$ with size at least $r$, we can see that the aforementioned 
\begin{align}
    \mathcal{F}_2 = \bigboxplus\limits_{B\subseteq J_{T+1}^D, B\in B^{(R)}}\bigboxplus\limits_{r\geq 2 |B|}\bigboxplus\limits_{R\in CO(r,B)}\left\{K\subseteq R: |K| = |R|/4, K \subseteq J_{\mathrm{bulk}}\right\}
\end{align}
is exactly the set which allows us to restrict these cluster covers.  Specifically, we consider every set $B\in B^{(R)}$ that we want to restrict. For each such set $B$, the cluster cover should have size $r = |V| \geq 2|B|$. For all cluster covers of this size, we consider the sets of errors within the cluster cover with at least $|R|/4$ errors.  

Having shown that the output components cannot cause large codespace deviations, we show that the output components also cannot cause logical errors. Formally, we must rule out that $\phi(e + c) +y \notin \mathrm{rowsp}(H_X)$.  We again make use of the fact that the components in the output configuration have no support on the growth boundary condition and hence $Hx=0$ in the bulk.  Invoking \ref{lemma:error_to_large_connected_cluster} tells us that there is a connected cluster of vertices $S\subset (c+e)\sqcup y$ of size $d$ such that $|(e + c)\cap S|\geq |S|/2$.  Then we have that 
\begin{align}
    |e \cap S|\geq |S|/4 \ ,
\end{align}
because the errors need to have a larger intersection with $S$ than the corrections by the minimum-weight property of the decoder.  Since $|S|\geq d$ we then have that 
\begin{align}
    |e \cap S|\geq d/4.
\end{align}
In words, this tells us that a logical error would correspond to a cluster with more than a quarter errors. This is directly prohibited by the avoiding set $\mathcal{F}_1$, since it is a connected cluster with at least $1/4$ errors. Hence, $\chi=4$ suffices for the remainder of the proof.

Now that we've shown that our choice of avoiding set prevents large residuals and logical errors, we turn to bounding the weight enumerator for the gadget to obtain our fault-tolerance guarantees. The weight enumerator polynomial $\mathcal{W}(F, x)$ for a family $F$ of bad supports is given by $\mathcal{W}(F, x) = \sum_{w\geq 0} A_w x^w$, where $A_w$ is the number of terms with support $w$ in $F$.  We now independently bound the contribution of each part of $\mathcal{F}$ to the weight enumerator.  
\begin{align}
    \mathcal{W}(\mathcal{F}, x) \leq \mathcal{W}(\mathcal{F}_{\rm in}, x) + \mathcal{W}(\mathcal{F}_1, x) + \mathcal{W}(\mathcal{F}_2, x)
\end{align}
First, let us consider $\mathcal{F}_{\rm in}$, which captures all possible errors which happen in the circuit growth.  We simply consider an error in any one of the locations before the growth.  Then, we have that for the input region with fault volume $|V_{\rm in}|$,
\begin{align}
    \mathcal{W}(\mathcal{F}_{\rm in}, x) \leq |V_{\rm in}| x
\end{align}
if $x$ is between $0$ and $1$.  This is the expected form, as we have linear scaling in $x$ from the input component.  

Next, we look at $\mathcal{F}_1$ which is the most non-trivial in this setting of surface code growth because our distance is dependent on the space-time coordinate. We make use of two facts.  First, the number of connected sets of size $m$ on a graph with degree $D$ is less than $(D e)^{m-1} \leq (D e)^m$.  Second, the weight enumerator for connected clusters in a graph with $V$ vertices is given by
\begin{align}
    \mathcal{W}(CG(d,t), x)\leq |V| \binom{d}{t-1} (D e)^{m-1} x^t .
\end{align}
Using the definition of $\mathcal{F}_1$, we can write
\begin{equation}
\mathcal{F}_1 = \bigcup_{v=(x,y,t)}  \bigcup_{m\geq d_v}  CG_v(m,\lceil m/\chi \rceil) 
 ,
\end{equation}
where $CG_v(m,\lceil m/\chi \rceil)$ denotes the family of connected clusters of size $m$ touching $v$ with $\lceil m/\chi \rceil$ faults. Therefore,
\begin{equation}
\mathcal{W}(\mathcal{F}_1,x) \leq \sum_{v=(x,y,t)}\sum_{m\geq d_v} \mathcal{W}(CG_v(m,\lceil m/\chi \rceil),x).
\end{equation}
For each term,
\begin{equation}
\mathcal{W}(CG_v(m,\lceil m/\chi \rceil),x) \leq (De)^{m-1}\sum_{s\geq m/\chi}\binom{m}{s}x^s \leq (2De)^m x^{m/\chi}
\end{equation}
where for the first bound we sum over all fault supports with at least $m/\chi$ errors inside each connected cluster.  Defining $\eta_x = 2De\,x^{1/\chi}$, we obtain
\begin{equation}
\mathcal{W}(\mathcal{F}_1,x) \leq \frac{1}{1-\eta_x} \sum_{v=(x,y,t)}\eta_x^{d_v}
\end{equation}
As such, the bound reduces reduces to evaluating this sum and making use of the spatially-dependent distance. In the following manipulations, we use $\Delta = d_2-d_1$. Then we have
\begin{align}
    \mathcal{W}(\mathcal{F}_1, x) &\leq \frac{1}{1-\eta_x}\sum_{v=(x,y,t)} \eta_x^{d_v} \nonumber , \\
    &\leq \frac{1}{1-\eta_x}\Bigg[ 
        d_1^2\sum_{\tau=0}^{T-1} \eta_x^{d_1 + \min(\tau, \Delta)} + d_1\sum_{u = 1}^{\Delta} \sum_{\tau=0}^{T-1} \eta_x^{d_1 + \min(u + \tau, \Delta)} \nonumber  \\
        &\quad \quad \quad \quad \quad + \sum_{x=1}^\Delta (d_1+x) \sum_{\tau=0}^{T-1} \eta_x^{d_1 + \min(x+\tau, \Delta)} + \sum_{x=1}^\Delta \sum_{y=x+1}^\Delta \sum_{\tau=0}^{T-1} \eta_x^{d_1+\min(y+\tau, \Delta)}
    \Bigg]\nonumber , \\
    &= \frac{1}{1-\eta_x}\Bigg[ 
        d_1^2\sum_{\tau=0}^{T-1} \eta_x^{d_1 + \min(\tau, \Delta)} + d_1\sum_{u = 1}^{\Delta} \sum_{\tau=0}^{T-1} \eta_x^{d_1 + \min(u + \tau, \Delta)} \nonumber \\
        &\quad \quad \quad \quad \quad + \sum_{x=1}^\Delta (d_1+x) \sum_{\tau=0}^{T-1} \eta_x^{d_1 + \min(x+\tau, \Delta)} + \sum_{y=1}^\Delta (y-1) \sum_{\tau=0}^{T-1} \eta_x^{d_1+min(y+\tau, \Delta)}
    \Bigg]\nonumber , \\
    &= \frac{1}{1-\eta_x}\Bigg[ 
        d_1^2\sum_{\tau=0}^{T-1} \eta_x^{d_1 + \min(\tau, \Delta)} + \sum_{u = 1}^{\Delta} (2d_1+2u-1) \sum_{\tau=0}^{T-1}  \eta_x^{d_1 + \min(u + \tau, \Delta)}
    \Bigg]\nonumber ,\\
    &\leq \frac{1}{1-\eta_x}\Bigg[ 
        d_1^2 \left((T-\Delta) \eta_x ^{d_2} + \sum_{\tau=0}^{\Delta-1}\eta_x^{d_1 + \tau} \right)+ \sum_{u = 1}^{\Delta} (2d_1+2u-1) \left((T+u-\Delta)\eta_x^{d_2}  + \sum_{\tau=0}^{\Delta-u-1}\eta_x^{d_1 + u + \tau}\right)
    \Bigg]\nonumber , \\
    &\leq \frac{1}{1-\eta_x}\Bigg[ 
        d_1^2 \left((T-\Delta) \eta_x ^{d_2} + \frac{\eta_x^{d_1}}{1-\eta_x} \right)+ \sum_{u = 1}^{\Delta} (2d_1+2u-1) \left((T+u-\Delta)\eta_x^{d_2}  + \frac{\eta_x^{d_1 + u}}{1-\eta_x}\right)
    \Bigg]\nonumber , \\
    &\leq \frac{\eta_x^{d_2}}{1-\eta_x}\Bigg[ d_1^2 (T-\Delta) + \sum_{u     = 1}^{\Delta} (2d_1+2u-1) (T+u-\Delta)\Bigg]
        + \frac{\eta_x^{d_1}}{(1-\eta_x)^2}\Bigg[d_1^2 + \sum_{u = 1}^{\Delta} (2d_1+2u-1) \eta_x^{u}
    \Bigg]\nonumber , \\
    &\leq \frac{T\eta_x^{d_2}}{1-\eta_x}\Bigg[ d_1^2 + \sum_{u     = 1}^{\Delta} (2d_1+2u-1)\Bigg]
        + \frac{\eta_x^{d_1}}{(1-\eta_x)^2}\Bigg[d_1^2 + \sum_{u = 1}^{\Delta} (2d_1+2u-1) \eta_x^{u}
    \Bigg]\nonumber , \\
    &\leq \frac{T d_2^2 \eta_x^{d_2}}{1-\eta_x}
        + \frac{\eta_x^{d_1}}{(1-\eta_x)^2}\Bigg[d_1^2 + \frac{(2 d_1-1) \eta_x}{1-\eta_x} + \frac{2\eta_x}{(1-\eta_x)^2}
    \Bigg] .
\end{align}

As desired, this weight enumerator bound includes $d_1$ dependent contributions and $d_2$ dependent terms which are are all exponentially decaying in $d_2$.  Lastly, we bound $\mathcal{F}_2$.  As before, we have
\begin{align}
    \mathcal{W}(\mathcal{F}_2, x) &\leq \sum_{B \in B^{(R)}} \sum_{r \geq 2 |B|} \sum_{R \in CO(r,B)} \binom{|R|}{|R|/4}x^{|R|/4} \nonumber \\
    &\leq \sum_{B \in B^{(R)}} \sum_{r \geq 2 |B|}  e^{|B|-1}(D e)^{r-|B|} 2^{r}x^{r/4} \nonumber \\
    &\leq \sum_{B \in B^{(R)}} \sum_{r \geq 2 |B|}  (D e)^{r} 2^{r}x^{r/4} \nonumber \\
    &\leq \sum_{B \in B^{(R)}} \frac{(2 D e x^{1/4})^{2|B|}}{1- 2D e x^{1/4}}\nonumber \\
    &\leq \frac{1}{1- 2D e x^{1/4}}\mathcal{W}\left(B^{(R)}, (2 D e x^{1/4})^{2}\right)\nonumber \\
\end{align}
The total weight enumerator of the gadget is now given by
\begin{align}
    \mathcal{W}(\mathcal{F}, x) &\leq \mathcal{W}(\mathcal{F}_{\rm in}, x) + \mathcal{W}(\mathcal{F}_1, x) + \mathcal{W}(\mathcal{F}_2, x)\nonumber \\
    &\leq |V_{\rm in}| x
            + \frac{T d_2^2 \eta_x^{d_2}}{1-\eta_x}
        + \frac{\eta_x^{d_1}}{(1-\eta_x)^2}\Bigg[d_1^2 + \frac{(2 d_1-1) \eta_x}{1-\eta_x} + \frac{2\eta_x}{(1-\eta_x)^2}\Bigg] \nonumber \\
        &\quad+ \frac{1}{1- 2D e x^{1/4}} \mathcal{W}\left(B^{(R)}, (2De)^2 x^{1/2}\right)
\end{align}
As desired, this weight enumerator polynomial has terms which are dependent only on the distance $d_1$, and all other terms are exponentially suppressed in $d_2$.  
\end{proof}

Thus far we have handled decoding of this gadget for one logical sector.  Now we extend to correction of errors in both sectors to obtain the main result.  
\begin{theorem}[One step surface code injection gadget]
\label{thm:unregularized_input_gadget}
    There exists a constant $\epsilon_{*,\rm inject}\in (0,1)$ and a gadget $\mathrm{INJECT}^d$ such that given a reduced state $\rho_{Ri}$ of a single qubit $i$, possibly entangled with or purified by a reference system $R$, 
    the operation $\left(I_R \otimes \mathrm{INJECT}^d\right)(\rho_{Ri})$ outputs a state that is $SC(d)_{1/5}$-deviated from $\left(I_R \otimes \mathrm{Enc}_{SC(d)}\right)(\rho_{Ri})$ for every fault $f$ that avoids $\mathcal{F}_{\rm inject}^d$.  The weight enumerator is bounded on $x\in [0, \epsilon_{*,\rm inject}]$ as 
    \begin{align}
        \mathcal{W}(\mathcal{F}_{\rm inject}^d, x) \leq c \cdot x
    \end{align}
    for a constant $c>0$.  The gadget has space overhead $O(d^2)$ and time overhead $O(d)$. 
\end{theorem}
\begin{proof}
    We use the one step injection gadget we have introduced.  If each sector has no logical error then the overall channel has no logical error as shown in Theorem \ref{thm:single_sector_decoding}.  We found that the weight enumerator in one sector is given by
    \begin{align}
        \mathcal{W}(\mathcal{F}_{\rm inject}^{d,X}, x) &\leq |V_{\rm in}| x 
        + \frac{\eta_x^{d_1}}{(1-\eta_x)^2}\Bigg[d_1^2 + \frac{(2 d_1-1) \eta_x}{1-\eta_x} + \frac{2\eta_x}{(1-\eta_x)^2}\Bigg] + \frac{T d_2^2 \eta_x^{d_2}}{1-\eta_x}\nonumber \\
        &\quad+ \frac{1}{1- 2D e x^{1/4}}\mathcal{W}\left(B^{(R)}, (2De)^2 x^{1/2}\right)
    \end{align}
    We would like to show that the weight enumerator can be bounded by $c \cdot x$.  The first term $|V_{\rm in}| x$ already scales linearly in $x$.  The next terms 
    \begin{align}
        \frac{\eta_x^{d_1}}{(1-\eta_x)^2}\Bigg[d_1^2 + \frac{(2 d_1-1) \eta_x}{1-\eta_x} + \frac{2\eta_x}{(1-\eta_x)^2}\Bigg]
    \end{align}
    can be bounded taking $x< \epsilon_*$ and $\epsilon_*$ sufficiently small so that $\eta_x \leq 1/2$, and moreover by taking $d_1\geq \chi=4$ so that $\eta_x^{d_1}$ scales at least linearly in $x$.  This gives the desired bound
    \begin{align}
    	\frac{\eta_x^{d_1}}{(1-\eta_x)^2}\Bigg[d_1^2 + \frac{(2 d_1-1) \eta_x}{1-\eta_x} + \frac{2\eta_x}{(1-\eta_x)^2}\Bigg] \leq C_1 (2De\,x^{1/\chi})^{d_1}
        \leq c_1 \cdot x.
    \end{align}  
    For the term
    \begin{align}
\label{eq:bound_on_d2_scaling_terms}
         \frac{T d_2^2 \eta_x^{d_2}}{1-\eta_x} ,
    \end{align}
    The coefficient $P(d_2, T, \Delta)=T d_2^2 $ can be bounded by $C d_2^3$ for some constant $C$ noting that $T$ is $O(d_2$).  We take $\epsilon_*$ sufficient small so $\eta_x \leq 1/2$, and since $d_2 \geq d_1$ this can be bounded by 
    \begin{align}
         \frac{T d_2^2 \eta_x^{d_2}}{1-\eta_x} \leq C d_2^3 (2De\,x^{1/\chi})^{d_2}
         \leq C d_2^3 (2De\,x^{1/\chi})^\chi \,(2De\,x^{1/\chi})^{d_2-\chi}
         \leq C d_2^3 (2De)^{\chi}2^{-(d_2-\chi)}x
         \leq c_2 \cdot x
    \end{align}
    once we note that $\sup_{d_2 \geq \chi} d_2^3 2^{-(d_2-\chi)}$ is a constant. For the last term
    \begin{align}
        \frac{1}{1- 2D e x^{1/4}}\,\mathcal{W}\!\left(B^{(R)}, (2 D e)^2( x)^{1/2}\right) \ ,
    \end{align}
    we note that we have $B^{(R)} = B_{1/5}$ and that after evaluation this will take the form $\sum_{m\geq d_2} \mathrm{poly}(m) (c_1 x^{1/c_2})^{m}$.  We can expand this as
    \begin{align}
    	\sum_{m\geq d_2} \mathrm{poly}(m) (c_1 x^{1/c_2})^{m} 
        &= \sum_{m\geq d_2}  \mathrm{poly}(m) (c_1 x^{1/(2c_2)})^{m} (x^{1/(2c_2)})^{m} \nonumber \\
        &\leq \sum_{m\geq d_2}  \mathrm{poly}(m) 2^{-2 m} x \leq c_3 \cdot x
    \end{align}
    Where we have taken $\epsilon_*$ sufficiently small so $c_1 x^{1/(2c_2)} \leq \frac{1}{4}$ and $d_2 > d_1 \geq 2 c_2$.  Hence overall we need $d_1 \geq \max(2 c_2, \chi)$.  Since $\sum_{s\geq s_0}  \mathrm{poly}(s) 2^{-2 s}$ is convergent this can be bounded by $c_3 \cdot x$.  Thus overall we have that the weight enumerator is bounded by 
    \begin{align}
        \mathcal{W}(\mathcal{F}_{\rm inject}^d, x)\leq |V_{\rm in}|x + c_1 \cdot x + c_2 \cdot x + c_3 \cdot x.
    \end{align}
    Accounting for both $X,Z$ sectors will contribute a factor of 2 so that $c = 2(|V_{\rm in}| + c_1 + c_2 + c_3)$.  
\end{proof}
In order to couple this gadget to the rest of our circuit we need the guarantee that the output will not have too much residual error.  While our theorem already provides this guarantee, it will be useful in our proof to independently condition on the residual being small or the logical being correct. For simplicity, we will define a new gadget with a memory gadget appended.

\begin{theorem}[Regularized one-step surface-code injection gadget]
\label{thm:regularized_input_gadget}
Define the regularized input gadget by appending one surface-code error-correction
gadget to the injection gadget:
\begin{equation}
\mathrm{IN}_i^d = \mathrm{EC}_i^d\circ \mathrm{INJECT}_i^d .
\end{equation}
There exists a constant $\epsilon_{*,\mathrm{inject}}\in(0,1)$ such that the following holds. For every state $\rho_{Ri}$ of an input qubit $i$ possibly entangled with a reference system $R$, the gadget has a fault family
\begin{equation}
\mathcal{F}_{\mathrm{in}}^d = \mathcal{F}_{\mathrm{grow}}^d \boxplus \mathcal{F}_{\mathrm{res}}^d ,
\end{equation}
where $\mathcal{F}_{\mathrm{grow}}^d$ controls logical faults during the growth stage and $\mathcal{F}_{\mathrm{res}}^d$ controls excessive residual deviation after the appended memory gadget.
If a fault path $f$ is $\mathcal{F}_{\mathrm{res}}^d$-avoiding, then there is a
single-qubit CPTP map $\mathcal{L}_f$ such that
\begin{equation}
(I_R\otimes \mathrm{IN}_i^d)(\rho_{Ri})
\end{equation}
is $SC(d)_{1/5}$-deviated from
\begin{equation}
(I_R\otimes \mathrm{Enc}_{SC(d)}) (I_R\otimes \mathcal{L}_f)(\rho_{Ri}).
\end{equation}
If, in addition, $f$ is $\mathcal{F}_{\mathrm{grow}}^d$-avoiding, then
$\mathcal{L}_f=\mathrm{Id}$.
For all $x\in[0,\epsilon_{*,\mathrm{inject}}]$,
\begin{equation}
    \mathcal{W}(\mathcal{F}_{\mathrm{grow}}^d,x)\leq c x,
    \qquad
    \mathcal{W}(\mathcal{F}_{\mathrm{res}}^d,x)
    \leq \operatorname{poly}(d)e^{-\alpha d},
\end{equation}
for constants $c,\alpha>0$. The gadget has width $O(d^2)$ physical qubits and
time overhead $O(d)$ syndrome-extraction rounds.
\end{theorem}

\begin{proof}
    The surface code memory gadget has input type $SC(d)_{1/2}$ and output code type $SC(d)_{1/5}$ (Corollary 7.3  of Ref.~\cite{He2025ComposableFaultTolerance}).  Consider a generic fault path which we decompose across the growth and error correction gadget as $f = (f_{\rm grow}, f_{EC})$.
    
    By Definition \ref{def:circuit_gadget} and Lemma \ref{lemma:surface_code_ec_gadget}, the appended $\mathrm{EC}_{SC(d)}$ is a friendly gadget so if $f_{EC}$ is $\mathcal{F}_{\mathrm{res}}$-avoiding then for an arbitrary input the output will be $SC(d)_{1/5}$-deviated from a logical state.  In terms of the channel this tells us that if $f_{EC}$ is $\mathcal{F}_{\mathrm{res}}$ avoiding then $(I_R \otimes \mathrm{IN}_i^d )(\rho_{Ri})$ is $SC(d)_{1/5}$-deviated from $(I_R \otimes \mathrm{Enc}_{SC(d)} )(I_R \otimes \mathcal{L}_f)(\rho_{Ri})$ for some single qubit logical channel $\mathcal{L}_f$.  The channel $\mathcal{L}_f$ is given by $\mathcal{L}_f = \mathrm{dec}_{SC(d)} \circ \mathrm{IN}_i^d[f]$.  

    From our input gadget lemma \ref{thm:unregularized_input_gadget} we know that $\mathrm{INJECT}_i^d$ will not have any logical errors and its output is $SC(d)_{1/5}$-deviated if $f_{\rm grow}$ avoids $\mathcal{F}_{\mathrm{grow}}$.  If $f_{EC}$ is $\mathcal{F}_{\mathrm{res}}$-avoiding then $\mathrm{EC}_i^d$ will also not induce a logical error since its a friendly gadget acting on a state in the valid input class. Hence $\mathcal{L}_f= \mathrm{Id}$ if $f$ avoids $\mathcal{F}_{\mathrm{grow}}$ and $\mathcal{F}_{\mathrm{res}}$.  Reading off the weight enumerator bounds of $\mathcal{F}_{\mathrm{grow}}$ and $\mathcal{F}_{\mathrm{res}}$ completes the proof. The bounds for growth are from \ref{thm:unregularized_input_gadget} and the bounds for the padding from corollary 7.3.
\end{proof}

Next we use Theorem~\ref{thm:regularized_input_gadget} to show that any circuit which takes arbitrary physical input qubits can be simulated fault-tolerantly, at the cost of a one-time effective input-noise channel.
\begin{theorem}[Fault-tolerant simulation with arbitrary physical inputs]
\label{thm:arb_input_threshold_theorem}
There exists a constant $\epsilon_*\in(0,1)$ such that the following holds. Let $C$ be a Clifford+$T$ circuit of width $W$ and depth $D$ whose input includes $l$ physical qubits in an arbitrary state $\rho$. Assume local-stochastic circuit noise of rate $p<\epsilon_*$. For every target accuracy $\epsilon\in(0,1)$, there is a noisy fault-tolerant implementation $\tilde{\overline C}$ such that
\begin{equation}
\frac{1}{2}  \left\| \mathrm{dec}\circ \widetilde{\overline C}(\rho)  - C\!\left(\mathcal{N}_{\mathrm{inj}}^l(\rho)\right) \right\|_1 \leq \epsilon .
\end{equation}
Here $\mathcal{N}_{\mathrm{inj}}^l$ is an effective input channel induced by
the $l$ uploading gadgets. If the noise on distinct input gadgets is independent,
then
\begin{equation}
\mathcal{N}_{\mathrm{inj}}^l = \bigotimes_{i=1}^l \mathcal{N}_{\mathrm{inj},i}.
\end{equation}
If the input-gadget noise is independent and identically distributed, then
\begin{equation}
\mathcal{N}_{\mathrm{inj}}^l  = \mathcal{N}_{\mathrm{inj}}^{\otimes l}.
\end{equation}
Moreover, each marginal input channel can be written as
\begin{equation}
\mathcal{N}_{\mathrm{inj},i} = (1-q_i)\mathrm{Id}+q_i\mathcal{R}_i
\end{equation}
with $q_i\leq c_{\mathrm{in}}p$ for a constant $c_{\mathrm{in}}$. The fault-tolerant simulation has only polylogarithmic overhead in $(WD+l)/\epsilon$.
\end{theorem}
\begin{proof}
    For each qubit in $[l]$ apply the injection channel used in \ref{thm:regularized_input_gadget} given by
    \begin{align}
    \mathrm{IN}_i^d = \mathrm{EC}_i^d \circ \mathrm{INJECT}_i^d.
    \end{align}
    We will denote these $l$ gadgets by $V_1, ..., V_l$.  Each input gadget has an associated avoiding set $\mathcal{F}_{\mathrm{grow}, i}^d$ for the growth part and $\mathcal{F}_{\mathrm{res}, i}^d$ which controls the residual.  For all other elements $w$ in the circuit replace the gates by their corresponding fault-tolerant gadget $\mathcal{G}_{w}^d$ with type $SC(d)_{1/5}$.  Fault-tolerant gadgets for all operations required for Clifford + T are presented in \ref{lemma:surface_code_universal_gadgets}.  Collectively we will refer to these non-input gadgets as $w \in V_{\mathrm{bulk}}$.  We will use $\bar{C}$ to denote the $l$ input gadgets along with the fault-tolerant implementation of the remaining gates in $C$.  Define the set of failures for the bulk of the circuit by
    \begin{align}
        \mathcal{F}_{\rm bulk} = \left(\bigboxplus\limits_{i=1}^{l}  \mathcal{F}_{\mathrm{res}, i}^d \right) \boxplus  \left(\bigboxplus\limits_{w\in V_{\rm bulk}} \mathcal{F}_{w}^d\right).
    \end{align}
    The set $\mathcal{F}_{\rm bulk}$ deliberately excludes $\mathcal{F}_{\mathrm{grow}, i}^d$ because we will interpret those errors as input noise.  Now let us consider a fault path $f$ which is $\mathcal{F}_{\rm bulk}$-avoiding.  Since $\mathcal{F}_{\rm bulk}$ includes the $\mathcal{F}_{\mathrm{res}, i}^d$ for each input gadget we know that the outputs of each injection gadget are $SC(d)_{1/5}$-deviated from some encoded logical state as we showed in \ref{thm:regularized_input_gadget}.  This is true regardless of whether $f$ is $\mathcal{F}_{\mathrm{grow}, i}$-avoiding.  The result of the injection is the state 
    \begin{align}
        \operatorname{Enc}_{SC(d)}^{\otimes l}\left(\mathcal{L}_f(\rho)\right)
    \end{align}
    up to $SC(d)_{1/5}$ deviation on each encoded qubit.  For a fixed $\mathcal{F}_{\rm bulk}$ avoiding path $f$, $\mathcal{L}_f$ acts trivially on every input $i$ for which $f$ is $\mathcal{F}_{\mathrm{grow}, i}$-avoiding.  For the set of $i$ where $f$ includes a set from $\mathcal{F}_{\mathrm{grow}, i}$ the action of $\mathcal{L}_f$ can be an arbitrary (possibly correlated channel) on the set.  If the $f$ are sampled from a probability distribution then the average channel is given by $\mathcal{N}_{\mathrm{inj}}^l = \mathbb{E}_f \left[\mathcal{L}_f\right]$.  In the setting where the physical noise is local-stochastic which is independent across the gadgets then the channels factorize across the qubits as 
    \begin{align}
        \mathcal{N}_{\mathrm{inj}}^l = \bigotimes_{i=1}^l \mathcal{N}_{\mathrm{inj}, i}.
    \end{align}
    Furthermore if the noise is identical on each input then
    \begin{align}
        \mathcal{N}_{\mathrm{inj}}^l = \mathcal{N}_{\mathrm{inj}}^{\otimes l}.
    \end{align}
    As mentioned if the fault paths are $\mathcal{F}_{\mathrm{res}, i}^d$-avoiding the outputs of the injection gadgets are $SC(d)_{1/5}$-deviated.   Hence for subsequent gadgets $w \in V_{\rm bulk}$ if $f$ is $\mathcal{F}_{w}^d$-avoiding, we can iteratively apply \ref{thm:gadget_composition} to combine all the gadgets corresponding to $w\in V_{\mathrm{bulk}}$ to get a single $(C, \mathcal{F}_{\rm bulk})$-fault tolerant gadget with combined fault family $\mathcal{F}_{\rm bulk}$.  Overall then the bulk simulates the action of $C$ on the noisy logical input $\mathcal{L}_f(\rho)$.  We can bound the weight enumerator for the bulk, using the behavior of weight enumerators on unions, to be
    \begin{align}
        \mathcal{W}(\mathcal{F}_{\rm bulk}, p) \leq \sum_{i=1}^l \mathcal{W}(\mathcal{F}_{\mathrm{res}, i}^d, p)+\sum_{w\in V_{\rm bulk}} \mathcal{W}(\mathcal{F}_{w}^d, p) \leq (l+|\bar{V}_{\rm bulk}|) \operatorname{poly}(d) e^{-\beta d}.
    \end{align}
    Here we make use of the fact that the individual circuit gadgets, including the padding added to the injection gadgets, have when $p < \epsilon_*$ the bound $\mathcal{W}(\mathcal{F}_{w}^d, p) \leq \operatorname{poly}(d) e^{-\beta d}$ \ref{lemma:surface_code_universal_gadgets}.  For $p<\epsilon_*$ we choose $d=\operatorname{polylog}\left((l+W D)/\epsilon\right)$ so that $\mathcal{W}(\mathcal{F}_{\rm bulk}, p) \leq \epsilon$.  The weight enumerator for the input gadget growth step from \ref{thm:regularized_input_gadget} is given by
    \begin{align}
        \mathcal{W}(\mathcal{F}_{\mathrm{grow}, i}, x) \leq c \cdot x 
    \end{align}
    For independent local stochastic noise we have that the error probability of an input gadget $q_i$ failing can be conservatively bounded as
    \begin{align}
        q_i\leq \mathcal{W}(\mathcal{F}_{\mathrm{grow}, i}, p) + \mathcal{W}(\mathcal{F}_{\rm{res}, i}, p)\leq c \cdot p + \operatorname{poly}(d) e^{-\alpha d}.
    \end{align}
    Note that the $\mathcal{W}(\mathcal{F}_{\mathrm{res}}, x)$ term is controlled via a similar means as our bounds in \ref{thm:single_sector_decoding}. As such, following the ideas of \ref{eq:bound_on_d2_scaling_terms}, $q_i$ can equivalently be bounded by $\leq c’ \cdot p$ for a constant $c’$.  Hence the noise channel for the initial injection on one qubit is given by 
    \begin{align}
        \mathcal{N}_{\mathrm{inj}, i}^d = (1-q) \mathrm{Id} + q R
    \end{align}
    for some CPTP map $R$.  Since the physical noise is independent across the different injection gadgets then we have that the overall input noise channel is given by
    \begin{align}
        \mathcal{N}_{\mathrm{inj}}^d = \bigotimes_{i=1}^l \left((1-q) \mathrm{Id} + q R\right).  
    \end{align}
    Here we are assuming that the noise is uniform and hence the same $q$ and $R$ apply on each injection gadget. As an aside note that this final statement relies on the controlled residuals from each of the injection gadgets.  If the residuals were not controlled an individual input gadget could corrupt more of the circuit. For all $f$ that avoid the $\mathcal{F}_{\rm bulk}$ we have that 
    \begin{align}
        \mathrm{dec} \circ \bar{C}[f](\rho) = C(\mathcal{L}_f(\rho))
    \end{align}
    where $\mathrm{dec}$ is the ideal decoder and $\bar{C}[f]$ indicates the fault-tolerant circuit with noise applied according to $f$.  Hence when we average over all fault paths $f$ the outputs only differ on fault paths in $\mathcal{F}_{\rm bulk}$ and we get the bound on the trace distance between sampling from the noisy encoded circuit and the ideal circuit acting on noisy inputs given by
    \begin{align}
        \frac{1}{2}\Vert  \mathrm{dec} \circ \widetilde{\bar{C}}(\rho) -C(\mathcal{N}_{\mathrm{inj}}^l(\rho))  \Vert_1 
        &\leq P(f \text{ is not }\mathcal{F}_{\rm bulk}\mathrm{-avoiding}) \nonumber \\
        &\leq \mathcal{W}(\mathcal{F}_{\rm bulk}, p) \leq \epsilon.  
    \end{align}
    Here the tilde indicates noise being applied according to the distribution we sample $f$ from and the equality is due to our assumption that the noise is independent between the different injection gadgets.  As mentioned earlier if the noise acts independently and identically on the different input gadgets then
    \begin{align}
        \mathcal{N}_{\mathrm{inj}}^l(\rho) = \mathcal{N}_{\mathrm{inj}}^{\otimes l}(\rho).
    \end{align}
\end{proof}
For our later results which treat the injection noise channel as an effective depolarizing channel, we use the following simple lemma.
\begin{lemma}[Logical Clifford twirling]
\label{lemma:twirling}
Let $\mathsf{Cl}_1$ denote the single-qubit Clifford group. For any single-qubit
channel of the form
\begin{equation}
\mathcal{N} = (1-q)\,\mathrm{Id}+q\mathcal{R},
\end{equation}
with $\mathcal{R}$ CPTP, define its Clifford twirl by
\begin{equation}
\mathcal{T}(\mathcal N)(\rho) = \frac{1}{|\mathsf{Cl}_1|} \sum_{C\in\mathsf{Cl}_1}  C^\dagger \mathcal{N}(C\rho C^\dagger) C .
\end{equation}
Then
\begin{equation}
\mathcal{T}(\mathcal N)(\rho) = (1-p_L)\rho +\frac{p_L}{3} \left(X\rho X+Y\rho Y+Z\rho Z\right)
\end{equation}
for some $p_L\leq q$.
\end{lemma}
\begin{proof}
The Clifford twirl symmetrizes the action of the channel on the three nontrivial Pauli operators. Thus $\mathcal{T}(\mathcal R)$ is a depolarizing Pauli channel, say
\begin{equation}
\mathcal{T}(\mathcal R)(\rho) = (1-r)\rho+\frac{r}{3}(X\rho X+Y\rho Y+Z\rho Z)
\end{equation}
for some $r\in[0,1]$. Since the identity channel is fixed by the twirl,
\begin{equation}
\mathcal{T}(\mathcal N) = (1-q)\mathrm{Id}+q\,\mathcal{T}(\mathcal R),
\end{equation}
which has the same form with $p_L=qr\leq q$.
\end{proof}
As a corollary we obtain our main theorem:
\begin{corollary}[Uploaded oracle reduction for state-preparation oracles]
\label{cor:uploaded_oracle_reduction_supp}
Let $O$ be an $n$-qubit state-preparation oracle, $O:1\mapsto \rho_O$, and let
\begin{equation}
\mathcal{N}_{\lambda_*}(O) := \mathcal{D}_{\lambda_*}^{\otimes n}\circ O
\end{equation}
denote the oracle that prepares one depolarized copy of $\rho_O$. For every below-threshold physical noise rate, there is a constant $\lambda_*=O(p)$ such that, up to the usual bounded-error equivalence,
\begin{equation}
\textnormal{\textsf{NBQP}}^{O} = \textnormal{\textsf{BQP}}^{\mathcal{N}_{\lambda_*}(O)}
\end{equation}
for state-preparation oracles.
\end{corollary}
\begin{proof}
We first show $\textnormal{\textsf{NBQP}}^{O}\subseteq \textnormal{\textsf{BQP}}^{\mathcal{N}_{\lambda_*}(O)}$. Consider an $\textnormal{\textsf{NBQP}}^O$ computation making polynomially many queries to the state-preparation oracle. Since $O$ is a state-preparation oracle, the computation receives polynomially many physical copies of the same state $\rho_O$. We regard these copies as the physical input qubits to a larger circuit containing all subsequent coherent processing. Applying Theorem~\ref{thm:arb_input_threshold_theorem}, the entire noisy computation can be simulated by a fault-tolerant computation acting on the same circuit but with an effective one-time input channel applied independently to each uploaded oracle output. By applying independent logical Clifford twirls to the uploaded inputs and undoing them in the Pauli frame, Lemma~\ref{lemma:twirling} converts this input channel into a local depolarizing channel $\mathcal{D}_{\lambda_*}^{\otimes n}$ with $\lambda_*=O(p)$. Therefore the noisy oracle computation is simulated by a noiseless $\textnormal{\textsf{BQP}}$ computation with access to $\mathcal{N}_{\lambda_*}(O)$.
Conversely, $\textnormal{\textsf{BQP}}^{\mathcal{N}_{\lambda_*}(O)} \subseteq \textnormal{\textsf{NBQP}}^{O}$ because an $\textnormal{\textsf{NBQP}}^O$ machine can query $O$, apply the physical noise already present in the model, add extra local depolarizing noise if needed to match $\lambda_*$, and then run the remaining computation fault-tolerantly by Theorem \ref{thm:arb_input_threshold_theorem}. This implements the noisy state-preparation oracle $\mathcal{N}_{\lambda_*}(O)$ inside the $\textnormal{\textsf{NBQP}}^O$ model and runs a fault-tolerant circuit that simulates the idealized $\bqp$ circuit.
\end{proof}

\begin{remark}[Extension to channel oracles]
Corollary~\ref{cor:uploaded_oracle_reduction_supp} is stated for state-preparation oracles to avoid bookkeeping at the oracle input boundary. The same idea should extend to general CPTP channel oracles
\begin{align}
O:\mathcal{L}(\mathcal{H}_A) \to \mathcal{L}(\mathcal{H}_B).
\end{align}
In a fault-tolerant simulation, each oracle call is a physical boundary operation. The oracle output $B$ can be uploaded using Theorem~\ref{thm:arb_input_threshold_theorem}, even when it remains entangled with the rest of the computation, because the uploading theorem allows an arbitrary reference system.

For channel oracles with nontrivial quantum inputs, one must also account for the interface that exposes the encoded query register $A$ to the physical oracle. Thus the corresponding ideal noisy oracle should have the form
\begin{align}
\widetilde O_{\lambda_*} = \mathcal{N}_{\mathrm{out}}\circ O\circ \mathcal{N}_{\mathrm{in}},
\end{align}
where $\mathcal{N}_{\mathrm{out}}$ is the effective uploading noise on the oracle output and $\mathcal{N}_{\mathrm{in}}$ is the effective input-boundary noise. After logical Clifford twirling, these boundary channels can be taken to be local depolarizing channels with rates $O(p)$. For state-preparation oracles, $A$ is trivial, so $\mathcal{N}_{\mathrm{in}}$ is absent and the corollary above is recovered.
\end{remark}

\section{Exponential Speedup in Estimating Third Moments}

In this section, we prove an exponential speedup from uploading for estimating third-moment observables of quantum states. Multi-copy primitives are a central source of quantum learning advantages, but they are highly noise-sensitive because they often require entangling qubits across several copies before measurement. Third-moment estimation therefore serves as a representative setting in which fault tolerance restores quantum learning strategies that would otherwise be lost to unprotected physical noise. We begin with a separation for estimation of the Renyi-3-like quantity $\tr(\rho^3)$, then extend to observables of the form $\tr(O\rho^3)$. Such quantities are important signatures of quantum chaos, operator scrambling, and topological order in many-body physics, eigenvalue statistics in quantum thermalization, and correlators in conformal field theories. 

We make explicit the folklore belief that computing $k$-th order traces requires generating extensive entanglement between $k$ copies of the state, proving that even a noiseless protocol that uses a depth-1 circuit composed of $2$-qubit gates and performs joint measurements on three copies of the $n$-qubit state $\rho$ requires an exponential-in-$n$ number of copies to estimate $\tr(\rho^3)$. We then prove that circuits of depth $2$ or greater face a noise-dependent exponential lower bound on sample complexity, whereas an injection-enabled upper bound remains exponentially more efficient for constant noise overhead in performing the injection. We use these results to extend to the case of general observables and prove worst-case lower bounds that are exponential in $n$ regardless of the locality of the observable, going beyond previous approaches whose lower bounds are derived for $1$-local observables, with natural generalizations that would depend on $n-k$ for $k$-local observables \cite{Huang_adv_2022, liu2025exponentialseparationsquantumlearning}.

Our starting point is the following hypothesis testing problem inspired by \cite{ye2025exponentialadvantagereplicaestimating}. We construct two ensembles of quantum states which are indistinguishable up to their third moment, but whose third moments differ by a constant gap. As a result, any algorithm which can estimate $\tr(\rho^3)$ can distinguish the ensembles, so a lower bound on the sample complexity of hypothesis testing provides a lower bound on the learning problem.

\begin{definition}[Many-vs-many distinguishing for third-moment estimation] \label{def:moment_testing_prob}

Consider the following two ensembles of $n$-qubit quantum states,
\begin{equation}
    \mathcal{E}_p = \left\{\frac{\mathds{1}}{2^{n+1}} + \frac{\ketbra{\psi_1}{\psi_1} + \ketbra{\psi_2}{\psi_2}}{4}\right\}, \quad     \mathcal{E}_q = \left\{\frac{\mathds{1}}{2^{n+1}} + \frac{\ketbra{\psi_1}{\psi_1}}{3} + \frac{\ketbra{\psi_2}{\psi_2}+\ketbra{\psi_3}{\psi_3}}{12}\right\} \ ,
\end{equation}
Equivalently, $\mathcal{E}_v = \rho_{mm}/2 + \sum_{j=\{1,2,3\}} v_j\ketbra{\psi_j}{\psi_j}$ (with $\rho_{mm}$ the $n$-qubit maximally mixed state), and we set $p = (1/4, 1/4, 0), q = (1/3, 1/12, 1/12)$. Suppose a state $\rho$ is sampled from either $\mathcal{E}_p$ or $\mathcal{E}_q$, with every $\ket{\psi_j}$ sampled independently from the Haar measure over $n$-qubit pure states. The many-vs-many distinguishing task is to decide, given many copies of $\rho$, which distribution $\rho$ is sampled from.
\end{definition}

\subsection{Lower bound without injection}

In \cite{ye2025exponentialadvantagereplicaestimating} it is demonstrated that any algorithm making joint measurements of at most two copies of $\rho$ at a time requires $\Omega(2^{n/2})$ measurements, while a protocol with access to three-copy measurements requires only constant measurements. Here, our focus is different: we argue that even given three-copy measurements, any strategy faces an exponential sample complexity lower bound in the presence of noise. Our first step is to argue that if a protocol uses a depth-1 unitary followed by computational-basis measurements, it must incur an exponential sample complexity lower bound. Note that the algorithm provided in \cite{ye2025exponentialadvantagereplicaestimating}, namely a version of the well-known generalized SWAP test, may be used to estimate $\tr(\rho^3)$ with $O(1)$ measurements using a depth-2 quantum circuit followed by computational basis measurement, so this result establishes that the generalized SWAP test is both sample-optimal and computationally optimal.

\begin{lemma}
\label{lemma:depth_1_lb}
    Any algorithm which can solve the many-vs.-many distinguishing task for third-moment estimation using three-copy measurements and a depth-1 quantum circuit requires $\Omega(2^{n/2})$ measurements. 
\end{lemma}

We defer the proof of this lemma to the end of this section, as it follows naturally from an intermediate step in the proof of the following theorem. This result establishes that any protocol which could incur a subexponential sample cost must utilize depth-$2$ or deeper quantum circuits, and thereby encounter at least two layers of noise. The following learning model therefore encompasses all possible strategies of relevance.

\begin{definition}[Learning tree for noisy three-copy algorithm with depth at least $2$]
\label{def:three_copy_tree}
    Any $\lambda$-noisy quantum algorithm with query access to three noisy copies of the state $\rho$, i.e. $\DN[\rho]^{\otimes 3}$ can be represented by a learning tree $\mathcal{T}$ as follows. At each node $u$ of $\mathcal{T}$, the algorithm can choose any $3n$-qubit depth-1 unitary $U_u$, and any  $3n$-qubit POVM of the form $M = \{F_s\}_u$. The transition rule is then
\begin{equation}
    p_\rho(v) = p_\rho(u)\tr(F_s \DNT ( U(\DN[\rho]^{\otimes 3})U^\dagger))\,.
\end{equation}
We also write this as 
\begin{equation}
    p_\rho(v) = p_\rho(u)\tr(F_s \Phi_u[\rho])\,,
\end{equation}
where $\Phi_u[\rho] = \DNT \circ \mathcal{U}_u\circ \DNT$ and $\mathcal{U}_u(\rho) = U_u\rho U_u^\dagger$ is the channel implemented by $U$.
\end{definition}

Before we proceed with the lower bound, we will prove one fact we will require regarding the channel $\Phi$.
\begin{lemma}\label{lemma:phi_trace_bound}
    Let $P$ denote an $n$-qubit Pauli operator and let $w(P)$ denote its weight, the number of non-identity single-qubit terms in $P$. Moreover, let $N$ be the map defined by its action on $n$-qubit Paulis: $N(P) = (1-\lambda)^{w(P) + \lceil w(P)/2\rceil}P$. Then for any operator $X$ on $3n$ qubits, $\tr(\Phi[X]^2) \leq \tr(N[X]^2)$.
\end{lemma}
\begin{proof}
We prove the claim for Hermitian $X$, which is the only case needed below. Write $a = 1-\lambda$ and suppress the node label on $\Phi$. Since $U$ has depth $1$, the $3n$ qubits are partitioned into disjoint blocks $B_1,\dots,B_m$, each of size $1$ or $2$, such that $U = \bigotimes_{j=1}^m U_j$.

For each $S \subseteq [m]$, let $\mathcal{P}_S$ denote the subspace spanned by Pauli strings which are identity on every block $B_j$ with $j \notin S$, and traceless on every block $B_j$ with $j \in S$. These subspaces are mutually orthogonal for different $S$, and both $\mathcal{U}$ and $\DNT$ preserve each $\mathcal{P}_S$. Indeed, conjugation by a local unitary $U_j$ preserves the decomposition of operators on a block into multiples of the identity and traceless operators, while $\DNT$ acts diagonally in the Pauli basis.

Expand $X$ in the $3n$-qubit Pauli basis as
\begin{equation}
    X = \sum_P x_P P \ ,
\end{equation}
and for each Pauli string $P$, let $S(P)$ be the set of blocks on which $P$ is non-identity. Since each block has size at most $2$, we have
\begin{equation}
    |S(P)| \geq \left\lceil \frac{w(P)}{2} \right\rceil\ .
\end{equation}
Now decompose
\begin{equation}
    \DNT[X] = \sum_{S \subseteq [m]} Y_S \ ,
\end{equation}
where $Y_S \in \mathcal{P}_S$. Explicitly,
\begin{equation}
    Y_S = \sum_{P:S(P)=S} a^{w(P)} x_P P \ .
\end{equation}
Because the subspaces $\mathcal{P}_S$ are orthogonal and preserved by both $\mathcal{U}$ and $\DNT$, we have
\begin{equation}
    \tr(\Phi[X]^2) = \sum_{S \subseteq [m]} \tr\!\Big(\DNT[\,\mathcal{U}(Y_S)]^2\Big) \ .
\end{equation}

It therefore suffices to bound the effect of the second noise layer on a fixed sector $\mathcal{P}_S$. Let $Z \in \mathcal{P}_S$. On a $1$-qubit block, the depolarizing channel multiplies every traceless Pauli by $a$. On a $2$-qubit block, $\mathcal D_\lambda^{\otimes 2}$ multiplies Paulis of local weight $1$ by $a$ and of local weight $2$ by $a^2$, so its Hilbert-Schmidt operator norm on the traceless subspace is at most $a$. Therefore, on a sector with $|S|$ active blocks,
\begin{equation}
    \tr\!\Big(\DNT[Z]^2\Big) \leq a^{2|S|}\tr(Z^2) \ .
\end{equation}
Applying this to $Z = \mathcal{U}(Y_S)$ and using that $\mathcal{U}$ is unitary, we obtain
\begin{equation}
    \tr(\Phi[X]^2) \leq \sum_{S \subseteq [m]} a^{2|S|}\tr(Y_S^2) \ .
\end{equation}

Finally, since distinct Pauli strings are orthogonal and $X$ is Hermitian, the Pauli coefficients $x_P$ are real, so
\begin{equation}
    \tr(Y_S^2) = \sum_{P:S(P)=S} a^{2w(P)} x_P^2 \tr(P^2) \ .
\end{equation}
Hence
\begin{align}
    \tr(\Phi[X]^2) &\leq \sum_P a^{2w(P)} a^{2|S(P)|} x_P^2 \tr(P^2) \\
    &\leq \sum_P a^{2w(P) + 2\lceil w(P)/2\rceil} x_P^2 \tr(P^2) \\
    &= \tr(N[X]^2) \ ,
\end{align}
where in the second line we used $|S(P)| \geq \lceil w(P)/2\rceil $. This proves the lemma.
\end{proof}
Now we are prepared to state the main lower bound of this section.
\begin{theorem}
\label{thm:noisy_moment_lb}
    Any $\lambda$-noisy quantum algorithm with query access to $\DN[\rho]^{\otimes 3}$ requires 
    \begin{equation}
        \Omega\left(\min \left\{2^{n/2}, \frac{(2^n+1)^2(2^n+2)^2}{R_n(\eta)}\right\}\right)
    \end{equation}
    experiments to solve the distinguishing task with high probability, where
    \begin{equation}
        R_n(\eta) = 2(1+9\eta^6 + 6\eta^{10})^n + 2(1+9\eta^6 - 6\eta^{10})^n - 12(1+3\eta^6)^n + 8.
    \end{equation}   
    for $\eta = 1-\lambda$. The second term goes to $\infty$ as $\lambda\rightarrow 1$, but for small constant $\lambda$ the lower bound scales like 
    \begin{equation}
        \Omega\left(\left(\frac{16}{1+9(1-\lambda)^6 + 6(1-\lambda)^{10}}\right)^n\right) \ .
    \end{equation}

\end{theorem}
\begin{proof}
    By Lemma \ref{lemma:depth_1_lb}, it follows that we need only prove the claim for circuits of depth $\geq 2$, placing us in the setting of the learning tree model given in Definition \ref{def:three_copy_tree}. Our proof uses the Heisenberg learning tree method introduced in Section \ref{app:Heisenberg_learning_tree}. Let $T$ denote the depth of the tree, and let $d_{M}(\rho, \sigma)$ denote the total variation distance between classical output distributions obtained by measuring POVM $M$ on states $\rho, \sigma$. Moreover, let $\sigma_a = \mathbb{E}_{\mathcal{E}_a}[\rho^{\otimes 3}]$ and $\Sigma_a = \mathbb{E}_{\mathcal{E}_a}[\rho^{\otimes 3T}]$, for $a\in\{p, q\}$. It is shown in Lemma 34 of \cite{ye2025exponentialadvantagereplicaestimating} that 
    \begin{equation}
        \max_M d_M(\Sigma_a, \sigma_a^{\otimes T}) \leq  \frac{(3T/2)^2 + 3T/2}{2^n} \ ,
    \end{equation}
    where the maximum is taken over all $\mathcal{M}_{3, T}$ POVMs (as in Definition \ref{def:m_k,t_POVMs}). Our noisy depth-2 model is included in these, so this bound applies. By triangle inequality,
    \begin{equation}
        \label{eq:tv_triangle}
        d_M(\Sigma_p, \Sigma_q) \leq  d_M(\Sigma_p, \sigma_p^{\otimes T}) + d_M(\Sigma_q, \sigma_q^{\otimes T}) + d_M(\sigma_p^{\otimes T}, \sigma_q^{\otimes T}) \leq 2\frac{(3T/2)^2 + 3T/2}{2^n} + d_M(\sigma_p^{\otimes T}, \sigma_q^{\otimes T})\ .
    \end{equation}
    Note that for $M \in \mathcal{M}_{3, T}$, $\max_M d_M(\sigma_p^{\otimes T}, \sigma_q^{\otimes T})$ is simply an upper bound on the total variation between output leaf distributions of a new learning tree (still of depth $T$) which measures copies of \textit{exact} states $\sigma_p$ or $\sigma_q$ at each node. As such, the first step of the Heisenberg method is complete, and we can focus on a one-vs.-one distinguishing task corresponding to this new learning tree. Now define
    \begin{equation}
    \tilde{\sigma}_a^u \coloneqq \Phi_u[\sigma_a], \quad
    \Delta^u \coloneqq \tilde\sigma_p^u - \tilde\sigma_q^u \ .
    \end{equation}
    where we have absorbed the adaptive, noisy behavior of the original learning tree model into the input states of the simplified one-vs.-one tree. For the second step, we obtain the requisite bounds on $\tilde \sigma_q^u$ and $\tr((\Delta^u)^2)$. First, we have the following operator inequality:
    \begin{equation}
        \sigma_q \succeq \frac{\rho_{mm}^{\otimes 3}}{8} = \frac{1}{8\cdot2^{3n}}\mathds{1}_{3n} \ , \label{eq:sigma_q_opineq}
    \end{equation}
    which implies $\tilde\sigma_q^u \succeq (8\cdot2^{3n})^{-1}\mathds{1}_{3n}$ since $\Phi$ is unital and CPTP. Letting $\mu = (8\cdot2^{3n})^{-1}$, we have
    \begin{equation}
        \sum_s \frac{\tr(F_s\Delta^u)^2}{\tr(F_s\tilde\sigma_q^u)} \leq \mu^{-1}\sum_s \frac{\tr(F_s\Delta^u)^2}{\tr(F_s)} \ .
    \end{equation}
    Next, let $x$ be a set of indices and let $S_x$ denote the sum of all global permutation operators on $|x|$ copies of an $n$-qubit Hilbert space labeled by $x$. The relevant examples are
    \begin{equation}
        S_i = \frac{\mathds{1}_n}{2^n}, \quad S_{ij} = \frac{\mathds{1}_{2n} + \swap_{ij}}{2^n(2^n+1)}, \quad S_{ijk} = \frac{\mathds{1}_{3n} + \swap_{ij}+\swap_{jk}+\swap_{ki} + \pi_{ijk} + \pi_{ijk}^{-1}}{2^n(2^n+1)(2^n+2)}
    \end{equation}
    Here $\swap_{ij}$ denotes the $\swap$ operator acting on two copies of the $n$-qubit Hilbert space labeled by $i$ and $j$, and $\pi_{123}$ denotes the cyclic permutation on three such copies. Now we define the operator
    \begin{equation}
        \Delta_3 = S_{123} - S_{12}S_3 - S_{13}S_2 - S_{23}S_1 + 2S_1^{\otimes 3} \ .
    \end{equation}
    By Lemma 33 in \cite{ye2025exponentialadvantagereplicaestimating}, we have the exact relationship $\Delta^u = \Phi_u[\Delta_3]/18$. For the second step of the Heisenberg method, it remains only to bound $\tr((\Delta^u)^2)$, and by Lemma \ref{lemma:phi_trace_bound}, we can instead consider the simpler object $\tr\!\big(N(\Delta_3)^2\big)$. Using the explicit formulas for $S_i$, $S_{ij}$, and $S_{123}$, we first rewrite $\Delta_3$ in a more convenient form:
    \begin{equation}
        \label{eq:delta_3_decomp}
        \Delta_3 = \frac{2^{2n}(\pi_{123}+\pi_{123}^{-1}) - 2^{n+1}(\swap_{12}+\swap_{13}+\swap_{23}) + 4\mathds{1}_{3n}}{2^{3n}(2^n+1)(2^n+2)} \ .
    \end{equation}
    Now work site-by-site on the three-copy Hilbert space of a single physical qubit, and temporarily let $\mathds{1}$ denote the identity on this $8$-dimensional space. Define
    \begin{equation}
        A_{12} = XXI + YYI + ZZI, \quad A_{13} = XIX + YIY + ZIZ, \quad A_{23} = IXX + IYY + IZZ
    \end{equation}
    and
    \begin{equation}
        B = XYZ - XZY + YZX - YXZ + ZXY - ZYX \ .
    \end{equation}
    Then on one site,
    \begin{equation}
        \swap_{12} = \frac{\mathds{1}+A_{12}}{2}, \quad \swap_{13} = \frac{\mathds{1}+A_{13}}{2}, \quad \swap_{23} = \frac{\mathds{1}+A_{23}}{2},
    \end{equation}
    while
    \begin{equation}
        \pi_{123} = \frac{\mathds{1}+A_{12}+A_{13}+A_{23}+iB}{4}, \quad \pi_{123}^{-1} = \frac{\mathds{1}+A_{12}+A_{13}+A_{23}-iB}{4} \ .
    \end{equation}
    Since the global permutation operators factor over the $n$ sites, and since $N$ multiplies Pauli strings of weight $2$ by $(1-\lambda)^3$ and of weight $3$ by $(1-\lambda)^5$, we obtain, letting $\eta = 1-\lambda$ for brevity,
    \begin{align}
        &N(\swap_{12}) = \frac{(\mathds{1}+\eta^3 A_{12})^{\otimes n}}{2^n}, \quad N(\swap_{13}) = \frac{(\mathds{1}+\eta^3 A_{13})^{\otimes n}}{2^n}, \quad N(\swap_{23}) = \frac{(\mathds{1}+\eta^3 A_{23})^{\otimes n}}{2^n}, \\
        &N(\pi_{123}) = \frac{(\mathds{1}+\eta^3(A_{12}+A_{13}+A_{23})+i\eta^5 B)^{\otimes n}}{4^n}, 
        N(\pi_{123}^{-1}) = \frac{(\mathds{1}+\eta^3(A_{12}+A_{13}+A_{23})-i\eta^5 B)^{\otimes n}}{4^n} \ .
    \end{align}
    For notational convenience, define
    \begin{equation}
        X_\pm = \mathds{1}+\eta^3(A_{12}+A_{13}+A_{23}) \pm i\eta^5 B, \quad Y_{ij} = \mathds{1}+\eta^3 A_{ij} \ .
    \end{equation}
    Substituting the above formulas into $\Delta_3$ gives
    \begin{equation}
        N(\Delta_3) = \frac{X_+^{\otimes n} + X_-^{\otimes n} - 2Y_{12}^{\otimes n} - 2Y_{13}^{\otimes n} - 2Y_{23}^{\otimes n} + 4\mathds{1}_{3n}}{2^{3n}(2^n+1)(2^n+2)} \ .
    \end{equation}
    We want to control the square of this object. To do so, we explicitly compute traces of products of the previously defined operators. The Pauli strings appearing in $A_{12},A_{13},A_{23},B$ are all mutually orthogonal, so with the inner product $2^{-3}\tr(\cdot\,\cdot)$ (where the $2^{-3}$ factor is normalizing by dimension) we have
    \begin{equation}
        2^{-3}\tr(A_{ij}^2) = 3, \quad 2^{-3}\tr(B^2) = 6, \quad 2^{-3}\tr(A_{ij}A_{k\ell}) = 0 \ \text{for } (ij)\neq(k\ell), \quad 2^{-3}\tr(BA_{ij}) = 0 \ .
    \end{equation}
    It follows that
    \begin{align}
        2^{-3}\tr(X_+X_-) &= 1 + 9\eta^6 + 6\eta^{10}, \\
        2^{-3}\tr(X_\pm^2) &= 1 + 9\eta^6 - 6\eta^{10}, \\
        2^{-3}\tr(X_\pm Y_{ij}) &= 1 + 3\eta^6, \\
        2^{-3}\tr(Y_{ij}^2) &= 1 + 3\eta^6, \\
        2^{-3}\tr(Y_{ij}Y_{k\ell}) &= 1 \  \text{for } (ij)\neq(k\ell) \ .
    \end{align}
    Since traces factor over tensor products, $2^{-3n}\tr(A^{\otimes n}B^{\otimes n}) = \big(2^{-3}\tr(AB)\big)^n$. Expanding the square of the numerator and substituting the derived expressions yields
    \begin{align}
        2^{3n}\tr\!\big(N(\Delta_3)^2\big) &= \frac{1}{(2^n+1)^2(2^n+2)^2} \Big(2(1+9\eta^6+6\eta^{10})^n + 2(1+9\eta^6-6\eta^{10})^n \\
        &\qquad\qquad\qquad - 24(1+3\eta^6)^n + 12(1+3\eta^6)^n + 16 + 24 - 48 + 16\Big) \\
        &= \frac{1}{(2^n+1)^2(2^n+2)^2} \Big(2(1+9\eta^6+6\eta^{10})^n + 2(1+9\eta^6-6\eta^{10})^n - 12(1+3\eta^6)^n + 8\Big) \\
        &= \frac{R_n(\eta)}{(2^n+1)^2(2^n+2)^2} \ .
    \end{align}
    Equivalently,
    \begin{equation}
        \tr\!\big(N(\Delta_3)^2\big) = \frac{R_n(\eta)}{2^{3n}(2^n+1)^2(2^n+2)^2} \ , 
        \label{eq:tr_N_bound}
    \end{equation}
    which is the desired bound on the trace term. Finally leveraging the Heisenberg Lemma \ref{thm:Heisenberg}, this brings us to the bound
    \begin{equation}
    \label{eq:likelihood_nbound_thm1}
        \mathbb{E}_s[(L(s|u) - 1)^2]\leq \frac{
        2^{3n}\cdot 8}{18^2}\tr\!\big(N(\Delta_3)^2\big) \leq \frac{
        8}{18^2}\frac{R_n(\eta)}{(2^n+1)^2(2^n+2)^2} \ ,
    \end{equation} 
    where
  \begin{align}
        R_n(\eta) &= 2(1+9\eta^6 + 6\eta^{10})^n + 2(1+9\eta^6 - 6\eta^{10})^n - 12(1+3\eta^6)^n + 8 \\
        &= O((1+9(1-\lambda)^6 + 6(1-\lambda)^{10})^n \ .
    \end{align}
    Thus we conclude that 
    \begin{equation}
        T \geq \Omega\left(\left(\frac{(2^n+1)^2(2^n+2)^2}{R_n(\eta)}\right)\right) = \Omega\left(\left(\frac{16}{1+9(1-\lambda)^6 + 6(1-\lambda)^{10}}\right)^n\right)
    \end{equation}
    is necessary for a success probability of $2/3$. Recalling our original total variation bound,
    \begin{equation}
        d_M(\Sigma_p, \Sigma_q) \leq \frac{(3T/2)^2 + 3T/2}{2^n} + d_M(\sigma_p^{\otimes T}, \sigma_q^{\otimes T}) \ ,
    \end{equation}
    the first term imposes the independent condition $T\geq \Omega(2^{n/2})$. This concludes the proof.
\end{proof}
We now complete the argument by proving Lemma \ref{lemma:depth_1_lb}.
\begin{proof}[Proof of Lemma \ref{lemma:depth_1_lb}]
In the depth-1 setting, we assume all POVM elements can be constructed by noiseless depth-1 quantum circuits with arbitrary fresh ancillas initialized in a product state. Now we return to the proof of Theorem \ref{thm:noisy_moment_lb} and note that Equation \eqref{eq:tv_triangle} applies directly for the depth-1 case, simply restricting the allowable POVMs to this depth-1 subset. That is, let $M$ be a partial matching on the $3n$ system sites, and define
\begin{equation}
    \mathcal{B}_M = \left\{F_z = \bigotimes_{e\in M} \ket{\phi_{e, z_e}}\bra{\phi_{e, z_e}}\otimes \bigotimes_{u\notin M} E^{(u)}_{z_u}\right\}
\end{equation}
where $e$ is an edge within the matching, $u\notin M$ denotes a system site not incident to any edge in $M$, and for each $e$ we have an orthonormal basis spanning the local $2$-qubit Hilbert space $\{\ket{\phi_{e, j}}\}_{j=1}^4$. For each unmatched system site $u$, $\{E^{(u)}_{z_u}\}_{z_u}$ is an arbitrary single-qubit POVM. Note that $z$ restricted to the two indices corresponding to edge $e$ defines $z_e$, while $z_u$ denotes the local outcome of the singleton POVM at $u$.

This form describes every refined depth-1 POVM in the circuit model. Indeed, after adding ancillas, a depth-1 circuit partitions the system and ancilla qubits into disjoint blocks of size at most two. If a block contains two system qubits, then the local unitary followed by computational-basis measurement induces a projective measurement in some orthonormal two-qubit basis on those two system qubits. If a block contains one system qubit and one ancilla, or just one system qubit, then after tracing over the fixed ancilla input it induces a single-qubit POVM on that system qubit. Blocks containing only ancillas have outcome distributions independent of the input state and can be absorbed into classical postprocessing. Finally, any coarse-graining of the refined outcomes can only decrease total variation distance since any coarse-grained POVM can be simulated by a POVM of the above refined form (see e.g. Lemma 4.8 in \cite{chen2021exponentialseparationslearningquantum}). Thus $\mathcal{B} = \bigcup_{M\in \mathcal{M}}\mathcal{B}_M$, where $\mathcal{M}$ denotes the set of all partial matchings on $3n$ system sites together with all choices of local bases and singleton POVMs, describes any valid depth-1 POVM.

With this formalism and the definitions $\sigma_a = \mathbb{E}_{\mathcal{E}_a}[\rho^{\otimes 3}]$ for $a\in\{p, q\}$ and $\sigma_p - \sigma_q = \Delta_3/18$, we have
\begin{equation}
\label{eq:delta_3_tv_d1}
    d_{\mathcal{B}_M}(\sigma_p, \sigma_q) = \frac{1}{36}\sum_z |\tr(F_z\Delta_3)|
\end{equation}
for any particular refined depth-1 POVM $\mathcal{B}_M$. By equation \eqref{eq:delta_3_decomp}, we see that $\Delta_3$ can be decomposed into a constant number of identity, $\swap$, and $\pi_{123}$ terms on the $3n$ qubits. For the identity term, positivity and completeness of the POVM give
\begin{equation}
    \sum_z |\tr(F_z\mathds{1})| = \sum_z \tr(F_z) = 2^{3n}.
\end{equation}
Similarly, because $\swap$ is a Hermitian unitary, the variational characterization of the trace norm gives
\begin{equation}
    \sum_z |\tr(F_z\swap)| \leq \|\swap\|_1 = 2^{3n}.
\end{equation}
It remains to bound $\sum_z|\tr(F_z\pi_{123})|$ (its inverse will follow the same bound).
\begin{lemma}
    For any refined depth-1 POVM $\mathcal{B}_M$ of the form above,
    \begin{equation}
        \sum_z|\tr(F_z\pi_{123})| \leq 2^{9n/4} \ .
    \end{equation}
\end{lemma}
\begin{proof}
We begin with a high-level idea of the proof. Since each $F_z$ decomposes into blocks of two-qubit projectors and single-qubit POVM elements, we can visualize the refined measurement corresponding to $\mathcal{B}_M$ as a graph on $3n$ vertices, with matched edges carrying two-qubit projective measurements and unmatched vertices carrying singleton POVM measurements. Each matched edge $e$ is accompanied by a family of $2$ by $2$ matrices $\{A_{e,j}\}_{j=1}^4$, one for each projective outcome of that two-qubit block, while each singleton site $u$ is accompanied by a family of positive $2$ by $2$ matrices $\{E^{(u)}_{\alpha}\}_{\alpha\in\Omega_u}$ describing the local one-qubit POVM. For a fixed global outcome $z$, we write $A_e=A_{e,z_e}$ and $E^{(u)}_{z_u}$ for the particular local objects selected by that outcome.. We then notice that each term in the expression we wish to bound is a contraction of the cyclic permutation operator $\pi_{123}$ on this graph, and that any such term can be decomposed into a product of contractions along open or closed loops on the graph. In this way, we collapse a tensor network over the refined matching graph into a scalar, obtaining the final bound.

To begin, for every matched edge $e$ we write
\begin{equation}
    \ket{\phi_{e, j}} = \sum_{a, b \in \{0, 1\}} (A_{e, j})_{ab}\ket{a}\ket{b}
\end{equation}
in the computational basis, where the four matrices $A_{e, j}, j \in \{1,2,3,4\}$ are orthonormal with respect to the Hilbert-Schmidt inner product because $\{\ket{\phi_{e, j}}\}_{j=1}^4$ is defined to be an orthonormal basis. This allows us to identify the projector onto each $2$-qubit block with a 2 by 2 matrix. Moreover, as for any orthonormal basis of 2 by 2 matrices, we have
\begin{equation}
    \sum_{j=1}^4 A_{e, j}^\dagger A_{e, j} = \sum_{j=1}^4 A_{e, j} A_{e, j}^\dagger = 2\mathds{1}_2 \ .
\end{equation}
We will also use the defining property of the singleton POVMs,
\begin{equation}
    E^{(u)}_{z_u}\succeq 0, \quad \sum_{z_u} E^{(u)}_{z_u} = \mathds{1}_2 \ .
\end{equation}
We will use two operator contraction bounds that will be used to collapse the tensor network. For a two-qubit projective block and any $2$ by $2$ matrix $X$, we have
\begin{equation}
    \sum_{j=1}^4 |\tr(XA_{e,j})| \leq 2\|X\|_F
\end{equation}
and
\begin{equation}
    \sum_{j=1}^4 \|XA_{e,j}\|_F \leq 2\sqrt{2}\|X\|_F, \quad \sum_{j=1}^4 \|A_{e,j}X\|_F \leq 2\sqrt{2}\|X\|_F .
\end{equation}
The first estimate follows from Cauchy-Schwarz and orthonormality of the matrices $A_{e,j}$. For example,
\begin{equation}
    \left(\sum_{j=1}^4 |\tr(XA_{e,j})|\right)^2 \leq 4\sum_{j=1}^4 |\tr(XA_{e,j})|^2 \leq 4\|X\|_F^2 .
\end{equation}
For the second estimate, we use Cauchy-Schwarz and $\sum_j A_{e,j}A_{e,j}^\dagger = 2\mathds{1}_2$:
\begin{equation}
    \left(\sum_{j=1}^4 \|XA_{e,j}\|_F\right)^2 \leq 4\sum_{j=1}^4 \tr(XA_{e,j}A_{e,j}^\dagger X^\dagger) = 8\|X\|_F^2 .
\end{equation}
The bound with $A_{e,j}$ on the left follows identically. For a singleton POVM block $\{E_j\}_j$ and any $2$ by $2$ matrix $X$, we have
\begin{equation}
    \sum_j |\tr(XE_j)| \leq \sqrt{2}\|X\|_F
\end{equation}
and
\begin{equation}
    \sum_j \|XE_j\|_F \leq \sqrt{2}\|X\|_F, \quad \sum_j \|E_jX\|_F \leq \sqrt{2}\|X\|_F .
\end{equation}
Indeed,
\begin{equation}
    |\tr(XE_j)| = |\tr(XE_j^{1/2}E_j^{1/2})| \leq \|XE_j^{1/2}\|_F\|E_j^{1/2}\|_F ,
\end{equation}
so Cauchy-Schwarz gives
\begin{equation}
    \left(\sum_j |\tr(XE_j)|\right)^2 \leq \left(\sum_j \tr(XE_jX^\dagger)\right)\left(\sum_j \tr(E_j)\right) = 2\|X\|_F^2 .
\end{equation}
For the norm estimate, ignoring terms with $\tr(E_j)=0$, Cauchy-Schwarz gives
\begin{equation}
    \left(\sum_j \|XE_j\|_F\right)^2 \leq \left(\sum_j \tr(E_j)\right)\left(\sum_j \frac{\|XE_j\|_F^2}{\tr(E_j)}\right).
\end{equation}
Since $E_j\succeq 0$, we have $E_j^2\preceq \tr(E_j)E_j$, and therefore
\begin{equation}
    \sum_j \frac{\|XE_j\|_F^2}{\tr(E_j)} = \sum_j \frac{\tr(XE_j^2X^\dagger)}{\tr(E_j)} \leq \sum_j \tr(XE_jX^\dagger) = \|X\|_F^2 .
\end{equation}
Using $\sum_j \tr(E_j)=2$ gives the desired result. The bound with $E_j$ on the left follows identically. All stated operator bounds are unchanged if the matrices are transposed, conjugated, or adjointed; for singleton POVM elements, transposition preserves positivity and the condition that the effects sum to $\mathds{1}_2$.

Now consider the single term $\tr(F_z\pi_{123})$ and insert the identity in computational basis strings:
\begin{equation}
    \tr(F_z\pi_{123}) = \sum_{x\in\{0, 1\}^{3n}}\bra{x}F_z\ket{\pi_{123}^{-1}x}.
\end{equation}
Using the local form of $F_z$, we have
\begin{equation}
    \tr(F_z\pi_{123}) = \sum_{x\in\{0, 1\}^{3n}}\prod_{e = (u,v)\in M} ((A_e)_{x_u, x_v})^*(A_e)_{\pi(x_u), \pi(x_v)}\prod_{u\notin M}\bra{x_u}E^{(u)}_{z_u}\ket{\pi(x_u)}
\end{equation}
where we write $A_e$ to denote the specific choice of $j$ corresponding to the $\ket{\phi_{e, j}}$ that appears in the outcome $z$, and we suppress the indices on the permutation.

Now we can imagine a contraction graph for the sum, where every two-qubit block matrix $A_e$ or singleton POVM element $E^{(u)}_{z_u}$, along with its related conjugates, transposes, and adjoints, is a node and an edge is drawn whenever a given index $x_q$ is repeated in the product. Suppose this contraction graph has connected components $C_1, ... , C_m$. Because we take the product over all local blocks, note that each index $x_q$ appears at most twice. For this reason, all connected components are either closed or open loops. Now we can write
\begin{equation}
     \tr(F_z\pi_{123}) = \prod_{j =1}^m f_{C_j}(z_{C_j})
\end{equation}
where $z_{C_j}$ is the set of local outcome indices that appear in component $C_j$, and
\begin{equation}
    f_{C_j}(z_1, ..., z_l) = \tr(B_{1, z_1}...B_{l, z_l})
\end{equation}
if $C_j$ is a closed loop, or
\begin{equation}
    f_{C_j}(z_1, ..., z_l) =\bra{\alpha_C}B_{1, z_1}...B_{l, z_l}\ket{\beta_C}
\end{equation}
for an open loop, where $\ket{\alpha_C}$ and $\ket{\beta_C}$ are fixed computational-basis endpoint vectors. Each $B_{i,z_i}$ is either one of the matrices associated with a two-qubit projective block or one of the matrices associated with a singleton POVM block. We are left with
\begin{equation}
\label{eq:conn_comps}
    \sum_z|\tr(F_z\pi_{123})| = \sum_z \left|\prod_C f_C(z_C)\right| = \sum_{z_{C_1}}\sum_{z_{C_2}}...\sum_{z_{C_m}}\left|\prod_{j=1}^m f_{C_j}(z_{C_j})\right| = \prod_{C}\sum_{z_C} |f_C(z_C)|
\end{equation}
where we use the fact that all values $f_C(z_C)$ are completely independent across clusters, so we can iteratively apply e.g.
\begin{equation}
    \sum_{z_{C_1}, z_{C_2}} |f_{C_1}(z_{C_1})||f_{C_2}(z_{C_2})| = \left(\sum_{z_{C_1}} |f_{C_1}(z_{C_1})|\right)\left(\sum_{z_{C_2}}|f_{C_2}(z_{C_2})|\right) \ .
\end{equation}

Now consider some connected component $C$. Let $\ell_2(C)$ denote the number of two-qubit projective blocks appearing in $C$, and let $\ell_1(C)$ denote the number of singleton POVM blocks appearing in $C$. First suppose $\ell_2(C)\geq 1$. If $C$ is an open loop, then repeated use of the two local norm estimates gives
\begin{align}
    \sum_{z_C}|f_C(z_C)| &\leq \sum_{z_1,...,z_l}\|B_{1,z_1}...B_{l,z_l}\ket{\beta_C}\|_2\\
    &\leq (2\sqrt{2})^{\ell_2(C)}(\sqrt{2})^{\ell_1(C)}\\
    &= 2^{(3\ell_2(C)+\ell_1(C))/2}.
\end{align}
If $C$ is a closed loop, cyclically order the product so that the final matrix comes from a two-qubit projective block. Then
\begin{align}
    \sum_{z_C}|f_C(z_C)| &= \sum_{z_C}|\tr(B_{1,z_1}...B_{l,z_l})|\\
    &\leq 2\sum_{z_1,...,z_{l-1}}\|B_{1,z_1}...B_{l-1,z_{l-1}}\|_F\\
    &\leq 2\sqrt{2}(2\sqrt{2})^{\ell_2(C)-1}(\sqrt{2})^{\ell_1(C)}\\
    &= 2^{(3\ell_2(C)+\ell_1(C))/2}.
\end{align}
In the second line we used the trace estimate for the final two-qubit projective block, and in the third line we collapsed the remaining matrices using the norm estimates above. Hence any component containing at least one two-qubit projective block satisfies
\begin{equation}
    \sum_{z_C}|f_C(z_C)| \leq 2^{(3\ell_2(C)+\ell_1(C))/2}.
\end{equation}

It remains to consider components with $\ell_2(C)=0$. Such a component can not connect different physical sites, because $\pi_{123}$ only cycles the three copies of the same physical site and no two-qubit matching edge is present to connect this site to any other site. Thus a pure singleton component consists of the three singleton POVM blocks associated with one physical site. Writing their local effects as $\{E_a\}_a$, $\{F_b\}_b$, and $\{G_c\}_c$, we obtain
\begin{align}
    \sum_{a,b,c}|\tr(E_aF_bG_c)| &\leq \sqrt{2}\sum_{a,b}\|E_aF_b\|_F\\
    &\leq 2\sum_a\|E_a\|_F\\
    &\leq 4.
\end{align}
The first line uses the singleton trace estimate to sum over $c$, the second line uses the singleton norm estimate to sum over $b$, and the final line uses the singleton norm estimate with $X=\mathds{1}_2$ to sum over $a$.

Returning to Equation \eqref{eq:conn_comps}, let $L_2$ be the total number of two-qubit projective blocks, let $L_1^{\mathrm{mix}}$ be the number of singleton POVM blocks appearing in components that contain at least one two-qubit projective block, and let $L_1^{\mathrm{sing}}$ be the number of singleton POVM blocks appearing in pure singleton components. Since every pure singleton component contains exactly three singleton blocks, the product of all pure singleton contributions is bounded by $4^{L_1^{\mathrm{sing}}/3}$. Combining the component bounds gives
\begin{equation}
    \sum_z|\tr(F_z\pi_{123})| \leq 4^{L_1^{\mathrm{sing}}/3}2^{(3L_2+L_1^{\mathrm{mix}})/2}.
\end{equation}
Because the refined measurement acts on all $3n$ system qubits, we have
\begin{equation}
    2L_2 + L_1^{\mathrm{sing}} + L_1^{\mathrm{mix}} = 3n.
\end{equation}
Therefore
\begin{align}
    \frac{2L_1^{\mathrm{sing}}}{3}+\frac{3L_2}{2}+\frac{L_1^{\mathrm{mix}}}{2} &= 2n+\frac{L_2}{6}-\frac{L_1^{\mathrm{mix}}}{6}\\
    &\leq 2n+\frac{L_2}{6}\\
    &\leq \frac{9n}{4},
\end{align}
where in the final line we used $L_2\leq 3n/2$. We conclude that
\begin{equation}
    \sum_z|\tr(F_z\pi_{123})| \leq 2^{9n/4}
\end{equation}
as claimed.
\end{proof}
\noindent Returning to Equation \eqref{eq:delta_3_tv_d1} and substituting our upper bounds, applying the triangle inequality to the decomposition of $\Delta_3$, we are left with
\begin{align}
    d_{\mathcal{B}_M}(\sigma_p, \sigma_q) &\leq \frac{2^{2n}(2\cdot2^{9n/4}) + 2^{n+1}(3\cdot 2^{3n})+ 4\cdot 2^{3n}}{36\cdot 2^{3n}(2^n+1)(2^n+2)}\\
    &\leq \frac{2^{4.25n+1} + 3\cdot2^{4n+1}+2^{3n+2}}{2^{5n}} \\
    &\leq 2\cdot 2^{-3n/4} + 6\cdot 2^{-n} + 4\cdot 2^{-2n} \\
    &= O(2^{-3n/4})\ .
\end{align}
Since this bound holds for every refined depth-1 POVM and coarse-graining can only decrease total variation distance, it holds for every depth-1 POVM implementable in the ancilla-assisted circuit model. Recalling Equation \eqref{eq:tv_triangle}, we note that the first term imposes a sample lower bound of $2^{n/2} < 2^{3n/4}$; as a result, we obtain the final lower bound of $\Omega(2^{n/2})$.
\end{proof}

\subsection{Extension of lower bound to cubic observables}
The previous discussion establishes a noise-dependent lower bound for the particular task of estimating $\tr(\rho^3)$. Here we demonstrate that such a lower bound can also be established for more general observable-estimation tasks of the form $\tr(O \rho^3)$. While this task itself may be useful in certain experimental settings where higher-order correlations are informative, we note that this also includes problems where one hopes to learn from an ideal pure state, but the presence of noise may provide a mixed state. In many settings, the noise may transform the ideal pure state approximately as
\begin{equation}
    \ketbra{\psi}{\psi} \longrightarrow \rho \approx (1-p)\ketbra{\psi}{\psi} + p\sum_{j=1} c_j \ketbra{\phi_j}{\phi_j}, \quad \sum_{j} c_j = 1 \ .
\end{equation}
Moreover, in the presence of a dominant noise source, it may be the case that the coefficients $c_j$ decay quickly, so a rank-2 or rank-3 mixed state is a good approximation for the noisy state. This is the usual regime in which quantum principal component analysis (qPCA) \cite{Lloyd_2014} and virtual distillation \cite{Huggins_2021} operate, and here the ability to access polynomials of the density matrix is used as an error-mitigation primitive. One can also take our result to be a lower bound on the ability to execute these strategies for error-mitigation of physical states in the presence of processing noise, where the given density matrix is well approximated as at most rank 3. 

We begin with the following distinguishing task, which is a simple extension of Definition \ref{def:moment_testing_prob}.
\begin{definition}[Many-vs-many distinguishing for cubic observable estimation]
Let $\mathcal{E}_p, \mathcal{E}_q$, and more generally $\mathcal{E}_v$ on $n-1$ qubits be as in Definition \ref{def:moment_testing_prob}. Define the single-qubit state
\begin{equation}
    \tau = \frac{1}{2}\left(\mathds{1}+\frac{Z}{2}\right)
\end{equation}
Then define the family of ensembles over $n$-qubit states
\begin{equation}
    \tilde{\mathcal{E}}_v = \{\tilde \rho = \tau \otimes \rho: \rho \sim \mathcal{E}_v\} \ .
\end{equation}
The hypothesis testing task is now to decide whether a given $\tilde \rho$ was sampled from either $\tilde{\mathcal{E}}_p$ or $\tilde{\mathcal{E}}_q$. 
\end{definition}

Any algorithm that can measure $\tr(O \rho^3)$ for a fixed choice of $O$ to constant accuracy with high probability can solve this distinguishing task. This is because choosing $O = Z_1$, we observe
\begin{equation}
    \tr(Z_1\tilde{\rho}^3) = \tr(Z_1\tau^3)\tr({\rho}^3)
\end{equation}
and we can directly compute $\tr(Z_1\tau^3) = 13/32$,  so using the fact that the value of $\tr(\rho^3)$ differs by a constant when $\rho$ is sampled from $\mathcal{E}_p$ vs. $\mathcal{E}_q$, we find that $\tr(Z_1\tilde{\rho}^3)$ also differs by a constant gap with $\tilde{\rho}$ sampled from $\tilde{\mathcal{E}}_p$ vs. $\tilde{\mathcal{E}}_q$.

Then, the following corollary of Theorem \ref{thm:noisy_moment_lb} implies a lower bound on estimating $\tr(O\rho^3)$ for fixed $O$.

\begin{corollary}
\label{cor:z_1_lb}
    Any $\lambda$-noisy quantum algorithm with query access to $\DN[\rho]^{\otimes 3}$ using circuits of depth at least $2$ requires 
    \begin{equation}
        \Omega\left(\min \left\{2^{n/2}, \frac{(2^{n-1}+1)^2(2^{n-1}+2)^2}{R_{n-1}(\eta)}\right\}\right)
    \end{equation}
    experiments to solve the distinguishing task with high probability, where $R(\eta)$ is defined as in Theorem \ref{thm:noisy_moment_lb}.
\end{corollary}
\begin{proof}
    As before, let $\sigma_a = \mathbb{E}_{\tilde{\mathcal{E}}_a}[\tilde\rho^{\otimes 3}]$ and $\Sigma_a = \mathbb{E}_{\mathcal{\tilde E}_a}[\tilde \rho^{\otimes 3T}]$, for $a\in\{p, q\}$. By triangle inequality, equation \eqref{eq:delta_3_tv_d1} holds again with these redefined states:
    \begin{equation}
        d_M(\Sigma_p, \Sigma_q) \leq \frac{(3T/2)^2 + 3T/2}{2^n} + d_M(\sigma_p^{\otimes T}, \sigma_q^{\otimes T})\ .
    \end{equation}
    Once again taking 
    \begin{equation}
        \tilde{\sigma}_a^u \coloneqq \Phi_u[\sigma_a], \quad \tilde \Delta^u \coloneqq \tilde\sigma_p - \tilde\sigma_q \ ,
    \end{equation}
    we can use $\tau \succeq \mathds{1}/4$ and equation \eqref{eq:sigma_q_opineq} to find 
    \begin{equation}
        \sigma_q \succeq \frac{1}{2^{3n+6}}\mathds{1}_{3n} \ .
    \end{equation}
    Setting $\mu = 2^{-(3n+6)}$ and following the proof of Theorem \ref{thm:noisy_moment_lb}, we arrive at the analogous form of equation \eqref{eq:likelihood_nbound_thm1}:
    \begin{equation}
        \mathbb{E}_s[(L(s|u)-1)^2] \leq \frac{2^{3n+6}}{18^2} \tr\!\big(N(\tilde \Delta_3)^2\big)
    \end{equation}
    with our redefined $\tilde \Delta_3$ instead of the original $\Delta_3$. Observe that 
    \begin{equation}
        \tau^{\otimes 3} = 2^{-3}\sum_{R\in\mathcal{Z}_3}2^{-w(R)}R , \label{eq:tau_decomp} 
    \end{equation}
    where $\mathcal{Z}_3$ is the set of $8$ Pauli strings on $3$ qubits composed of only identity and $Z$ terms. By the definition of the channel $N$, we have for a single $R$ and $n-1$-qubit Pauli $P$, and $a = 1-\lambda$, that
    \begin{equation}
        N(R\otimes P) = a^{w(R) + w(P) + \lceil (w(R) + w(P))/2\rceil}R\otimes P
    \end{equation}
    With $w(R)$ taking on values $\in \{0,1,2,3\}$ and leveraging the product structure of $\tilde \Delta_3$ with equation \eqref{eq:tau_decomp}, we obtain
    after elementary algebra
    \begin{equation}
        \tr\!\big(N(\tilde \Delta_3)^2\big) \leq \frac{1}{64}\left(1+\frac{3}{4}a^2 + \frac{3}{16}a^6 + \frac{1}{64}a^8 \right)\tr\!\big(N( \Delta_3)^2\big) \coloneqq \frac{\kappa_\lambda}{64}\tr\!\big(N( \Delta_3)^2\big)
    \end{equation}
    where the coefficients of the $a$ terms come from the factors of $2^{-w(R)}$, and $\Delta_3$ is defined on $n-1$ qubits. Using equation \eqref{eq:tr_N_bound}, we conclude that
    \begin{equation}
        \mathbb{E}_s[(L(s|u)-1)^2] \leq \frac{
        8\kappa_\lambda}{18^2}\frac{R_{n-1}(\eta)}{(2^{n-1}+1)^2(2^{n-1}+2)^2} . 
    \end{equation}
    Since $\kappa_\lambda$ is a $\Theta(1)$ constant for all values of $\lambda$, we conclude that 
    \begin{equation}
    T\geq \Omega\left(\min \left\{2^{n/2}, \frac{(2^{n-1}+1)^2(2^{n-1}+2)^2}{R_{n-1}(\eta)}\right\}\right)
    \end{equation}
    is necessary for any algorithm to succeed in the hypothesis distinguishing task with probability at least $2/3$.
\end{proof}
While this lower bound provides a worst-case bound on the broader task of estimating $\tr(O \rho^3)$, it explicitly uses the fact that $Z_1$ is a $1$-local observable. While it is possible to generalize Corollary \ref{cor:z_1_lb} to $k$-local observables with an appropriately chosen $k$-qubit $\tau$ and $n-k$-qubit $\rho$ sampled from the hard ensembles, the resulting exponents change from $n-1$ to $n-k$, weakening the bound for highly nonlocal observables which should not, \textit{a priori}, be easier to estimate than $Z_1$. To our knowledge, this feature is present in all prior results that lower bound the sample complexity of tasks which involve estimating higher-moment observables, such as quantum PCA \cite{Huang_adv_2022, liu2025exponentialseparationsquantumlearning}. Hence, it is also useful to establish that a fundamental exponential-in-$n$ lower bound, independent of locality, holds for a broad class of observables. For this, we begin with the following concentration bound.
\begin{lemma}
\label{lemma:third_moment_levy}
    Let $c,\epsilon > 0$ and let $P\neq \mathds{1}_n$ be any non-identity $n$-qubit Pauli. Define $O = \mathds{1}_n + cP$ and let $\rho$ be an $n$-qubit state sampled uniformly from either ensemble $\mathcal{E}_p$ or $\mathcal{E}_q$ as before. Then there exists a universal constant $C$ such that 
    \begin{equation}
    \textnormal{Pr}\left(|\tr(O\rho^3) - \tr(\rho^3)| > \epsilon \right) \leq 2\exp\left(-2^n C \epsilon^2/c^2\right) .
    \end{equation}
\end{lemma}
\begin{proof}
    Let $d=2^n$ and define the function
    \begin{equation}
        F(\psi_1, \psi_2, \psi_3) = \tr(P\rho^3)
    \end{equation}
    where $\rho$ is implicitly a function of the three states $\psi_1, \psi_2, \psi_3$ with coefficients defined by the ensemble from which $\rho$ is drawn. The domain of this function is $(S^{2d -1})^3$ with the product Euclidean metric
    \begin{equation}
        \textnormal{dist}(\Vec{\psi}, \Vec{\phi})^2 = \sum^3_{j=1}\|\psi_j -\phi_j\|_2^2 
    \end{equation}
    By definition, 
    \begin{equation}
        \tr(O\rho^3) - \tr(\rho^3) = cF .
    \end{equation}
    Clearly $\mathbb{E}[\rho^3]$ is unitarily invariant, so by Schur's lemma it is a multiple of the identity. This gives us $\mathbb{E}[F] = \tr(P\mathbb{E}[\rho^3]) = 0$, so $F$ is mean-0. To establish concentration of $F$, we need a bound on its Lipschitz constant. Since $P$ has operator norm at most $1$, we have
    \begin{align}
        |F(\psi)-F(\phi)| &\leq \|\rho_\psi^3 - \rho_\phi^3\|_1 \\ 
        &\leq 3\|\rho_\psi - \rho_\phi\|_1 \\ 
        &\leq 6\sum_{j=1}^3 a_j\|\psi_j - \phi_j\|_2 \\ 
        &\leq 6\left(\sum_j a_j^{2}\right)^{1/2}\left(\sum_j \|\psi_j -\phi_j \|_2^2\right)^{1/2} \\
        &= \frac{3}{\sqrt{2}} \textnormal{dist}(\Vec{\psi}, \Vec{\phi}) \ .
    \end{align}
    In the second line we use the fact that the $\rho$'s are density matrices, in the third we introduce the coefficients $a_j$ corresponding to either ensemble $p$ or $q$ and use the fact that $\|\ketbra{\psi_j}{\psi_j} - \ketbra{\phi_j}{\phi_j}\|_1 \leq 2 \|\psi_j - \phi_j\|_2$ for unit vectors, in the fourth step we invoke Cauchy Schwarz, and in the final we use that $\sum_j a_j^2 = 1/8$ for both $p$ and $q$. Once we have a bound on the Lipschitz constant $L$, we can invoke Levy's lemma to obtain
    \begin{equation}
    \textnormal{Pr}\left(|F - \mathbb{E}[F]| > \epsilon \right) \leq 2\exp\left(-C d \epsilon^2/L^2\right) 
    \end{equation}
    for some universal constant $C$. Absorbing $L$ into $C$ gives us the desired result.
\end{proof}
Because of this lemma, any algorithm that can estimate $\tr(O \rho^3)$ to within precision $\epsilon$ can also solve the original distinguishing problem in Definition \ref{def:moment_testing_prob}. As a result, the lower bound in Theorem \ref{thm:noisy_moment_lb} also applies to the task of cubic observable estimation — however, unlike in Corollary \ref{cor:z_1_lb}, this applies to highly nonlocal observables as well. Formally, the following Corollary is a directly result of Theorem \ref{thm:noisy_moment_lb} and Lemma \ref{lemma:third_moment_levy}.
\begin{corollary}
    Let $\rho$ be an $n$-qubit state. Then for any $k$ with $1\leq k \leq n$, there is a $k$-local observable $O$ which is $\Omega(1)$-far in operator norm from a multiple of the identity (that is, a nontrivial observable) such that any $\lambda$-noisy quantum algorithm with query access to $\DN[\rho]^{\otimes 3}$ using circuits of depth at least $2$ requires 
    \begin{equation}
        \Omega\left(\min \left\{2^{n/2}, \frac{(2^n+1)^2(2^n+2)^2}{R_n(\eta)}\right\}\right)
    \end{equation}
    to estimate $\tr(O\rho^3)$ to constant accuracy.
\end{corollary}

\subsection{Third moment estimation upper bound}

\begin{theorem} \label{theorem:moment_ub}
    Let $\lambda' \in [0,1)$, let $a = 1-\lambda'$, and suppose we are given query access to three noisy copies $\DN[\rho]^{\otimes 3}$ of an unknown $n$-qubit state $\rho$, with no further noise in the circuit. Then there is an algorithm which, using
    \begin{equation}
        O\left(\frac{1}{\epsilon^2}\left(\frac{1 - 6a^{-2} + 9a^{-4} + 12a^{-6}}{16}\right)^n \log\frac{1}{\delta}\right)
    \end{equation}
    experiments, outputs an estimate $\widehat m_3$ satisfying
    \begin{equation}
        \textnormal{Pr}\left[\left|\widehat m_3 - \tr(\rho^3)\right| > \epsilon\right] \leq \delta \ .
    \end{equation}
    If one only wishes to solve the distinguishing task in Definition \ref{def:moment_testing_prob}, there is a simpler algorithm, namely the ordinary three-copy cycle test applied directly to $\DN[\rho]$, whose sample complexity is
    \begin{equation}
        O\left(\left(\frac{16}{1+9a^2+6a^3}\right)^{2n}\log\frac{1}{\delta}\right) \ .
    \end{equation}
    for constant $\lambda' < 1$. The exact sample complexity for the $\lambda'\rightarrow 1$ regime is given in the proof, and diverges to $\infty$ as expected.
\end{theorem}
\begin{proof}
    Let $C = \pi_{123}$ denote the cyclic permutation on three copies of the $n$-qubit Hilbert space, and set
    \begin{equation}
        H = \frac{C + C^{-1}}{2} \ .
    \end{equation}
    For every state $\rho$, we have
    \begin{equation}
        \tr(H\rho^{\otimes 3}) = \tr(\rho^3) \ .
    \end{equation}
    Indeed, $\tr(C\rho^{\otimes 3}) = \tr(C^{-1}\rho^{\otimes 3}) = \tr(\rho^3)$.

    We first give an estimator for the generic quantity $\tr(\rho^3)$ from noisy copies. Let
    \begin{equation}
        \widetilde\rho = \DN[\rho]
    \end{equation}
    and let $\mathcal{I}_\lambda$ be the single-qubit (invertible) quantum channel which acts on Paulis as 
    \begin{equation}
        \mathcal{I}_\lambda[P] = \begin{cases}
        (1-\lambda)^{-1}P, \quad P\neq \mathds{1} \\ 
        P, \quad P = \mathds{1} 
        \end{cases} 
    \end{equation}
    Then define the corrected observable

    \begin{equation}
        H_{\lambda'} = \big(\DNINV\big)^{\otimes 3}[H] \ .
    \end{equation}
    It follows that
    \begin{equation}
        \tr(H_{\lambda'}\widetilde\rho^{\otimes 3}) = \tr(H\rho^{\otimes 3}) = \tr(\rho^3) \ .
    \end{equation}
    Thus it suffices to measure $H_{\lambda'}$ on three noisy copies.

    To make this explicit, work on one physical site and let $c$ denote the three-cycle on the three corresponding qubits. Then $C = c^{\otimes n}$ and
    \begin{equation}
        H_{\lambda'} = \frac{h_+^{\otimes n} + h_-^{\otimes n}}{2} \ ,
    \end{equation}
    where
    \begin{equation}
        h_+ = \big(\mathcal{I}_{\lambda'}\big)^{\otimes 3}[c], \quad h_- = \big(\mathcal{I}_{\lambda'}\big)^{\otimes 3}[c^{-1}] = h_+^\dagger \ .
    \end{equation}
    On one site we write
    \begin{equation}
        4c = \mathds{1} + A_{12}+A_{13}+A_{23} + iB, \quad 4c^{-1} = \mathds{1} + A_{12}+A_{13}+A_{23} - iB \ ,
    \end{equation}
    so with $a = 1-\lambda'$,
    \begin{equation}
        h_+ = \frac{\mathds{1} + a^{-2}(A_{12}+A_{13}+A_{23}) + ia^{-3}B}{4} \ .
    \end{equation}
    Using
    \begin{equation}
        A_{12}+A_{13}+A_{23} = 2(c+c^{-1}) - \mathds{1}, \quad iB = 2(c-c^{-1}),
    \end{equation}
    this becomes
    \begin{equation}
        h_+ = \frac{1-a^{-2}}{4}\mathds{1} + \frac{a^{-2}+a^{-3}}{2}c + \frac{a^{-2}-a^{-3}}{2}c^{-1} \ .
    \end{equation}
    Let $\omega = e^{2\pi i/3}$. Since $c$ has eigenvalues $1,\omega,\omega^2$, it follows that $h_+$ is diagonal in the same basis, with eigenvalues
    \begin{equation}
        \alpha_1 = \frac{1+3a^{-2}}{4}, \quad \alpha_\omega = \frac{1-3a^{-2}}{4} + i\frac{\sqrt{3}}{2}a^{-3}, \quad \alpha_{\omega^2} = \overline{\alpha_\omega} \ .
    \end{equation}
    Let $\Pi_1,\Pi_\omega,\Pi_{\omega^2}$ denote the corresponding local spectral projectors of $c$. Then $h_+$ and $h_-$ are simultaneously diagonal in the tensor-product basis $\Pi_{s_1}\otimes \cdots \otimes \Pi_{s_n}$, with $s_j \in \{1,\omega,\omega^2\}$. Therefore one experiment can be implemented as follows: on each site, measure the projective measurement $\{\Pi_1,\Pi_\omega,\Pi_{\omega^2}\}$ on the three corresponding noisy qubits; if the outcomes are $s_1,\dots,s_n$, output
    \begin{equation}
        X = \Re\left(\prod_{j=1}^n \alpha_{s_j}\right) \ .
    \end{equation}
    By construction,
    \begin{equation}
        \mathbb{E}[X] = \tr(H_{\lambda'}\widetilde\rho^{\otimes 3}) = \tr(\rho^3) \ ,
    \end{equation}
    so $X$ is an unbiased estimator.

    It remains to bound the range of $X$. We have
    \begin{equation}
        |\alpha_\omega|^2 = \frac{1 - 6a^{-2} + 9a^{-4} + 12a^{-6}}{16} \ .
    \end{equation}
    Moreover,
    \begin{equation}
        |\alpha_\omega|^2 - \alpha_1^2 = \frac{12a^{-6} - 12a^{-2}}{16} \geq 0 \ ,
    \end{equation}
    so $|\alpha_\omega| \geq \alpha_1$ for all $a \in (0,1]$. Hence every outcome obeys
    \begin{equation}
        |X| \leq |\alpha_\omega|^n = \left(\frac{1 - 6a^{-2} + 9a^{-4} + 12a^{-6}}{16}\right)^{n/2} \ .
    \end{equation}
    If we average $T$ independent copies of $X$, Hoeffding's inequality gives
    \begin{equation}
        \textnormal{Pr}\left[\left|\frac{1}{T}\sum_{t=1}^T X_t - \tr(\rho^3)\right| > \epsilon\right] \leq 2\exp\left(-\frac{T\epsilon^2}{2|\alpha_\omega|^{2n}}\right) \ .
    \end{equation}
    Thus it suffices to take
    \begin{equation}
        T = O\left(\frac{|\alpha_\omega|^{2n}}{\epsilon^2}\log\frac{1}{\delta}\right) = O\left(\frac{1}{\epsilon^2}\left(\frac{1 - 6a^{-2} + 9a^{-4} + 12a^{-6}}{16}\right)^n \log\frac{1}{\delta}\right) \ .
    \end{equation}
    For the promise problem of Definition \ref{def:moment_testing_prob}, the two hypotheses differ by a constant in $\tr(\rho^3)$, so taking $\epsilon$ to be a sufficiently small absolute constant yields the stated decision procedure.

    If one only wants the promise problem, a slightly better strategy is to omit the inverse channel and measure $H$ itself on three noisy copies. One experiment then outputs a random variable $Y \in [-1/2,1]$ with expectation
    \begin{equation}
        \mathbb{E}[Y] = \tr(\widetilde\rho^3) \ .
    \end{equation}
    Using the explicit form
    \begin{equation}
        \Delta_3 = \frac{2^{2n}(C+C^{-1}) - 2^{n+1}(F_{12}+F_{13}+F_{23}) + 4\mathds{1}_{3n}}{2^{3n}(2^n+1)(2^n+2)}
    \end{equation}
    and the convention $\sigma_p - \sigma_q = \Delta_3/18$, the expectation gap between the two hypotheses is
    \begin{equation}
        \left|\mathbb{E}_{\mathcal{E}_p}[Y] - \mathbb{E}_{\mathcal{E}_q}[Y]\right| = \frac{G_n(a)}{18\cdot 4^n(2^n+1)(2^n+2)} \ ,
    \end{equation}
    where
    \begin{equation}
        G_n(a) = (1+9a^2+6a^3)^n + (1+9a^2-6a^3)^n - 6(1+3a^2)^n + 4 \ .
    \end{equation}
    This follows from the local identities
    \begin{equation}
        2^{-3}\tr\!\big(c\,\mathcal D_{\lambda'}^{\otimes 3}[c^{-1}]\big) = \frac{1+9a^2+6a^3}{16}, \quad 2^{-3}\tr\!\big(c\,\mathcal D_{\lambda'}^{\otimes 3}[c]\big) = \frac{1+9a^2-6a^3}{16},
    \end{equation}
    and
    \begin{equation}
        2^{-3}\tr\!\big(c\,\mathcal D_{\lambda'}^{\otimes 3}[f_{ij}]\big) = \frac{1+3a^2}{8} \ ,
    \end{equation}
    with $f_{ij}$ the local swap on copies $i,j$. Since
    \begin{equation}
        1+9a^2+6a^3 > \max\{1+9a^2-6a^3,\,1+3a^2\}
    \end{equation}
    for every $a>0$, the first term dominates exponentially, so for fixed $\lambda'$ we have
    \begin{equation}
        G_n(a) = \Theta\big((1+9a^2+6a^3)^n\big) \ .
    \end{equation}
    Another application of Hoeffding therefore gives
    \begin{equation}
        T = O\left(\left(\frac{16}{1+9a^2+6a^3}\right)^{2n}\log\frac{1}{\delta}\right)
    \end{equation}
    experiments for the distinguishing task. This is slightly better than the generic deconvolved estimator because it uses the special promise structure of the hypothesis testing problem rather than recovering $\tr(\rho^3)$ for arbitrary $\rho$.
\end{proof}

\vspace{-2em}
\subsection{Exponential speedup}
The following is an immediate corollary upon combining the lower and upper bounds.
\begin{theorem}
    Let $N_{\rm inj}$ denote the sample complexity of the three-copy cycle test performed on injected copies of $\rho$ subject to depolarizing noise of strength $\lambda'$. Let $N_{\rm raw}$ denote the sample complexity of any three-copy algorithm for the distinguishing problem without injection, and thereby subject to at least two layers of depolarizing noise of strength $\lambda$. Then for any $\lambda' < 19\lambda / 12$, we have 
    \begin{equation}
        \frac{N_{\rm raw}}{N_{\rm inj}} \geq \Omega(\exp(\lambda n)) .
    \end{equation}
\end{theorem}
\begin{proof}
    Combining theorems \ref{thm:noisy_moment_lb} and \ref{theorem:moment_ub}, the ratio is
      \begin{equation}
        \frac{N_{\rm raw}}{N_{\rm inj}} \geq \Omega\left(\frac{G_n(1-\lambda')^2}{16^nR_n(1-\lambda)}\right) .
    \end{equation}
    Simplifying gives
    \begin{equation}
        \Omega\left(\left(\frac{(1 + 9(1-\lambda')^2 + 6(1-\lambda')^3)^2}{16(1 + 9(1-\lambda)^6 + 6(1-\lambda)^{10})}\right)^n\right) 
    \end{equation}
    This is exponential in $n$ whenever the numerator is greater than the denominator. Letting $x=1-\lambda, y=1-\lambda'$, it is easy to verify that this requirement imposes an upper constraint on $\lambda'$ for fixed $\lambda$ corresponding to the solution of
    \begin{equation}
        6y^3 + 9y^2 +1 = 4\sqrt{1+9x^6+6x^{10}}
    \end{equation}
    Working to first order in small $\lambda$, this gives 
    \begin{equation}
        \lambda'_{\rm max} < 19\lambda/12 ,
    \end{equation}
    and likewise to first order the ratio is lower bounded by $\exp(\Omega(\lambda n))$ when $\lambda' < \lambda_{\rm max}$. Naturally, a more detailed upper bound can be derived for larger $\lambda$, but this establishes an exponential sample complexity separation that permits the injection procedure to incur a constant fraction of additional error while growing the surface code compared to the case without injection.
\end{proof}

\section{Exponential Speedup in Classical Shadow Tomography}

We have seen that fault tolerance can recover many of the gains of multi-copy learning primitives that are otherwise exponentially degraded by the presence of noise. A second broad class of learning primitivies utilizes randomized preprocessing of quantum states to accurately learn a large array of properties. The most well-studied such learning algorithm is classical shadow tomography \cite{Huang_2020}. However, in physical experiments, such primitives would require applying many layers of complex, random unitary ensembles directly to fragile physical states. In this section, we show that classical shadow tomography also experiences an exponential speedup when input states are uploaded and sophisticated randomized preprocessing can be performed fault-tolerantly. More broadly, this indicates that the randomized toolbox of quantum learning theory regains its utility as a measurement primitive once fault tolerance is appropriately leveraged.

\subsection{Lower bound with noisy brickwork circuits}

\begin{lemma} \label{lemma:shadows_lower_bound}
    Consider a classical shadows protocol implemented with a unitary brickwork circuit of depth $d$ and local depolarizing noise of strength $\lambda$ between each circuit layer. Then the number of samples $N$ required to estimate $\tr(P_k \rho)$ to absolute error $\epsilon$, for arbitrary Pauli $P_k$ on $k$ contiguous sites and any quantum state $\rho$, scales as 
    \begin{equation}
        N \geq \Omega\left(\frac{1}{\omega_{\lambda, d}(P_k)(1-\lambda)^{2k}\epsilon^2}\right)
    \end{equation}
\end{lemma}

\begin{proof}
    Let $P_k\in \mathcal{P}_n$ be a Pauli operator acting as $X, Y$, or $Z$ on $k$ contiguous sites, and define the family of quantum states parameterized by $\theta$:
    \begin{equation}
        \rho_\theta = \frac{\mathds{1}+ \theta P_k}{2^n} \ .
    \end{equation}
    Note that $\tr(P_k\rho_\theta) = \theta$. Hence, any algorithm capable of estimating $\tr(P_k\rho)$ to within accuracy $\epsilon$ for arbitrary quantum states $\rho$ is capable of distinguishing the states $\rho_{-\epsilon}, \rho_{\epsilon}$ with high probability. Our goal is to show that to accomplish this distinguishing task, a classical shadow protocol implemented with $\lambda$-noisy brickwork circuits of depth $d$ will require at least $\Omega(\epsilon^{-2}\omega^{-1}_{\lambda, d}(P_k))$ samples of the unknown state, implying a lower bound on the sample complexity of learning $\tr(P_k\rho)$ for general $\rho$. 

    As in e.g. \cite{Hu_2025}, the classical shadow estimator for $\theta$ is given by 
    \begin{equation}
        \hat{\theta} = \frac{1}{\omega_{\lambda, d}(P_k)}\mathbb{E}_{\hat{\sigma}}[\tr(P_k\hat\sigma)]
    \end{equation}
    where $\hat{\sigma}$ is the classical shadow generated by our fixed unitary ensemble and computational basis measurement. Moreover, \cite{Hu_2025} demonstrates that under this $\lambda$-noisy, depth-$d$ brickwork ensemble, the Pauli shadow weight $\omega_{\lambda, d}(P_k)$ is given by 
    \begin{equation}
        \omega_{\lambda, d}(P_k) = \mathbb{E}\left[\frac{1}{3^{l_d}}\exp\left(-\lambda\sum_{i=1}^d l_i\right)\right] \ ,
    \end{equation}

    where $Z = (l_1, l_2, ..., l_d)$ is a random variable where $l_i$ denotes the weight of $P_k$ after applying $i$ layers of the unitary ensemble. Conceptually, $\omega_{\lambda, d}$ corresponds to how closely the eventual measurement basis aligns with the most ``informative" basis in which to measure $P_k$. Since $P_k$ is either $\sigma_X, \sigma_Y,$ or $\sigma_Z$ in each of $k$ relevant sites, only measurements along the correct direction at the correct site provide information relevant to the desired expectation value. This contributes the $1/3^{l_d}$ factor. Depolarizing noise sends non-identity Paulis to the identity with probability $\lambda$ after each layer, which results in the exponential damping term. It will be useful to define
    \begin{equation}
        q(Z) = 1/3^{l_d}, \quad s(Z) = \exp\left(-\lambda\sum_{i=1}^d l_i\right) \ .
    \end{equation}
    Now noting that the random variable 
    \begin{equation}
        Y = \tr(P_k\hat\sigma)
    \end{equation}
    can only take on values in $\{-1, 0, 1\}$, we see that for the state $\rho_\theta$,
    \begin{equation}
    \label{eq:shadows_functionals}
        \Pr{Y = 0| Z} = 1 - q(Z), \quad \Pr{Y = \pm 1| Z} = \frac{q(Z)(1\pm (1-\lambda)^{k}s(Z)\theta)}{2}
    \end{equation}
    Here, the additional $(1-\lambda)^{-k}$ factor comes from the fact that the shadow-weight noise model applies the first layer of noise after the first gate, so we must apply a noise channel directly to $\rho_\theta$ before computing our estimator. Our hypothesis testing task is to decide between having received $\rho_{-\epsilon}, \rho_{\epsilon}$ after collecting measurement outcomes. Let $P_-, P_+$ denote the outcome distributions corresponding to the shadow statistic $Y$ under ground truth $\rho_{-\epsilon}, \rho_{\epsilon}$, respectively. To obtain a lower bound on the sample complexity of the hypothesis testing task, our goal is to bound the distinguishability of $P_+$ and $P_-$. Morever, note that
    \begin{equation}
    \label{eq:shadows_chain_rule}
        d_{KL}(P_+\|P_-) \leq \mathbb{E}_Z\big[d_{KL}(P_+(Y|Z)\|P_-(Y|Z))\big] \ .
    \end{equation}
    due to the chain rule of KL divergence and the fact that $Z$ is independent of the ground truth. By \eqref{eq:shadows_functionals} and upon simplifying the definition of KL divergence, we obtain
    \begin{equation}
        d_{KL}(P_+(Y|Z)\|P_-(Y|Z)) = \epsilon (1-\lambda)^{k}q(Z)s(Z)\log \frac{1+ (1-\lambda)^{k}s(Z)\epsilon}{1-(1-\lambda)^{k}s(Z)\epsilon}
    \end{equation}
    For $s(Z)\epsilon < 1/2$ and noting $s(Z), (1-\lambda)^{k} \leq 1$, we have 
    \begin{equation}
    \label{eq:shadows_simp}
       \epsilon (1-\lambda)^{-k}q(Z)s(Z)\log \frac{1+ (1-\lambda)^{k}s(Z)\epsilon}{1-(1-\lambda)^{k}s(Z)\epsilon} \leq 4(1-\lambda)^{2k}q(Z)s(Z)\epsilon^2
    \end{equation}
    where we use
    \begin{equation}
        \log \frac{1+x}{1-x} \leq \frac{2x}{1-x} \leq 4x
    \end{equation}
    for $x \leq 1/2$, and we also use $s(Z)^2 \leq s(Z)$. Since $\omega_{\lambda, d}(P_k) = \mathbb{E}_Z[q(Z)s(Z)]$, we combine \eqref{eq:shadows_chain_rule} and \eqref{eq:shadows_simp} to obtain
    \begin{equation}
        d_{KL}(P_+\|P_-) \leq 4(1-\lambda)^{2k}\epsilon^2 \omega_{\lambda, d}(P_k) \ .
    \end{equation} 
    Collecting $N$ shots and applying Pinsker's inequality, we find
    \begin{equation}
        d_{\rm TV}(P_+, P_-) \leq \epsilon(1-\lambda)^{k}\sqrt{2N\omega_{\lambda, d}(P_k)} \ .
    \end{equation}
    By Lemma \ref{lemma:le_cam}, we obtain that to successfully solve the hypothesis testing task with probability at least $2/3$, we require
    \begin{equation}
        N \geq \Omega\left(\frac{1}{\omega_{\lambda, d}(P_k)(1-\lambda)^{2k}\epsilon^2}\right) \ .
    \end{equation}
\end{proof}
Next, we obtain a simple upper bound for $\omega_{\lambda, d}(P_k)$, allowing us to obtain a quantitative lower bound on the $\lambda$-noisy sample complexity. 
\begin{lemma} \label{lemma:shadow_norm_ub}
    \begin{equation}
        \sup_{d\geq 0} \omega_{\lambda, d}(P_k) \leq \max\left\{3^{-k}, e^{-\lambda \lfloor (k-1)/2\rfloor}\sup_{d\geq 0}\omega_{\lambda=0, d}(P_k)\right\}
    \end{equation}
\end{lemma}
\begin{proof}
    From \cite{Hu_2025}, we have 
    \begin{equation}
        \omega_{\lambda, d}(P_k) = \mathbb{E}\left[\frac{1}{3^{l_i}}\exp\left(-\lambda\sum_{i=1}^d l_i\right)\right] \ .
    \end{equation}
    Because all $l_i \geq 0$, the exponential term is bounded above by $\exp(-\lambda l_1)$. We note that for $\Theta(1)$ noise rate $\lambda$, adding depth tends to degrade performance, and the optimal depth approaches 0 as $k$ grows large. Hence, in the asymptotic regime of interest, bounding the sum in the exponential by its first term becomes relatively tight. 

    As established in \cite{Hu_2025}, in the brickwork circuit model, a random two-qubit unitary acting on a pair of sites which correspond to non-identity Pauli terms can only map the two-site operator to an operator with at least one non-identity Pauli. When we act on $P_k$ with pairwise two-qubit unitaries, this property implies that regardless of the realization of random unitaries sampled from our ensemble, the resulting operator must contain at least $\lfloor (k-1)/2\rfloor$ non-identity terms. For example, this lower bound can be achieved in the case that $k$ is even and the pairwise unitaries are misaligned with the sites corresponding to non-identity Paulis. With this, we obtain
    \begin{equation}
        \omega_{\lambda, d}(P_k) \leq \mathbb{E}\left[\frac{1}{3^{l_d}}\exp\left(-\lambda l_1\right)\right] \leq e^{-\lambda\lfloor (k-1)/2\rfloor}\mathbb{E}\left[\frac{1}{3^{l_d}}\right]
    \end{equation}
    Now we observe that the remaining expectation is exactly $\omega_{\lambda = 0, d}(P_k)$, the Pauli shadow weight under a \textit{noiseless} shadow circuit. Maximizing over $d\geq 1$, we have 
    \begin{equation}
        \sup_{d\geq 1} \omega_{\lambda, d}(P_k) \leq e^{-\lambda\lfloor (k-1)/2\rfloor}\sup_{d\geq 1}\omega_{\lambda=0, d}(P_k)
    \end{equation}
    At $d=0$, $\omega_{\lambda, d=0}(P_k) = 3^{-k}$ (independent of $\lambda$ as no noisy layers are executed), which completes the proof. 
\end{proof}
Now, we are able to obtain lower bounds on the sample complexity of brickwork shallow shadows under a $\lambda$-noisy circuit model. 
\subsection{Upper bound with injection and exponential speedup}
To compare this lower bound against the injection-enhanced procedure, we prove the following Lemma.

\begin{lemma} \label{lemma:shadows_ub}
Let $P_k$ be a Pauli operator supported on $k$ contiguous qubits. Consider the protocol that first injects the input state $\rho$ into a logical register via
\begin{equation}
\rho_{\mathrm{log}} = U_{\mathrm{enc}} \, \mathcal{D}_\eta^{\otimes n}(\rho)\, U_{\mathrm{enc}}^\dagger,
\end{equation}
and then applies a fault-tolerant implementation of the noiseless brickwork shallow-shadows ensemble on the logical qubits at depth $d = \Theta(\log k)$. Then there exists an estimator of $\mathrm{Tr}(P_k \rho)$ which, for constant success probability, achieves absolute error $\epsilon$ using
\begin{equation}
N = O\!\left(\frac{3^{3k/4}\,\mathrm{poly}(k)}{(1-\eta)^{2k}\,\epsilon^2}\right)
\end{equation}
copies of $\rho$. Achieving failure probability at most $\delta$ incurs an additional multiplicative factor of $\log(1/\delta)$.
\end{lemma}

\begin{proof}
Let $\overline{P}_k = U_{\mathrm{enc}} P_k U_{\mathrm{enc}}^\dagger$ denote the corresponding logical Pauli observable. Since the single-qubit depolarizing channel $\mathcal{D}_\eta$ acts diagonally in the Pauli basis and multiplies each non-identity Pauli by a factor $(1-\eta)$, we have
\begin{equation}
\mathrm{Tr}(\overline{P}_k \rho_{\mathrm{log}}) = \mathrm{Tr}\!\left(P_k \mathcal{D}_\eta^{\otimes n}(\rho)\right) = (1-\eta)^k \mathrm{Tr}(P_k \rho).
\end{equation}
Define $\mu = \mathrm{Tr}(P_k \rho)$ and $\mu_{\mathrm{log}} = \mathrm{Tr}(\overline{P}_k \rho_{\mathrm{log}})$. Then
$\mu_{\mathrm{log}} = (1-\eta)^k \mu$. Thus, estimating $\mu$ to additive error $\epsilon$ is equivalent to estimating $\mu_{\mathrm{log}}$ to additive error
$\epsilon_{\mathrm{log}} = (1-\eta)^k \epsilon$.

After injection, the protocol proceeds fault-tolerantly, so we may analyze the logical circuit using the noiseless shallow-shadows framework. For the brickwork ensemble at depth $d$, the squared shadow norm of a $k$-local Pauli satisfies
\begin{equation}
\log_3 \| P_k(d)\|_{\mathrm{sh}}^2 
\le (k+d)\left(\frac{3}{4} + \frac{e^{-\gamma d}}{d^{3/2}}\right),
\end{equation}
for $\exp(-\gamma) = (4/5)^2$, and the optimal depth is $d^* = \Theta(\log k)$, from Theorem 1 of \cite{Hu_2025}. Substituting $d = d^*$ yields
\begin{equation}
\| P_k(d^*)\|_{\mathrm{sh}}^2 \leq 3^{3k/4} \, \mathrm{poly}(k).
\end{equation}
The same bound applies to $\overline{P}_k$.

By standard concentration bounds for the classical shadows estimator, estimating $\mu_{\mathrm{log}}$ to accuracy $\epsilon_{\mathrm{log}}$ with constant success probability requires
\begin{equation}
N = O\!\left(\frac{\|\overline{P}_k(d^*)\|_{\mathrm{sh}}^2}{\epsilon_{\mathrm{log}}^2}\right).
\end{equation}
Substituting the bound on the shadow norm and $\epsilon_{\mathrm{log}} = (1-\eta)^k \epsilon$ gives
\begin{equation}
N = O\!\left(\frac{3^{3k/4}\,\mathrm{poly}(k)}{(1-\eta)^{2k}\,\epsilon^2}\right).
\end{equation}
Finally, rescaling the estimator by $(1-\eta)^{-k}$ yields an estimate of $\mu = \mathrm{Tr}(P_k \rho)$ with additive error at most $\epsilon$. This completes the proof.
\end{proof}
Once again, the following is an immediate corollary upon combining the lower and upper bounds.
\begin{theorem}
    Let $N_{\rm inj}$ denote the sample complexity of the shallow shadows protocol performed on injected copies of $\rho$ subject to depolarizing noise of strength $\lambda'$. Let $N_{\rm raw}$ denote the sample complexity of any brickwork-circuit strategy for estimating a $k$-local Pauli on contiguous sites without injection. Then for small $\lambda$, for any $\lambda' < 5\lambda / 4$, we have 
    \begin{equation}
        \frac{N_{\rm raw}}{N_{\rm inj}} \geq \exp (\Omega(\lambda k)) .
    \end{equation}
    More generally, this is true whenever 
    \begin{equation}
    \lambda' <  \begin{cases}
        1-(1-\lambda )e^{-\lambda/4}, \quad \lambda \leq \frac{1}{2}\log 3 \\
        1-(1-\lambda )3^{-1/8}, \quad \lambda >\frac{1}{2}\log 3
    \end{cases}  \ .     
    \end{equation}
\end{theorem}
\begin{proof}
    Because the squared shadow norm is the reciprocal of the Pauli shadow weight, the numerator $3^{3k/4}\textnormal{poly}(k)$ in the upper bound from Lemma \ref{lemma:shadows_ub} is exactly the reciprocal of $\omega^*(P_k) \coloneqq \sup_{d\geq 1}\omega_{\lambda=0, d}(P_k)$. Hence, 
    \begin{equation}
        \frac{N_{\rm raw}}{N_{\rm inj}} = \Omega\left(\left(\frac{1-\lambda'}{1-\lambda}\right)^{2k} \frac{\omega^*(P_k)}{\max \{3^{-k}, e^{-\lambda \lfloor (k-1)/2\rfloor}\omega^*(P_k)\}}\right) = \Omega\left(\left(\frac{1-\lambda'}{1-\lambda}\right)^{2k} \min\left\{3^{k}\omega^*(P_k), e^{\lambda \lfloor (k-1)/2\rfloor}\}\right\}\right)
    \end{equation}
    using Lemmas \ref{lemma:shadows_lower_bound} and \ref{lemma:shadow_norm_ub} for the denominator and Lemma \ref{lemma:shadows_ub} for the numerator. One can check that rewriting in base $e$ gives
\begin{equation}
     \frac{N_{\rm raw}}{N_{\rm inj}} \geq \Omega\left(\exp\left[\left(2\log\frac{1-\lambda'}{1-\lambda}+\min\left\{\frac{\log 3}{4}, \frac{\lambda}{2}\right\} \right)k\right]\right)
\end{equation}
so the prefactor of $k$ in the exponent must be a positive constant for an exponential separation to persist. Enforcing this condition yields the result, and for small $\lambda$ we have $1-(1-\lambda )e^{-\lambda/4} \approx 5\lambda/4$, so in the small-noise regime we retain a provable exponential speedup whenever injection incurs less than $\approx 20 \%$ noise overhead. We note that this does not rule out the potential for a quantitatively exponential separation at larger injection noise overheads, but simply establishes a rigorous condition under which a provable exponential speedup persists.
\end{proof}

\section{Description of Numerical Experiments}
In this section we describe the numerical experiments conducted in Figure \ref{fig:astronomy} at a high level.

\begin{center}
\textit{Input states}    
\end{center}

The optical imaging pipeline of Ref. \cite{mokeev2025enhancingopticalimagingquantum} can be understood as follows. The quantum state of the incoming photons is modeled as 
\begin{equation}
    \rho = b\ketbra{\psi_{\rm star}}{\psi_{\rm star}} + (1-b) \ketbra{\psi_{\rm planet}}{\psi_{\rm planet}}
\end{equation}
$b \in (0.5, 1)$ is the brightness parameter, and it may be the case that $\bra{\psi_{\rm star}}\psi_{\rm planet}\rangle \neq 0$. One assumes $b$ is given by other astrophysical measurements, and may be close to $1$ when the star is much brighter than the planet. The goal is to measure $\bra{\psi_{\rm planet}}O\ket{\psi_{\rm planet}}$; this is challenging due to the mixed state provided. The specific form of the two pure states is obtained from plane-wave solutions for the incoming light being reflected off the telescope, which in this case may be an array of atomic cavities; this structure is utilized in our numerical simulations, and the overlap between states is controlled by a translation of the plane-wave solutions in the plane of the telescope. The two states in our hypothesis testing problem correspond to switching $b \leftrightarrow (1-b)$ while fixing the spatial separation between the two pure states in the $x$-direction.  

\begin{center}
\textit{Observable estimation}    
\end{center}

The key observation of \cite{mokeev2025enhancingopticalimagingquantum} is that the quantity $\bra{\psi_{\rm planet}}O\ket{\psi_{\rm planet}}$ can be written in terms of expectations with respect to the \textit{eigenvectors} of $\rho$; namely, upon rewriting $\rho = r\ketbra{V_1}{V_1} + (1-r)\ketbra{V_2}{V_2}$ with $\bra{V_1}V_2\rangle = 0$ and $r, (1-r)$ the two eigenvalues, one has
\begin{equation}
    \bra{\psi_{\rm planet}}O\ket{\psi_{\rm planet}} = |c_1|^2 \tr(\rho \ketbra{V_1}{V_1}) + |c_2|^2 \tr(\rho \ketbra{V_2}{V_2}) + 2\Re(c_1^\star c_2\bra{V_1}O|V_2\rangle), 
\end{equation}
with $c_1, c_2$ calculable from known values. The task then becomes estimating the two diagonal expectations and the one off-diagonal expectation. 

\begin{center}
\textit{Filtering pipeline}    
\end{center}

To estimate these three quantities, Ref. \cite{mokeev2025enhancingopticalimagingquantum} provides a quantum algorithm that processes copies of $\rho$ into approximate copies of the states $\ket{V_1}, \ket{V_2}$. This is done with two nested primitives. First, sequential copies of $\rho$ are captured and stored, and the quantum principal component analysis algorithm is used to prepare an effective unitary that approximates $\exp(\pm i x \rho)$ for a chosen scalar parameter $x$. This unitary is implemented in a controlled fashion on another, coherently-evolving copy of $\rho$, interleaved with single-qubit rotations on the ancilla system. These single-qubit rotations are chosen to implement a quantum signal processing (QSP) subroutine such that the effective polynomial transform is close to a Heaviside function, with step between $xr$ and $x(1-r)$. If implemented appropriately, the resulting map that interleaves QSP rotations with queries of the approximate $\exp(\pm i x \rho)$ gate approximately implements the transform
\begin{equation}
    \ketbra{0}{0}_{\rm anc} \otimes \rho \,\,\longmapsto\,\, r\ketbra{0}{0}_{\rm anc} \otimes \ketbra{V_1}{V_1} +  (1-r)\ketbra{1}{1}_{\rm anc} \otimes\ketbra{V_2}{V_2} \ .
\end{equation}
Measuring the ancilla qubits produces labeled approximate samples of the eigenvectors; in practice, these can be mixed states that are close to pure. There are several sources of approximation error, even in the noiseless setting. First, the quality of approximation of each density-matrix exponentiation (DME) depends on the choice of $x$ and the number of rounds $M$ used to create each oracle query. As $x$ decreases, smaller $M$ will suffice to achieve a fixed target accuracy. However, the interleaved QSP protocol will allow either eigenvector to leak into the other branch with probability that increases as $x$ shrinks. This establishes a tradeoff for choosing the parameter $x$, and is explicitly considered in the asymptotic setting in \cite{mokeev2025enhancingopticalimagingquantum}; the takeaway is that $x$ must scale proportionally to the quantity $(\epsilon^{1/2}/\mathrm{QSP \ polynomial \ degree})$, where $\epsilon$ is the final admissible estimation error for the observable. Our code numerically optimizes for $x$ in tandem with minimizing the QSP degree, and the optimal solutions quantitatively agree with this asymptotic scaling. The QSP polynomial optimization is performed using the QRISP library. 

Once labeled copies of the eigenvectors have been stored, one can solve for the two diagonal quantities $\tr(\rho \ketbra{V_1}{V_1})$ and $\tr(\rho\ketbra{V_2}{V_2})$ by directly performing measurements of the  labeled eigenstates. For the off-diagonal quantity, one needs a more complex measurement based on a block-encoded variant of the SWAP test; see the Supplementary Material of \cite{mokeev2025enhancingopticalimagingquantum}. Our numerics directly estimate the means and variances of each of these three quantities, treating them as random variables. In computing sample complexity, we account for the sampling overhead of postselecting on the ancilla, as well as the number of copies utilized in each round of the block-encoding procedure.

\begin{center}
\textit{Sample complexity and success probability}    
\end{center}
We choose $O$ to be an observable that measures the horizontal spatial location of the source. In the raw noise model, a layer of depolarizing noise is applied to the memory register after each of the $M$ interactions in the DME step. However, at each error value, the QSP pipeline is optimized based on error in the \textit{final observable estimate}, a deliberate overfitting choice that physically corresponds to the raw learner being given prescient knowledge of the estimation bias caused by noise, which is our proxy for the $N_{\rm raw}$ setup in this work. Meanwhile, in the uploaded model, depolarizing noise is applied whenever a new state is loaded into the PCA register, and the QSP angles are fitted to satisfy a trace-distance condition on the output state. This configuration is fixed across all noise points and will yield the $N_{\rm inj}$ sample complexity.

Using the outputs of these circuits, we compute the mean and variance of the composite estimator. Given $N$ total repetitions, we assume a Gaussian model for the realized values of the estimator, and damp the standard deviation by $1/\sqrt{N}$. The resulting Gaussian density is used to compute hypothesis-testing success probability using standard z-score statistics. Sample complexity corresponding to a fixed success probability can be backed out directly from the z-score.

\putbib
\end{bibunit}

\end{document}